\documentclass{article}
\usepackage[utf8]{inputenc}
\usepackage[top=3cm,bottom=3cm,left=3cm,right=3cm,marginparwidth=1.75cm]{geometry}

\usepackage{amstext}
\usepackage{amsfonts}
\usepackage{amssymb}
\usepackage{xcolor,bbm}
\usepackage{mathtools}
\usepackage{bm}
\usepackage{algorithm}
\usepackage{algpseudocode}
\usepackage{mathrsfs}
\usepackage{amsthm}
\usepackage{amsmath}
\usepackage{thmtools,thm-restate}
\usepackage{float}
\usepackage{verbatim}
\usepackage{graphicx}
\usepackage{caption}
\usepackage{booktabs}
\usepackage{subcaption}
\usepackage{multirow}
\usepackage{natbib}
\usepackage{paralist}
\RequirePackage[colorlinks,citecolor={blue!60!black},urlcolor={blue!70!black},linkcolor={red!60!black},breaklinks,hypertexnames=false]{hyperref}
\usepackage{cleveref}
\usepackage{enumitem}
\newtheorem{theorem}{Theorem}
\newtheorem{lemma}{Lemma}
\newtheorem{proposition}{Proposition}
\newtheorem{cor}{Corollary}
\newtheorem{assumption}{Assumption}
\newtheorem{remark}{Remark}
\newtheorem{definition}{Definition}
\renewcommand{\proofname}{\textbf{Proof.}}

\DeclareMathOperator*\argmax{argmax}
\DeclareMathOperator*\argmin{argmin}
\newcommand{\norm}[1]{\left\lVert#1\right\rVert}
\newcommand{\E}{\mathbb{E}}
\newcommand{\p}{\mathbb{P}}
\newcommand{\R}{\mathbb{R}}

\newcommand{\red}[1]{{\color{red} #1}}

\newcommand{\code}[1]{\text{\ttfamily #1}}
\newcommand{\bx}{\bm{x}}
\newcommand{\bz}{\bm{z}}

\newcommand{\B}{\boldsymbol}

\newcommand{\dx}{~d\bm{x}}
\newcommand{\bphi}{\bm{\phi}}
\newcommand{\cef}{\mathrm{cef}}

\newcommand{\dom}[1]{\mathrm{dom}(#1)}

\crefalias{prop}{proposition}
\expandafter\let\expandafter\oldproof\csname\string\proof\endcsname
\let\oldendproof\endproof
\renewenvironment{proof}[1][\proofname]{%
  \oldproof[\textit{\textbf{#1}}]%
}{\oldendproof}

\begin{document}
\title{A new computational framework for log-concave density estimation}
\author{Wenyu Chen$^\ast$, Rahul Mazumder$^\ast$ and Richard J. Samworth$^\dagger$\\$^\ast$Massachusetts Institute of Technology and $^\dagger$University of Cambridge}
\date{January 27, 2023}
\maketitle
\begin{abstract}
In Statistics, log-concave density estimation is a central problem within the field of nonparametric inference under shape constraints.  Despite great progress in recent years on the statistical theory of the canonical estimator, namely the log-concave maximum likelihood estimator, adoption of this method has been hampered by the complexities of the non-smooth convex optimization problem that underpins its computation.  We provide enhanced understanding of the structural properties of this optimization problem, which motivates the proposal of new algorithms, based on both randomized and Nesterov smoothing, combined with an appropriate integral discretization of increasing accuracy.  We prove that these methods enjoy, both with high probability and in expectation, a convergence rate of order $1/T$ up to logarithmic factors on the objective function scale, where $T$ denotes the number of iterations.  The benefits of our new computational framework are demonstrated on both synthetic and real data, and our implementation is available in a github repository \texttt{LogConcComp} (Log-Concave Computation).
\end{abstract}

\section{Introduction}

In Statistics, the field of nonparametric inference under shape constraints dates back at least to \cite{grenander1956theory}, who studied the nonparametric maximum likelihood estimator of a decreasing density on the non-negative half line.  But it is really over the last decade or so that researchers have begun to realize its full potential for addressing key contemporary data challenges such as (multivariate) density estimation and regression.  The initial allure is the flexibility of a nonparametric model, combined with estimation methods that can often avoid the need for tuning parameter selection, which can often be troublesome for other nonparametric techniques such as those based on smoothing.  Intensive research efforts over recent years have revealed further great attractions: for instance, these procedures frequently attain optimal rates of convergence over relevant function classes.  Moreover, it is now known that shape-constrained procedures can possess intriguing adaptation properties, in the sense that they can estimate particular subclasses of functions at faster rates, even (nearly) as well as the best one could do if one were told in advance that the function belonged to this subclass.  

Typically, however, the implementation of shape-constrained estimation techniques requires the solution of an optimization problem, and, despite some progress, there are several cases where computation remains a bottleneck and hampers the adoption of these methods by practitioners.  In this work, we focus on the problem of log-concave density estimation, which has become arguably the central challenge in the field because the class of log-concave densities enjoys stability properties under marginalization, conditioning, convolution and linear transformations that make it a very natural infinite-dimensional generalization of the class of Gaussian densities \citep{samworth2018recent}.

The univariate log-concave density estimation problem was first studied in \cite{walther2002detecting}, and fast algorithms for the computation of the log-concave maximum likelihood estimator (MLE) in one dimension are now available through the \texttt{R} packages \texttt{logcondens} \citep{dumbgen2011logcondens} and \texttt{cnmlcd} \citep{liu2018fast}.  \cite{cule2010maximum} introduced and studied the multivariate log-concave maximum likelihood estimator, but their algorithm, which is described below and implemented in the \texttt{R} package \texttt{LogConcDEAD} \citep{JSSv029i02}, is slow; for instance, \cite{cule2010maximum} report a running time of 50 seconds for computing the bivariate log-concave MLE with 500 observations, and 224 minutes for computing the log-concave MLE in four dimensions with 2{,}000 observations.  An alternative, interior point method for a suitable approximation was proposed by \cite{koenker2010quasi}.  Recent progress on theoretical aspects of the computational problem in the computer science community includes \cite{axelrod2019polynomial}, who proved that there exists a polynomial time algorithm for computing the log-concave maximum likelihood estimator.  We are unaware of any attempt to implement this algorithm.  \cite{rathke2019fast} compute an approximation to the log-concave MLE by considering $-\log p$ as a piecewise affine maximum function, using the log-sum-exp operator to approximate the non-smooth operator, a Riemann sum to compute the integral and its gradient, and obtain a solution via L-BFGS.  This reformulation means that the problem is no longer convex.

\sloppy To describe the problem more formally, let $\mathcal{C}_d$ denote the class of proper, convex lower-semicontinuous functions $\varphi:\R^d \rightarrow (-\infty,\infty]$ that are coercive in the sense that $\varphi(\bx) \rightarrow \infty$ as $\|\bx\| \rightarrow \infty$.  The class of upper semi-continuous log-concave densities on $\R^d$ is denoted as 
\[
  \mathcal{P}_d := \biggl\{p:\mathbb{R}^d \rightarrow [0,\infty): p = e^{-\varphi} \text{ for some } \varphi \in \mathcal{C}_d, \int_{\R^d} p = 1\biggr\}.
\]
Given $\bx_1,\ldots,\bx_n \in \R^d$, \citet[Theorem~1]{cule2010maximum} proved that whenever the convex hull $C_n$ of $\bx_1,\ldots,\bx_n$ is $d$-dimensional, there exists a unique
\begin{equation}
  \label{Eq:MLE}
  \hat{p}_n \in \argmax_{p \in \mathcal{P}_d} \frac{1}{n}\sum_{i=1}^n \log p(\bx_i).
\end{equation}
If $\bx_1,\ldots,\bx_n$ are regarded as realizations of independent and identically distributed random vectors on $\R^d$, then the objective function in~\eqref{Eq:MLE} is a scaled version of the log-likelihood function, so $\hat{p}_n$ is called the {\emph{log-concave MLE}}.  The existence and uniqueness of this estimator is not obvious, because the infinite-dimensional class~$\mathcal{P}_d$ is non-convex, and even the class of negative log-densities $\bigl\{\varphi \in \mathcal{C}_d:\int_{\R^d} e^{-\varphi} = 1\bigr\}$ is non-convex.  In fact, the estimator belongs to a finite-dimensional subclass; more precisely, for a vector $\bphi = (\phi_1,\ldots,\phi_n) \in \R^n$, define $\mathrm{cef}[\bphi] \in \mathcal{C}_d$ to be the (pointwise) largest function with 
\[
\mathrm{cef}[\bphi](\bx_i) \leq \phi_i
\]
for $i=1,\ldots,n$.  \cite{cule2010maximum} proved that $\hat{p}_n = e^{-\mathrm{cef}[\bphi^*]}$ for some $\bphi^* \in \R^n$, and refer to the function $-\mathrm{cef}[\bphi^*]$ as a `tent function'; see the illustration in Figure~\ref{Fig:Tent}.  \cite{cule2010maximum} further defined the non-smooth, convex objective function $f:\R^n \rightarrow \R$ by
\begin{equation}
  \label{Eq:sigma}
  f(\bphi) \equiv f(\phi_1,\ldots,\phi_n):= \frac{1}{n}\sum_{i=1}^n \phi_i + \int_{C_n} \exp\bigl\{-\mathrm{cef}[\bphi](x)\bigr\} \, dx,
\end{equation}
and proved that $\bphi^* = \argmin_{\bphi \in \R^n} f(\bphi)$.

\begin{figure}[ht]
  \begin{center}
    \includegraphics[width=0.4\textwidth]{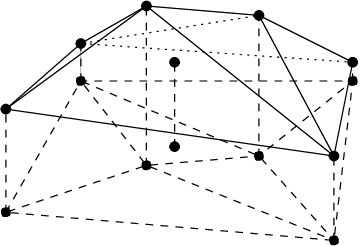}
  \end{center}
  \caption{\label{Fig:Tent}An illustration of a tent function, taken from \cite{cule2010maximum}.}
\end{figure}

The two main challenges in optimizing the objective function $f$ in~\eqref{Eq:sigma} are that the value and subgradient of the integral term are hard to evaluate, and that it is non-smooth, so vanilla subgradient methods lead to a slow rate of convergence.  To address the first issue, \cite{cule2010maximum} computed the exact integral and its subgradient using the \texttt{qhull} algorithm \citep{barber1993quickhull} to obtain a triangulation of the convex hull of the data, evaluating the function value and subgradient over each simplex in the triangulation. However, in the worst case, the triangulation can have $O(n^{d/2})$ simplices \citep{mcmullen1970maximum}. The non-smoothness is handled via Shor's $r$-algorithm \cite[Chapter~3]{shor1985minimization}, as implemented by \cite{kappel2000implementation}.

In Section~\ref{sec:Setup}, we characterize the subdifferential of the objective function in terms of the solution of a linear program (LP), and show that the solution lies in a known, compact subset of $\R^n$.   This understanding allows us to introduce our new computational framework for log-concave density estimation in Section~\ref{sec:computation}, based on an accelerated version of a dual averaging approach \citep{nesterov2009primal}.  This relies on smoothing the objective function, and encompasses two popular strategies, namely Nesterov smoothing \citep{nesterov2005smooth} and randomized smoothing \citep{lakshmanan2008decentralized,yousefian2012stochastic,duchi2012randomized}, as special cases. A further feature of our algorithm is the construction of approximations to gradients of our smoothed objective, and this in turn requires an approximation to the integral in~\eqref{Eq:sigma}.  While a direct application of the theory of \cite{duchi2012randomized} would yield a rate of convergence for the objective function of order $n^{1/4}/T + 1/\sqrt{T}$ after~$T$ iterations, we show in Section~\ref{sec:theory} that by introducing finer approximations of both the integral and its gradient as the iteration number increases, we can obtain an improved rate of order $1/T$, up to logarithmic factors.  Moreover, we translate the optimization error in the objective into a bound on the error in the log-density, which is uncommon in the literature in the absence of strong convexity.  A further advantage of our approach is that we are able to extend it in Section~\ref{Sec:sconcave} to the more general problem of quasi-concave density estimation \citep{koenker2010quasi,seregin2010nonparametric}, thereby providing a computationally tractable alternative to the discrete Hessian approach of \cite{koenker2010quasi}.  Section~\ref{sec:compute} illustrates the practical benefits of our methodology in terms of improved computational timings on simulated data. Additional experimental details and applications on real data sets are provided in Appendix~\ref{sec:addl-expts}. Proofs of all main results can be found in Appendix~\ref{sec:proofs}, and background on the field of nonparametric inference under shape constraints can be found in Appendix~\ref{app:background}.


\paragraph{Notation:} We write $[n] := \{1,2,\ldots, n\}$, let $\bm{1} \in \mathbb{R}^n$ denote the all-ones vector, and denote the cardinality of a set $S$ by $|S|$.  For a Borel measurable set $C\subseteq \R^d$, we use $\mathrm{vol}(C)$ to denote its volume (i.e.~$d$-dimensional Lebesgue measure).  We write $\|\cdot\|$ for the Euclidean norm of a vector.  For $\mu > 0$, a convex function $f:\R^n \rightarrow \R$ is said to be $\mu$-strongly convex if $\bphi \mapsto f(\bm{\phi})-\frac{\mu}{2}\|\bm{\phi}\|^2$ is convex.  The notation $\partial f(\bm{\phi})$ denotes the subdifferential (set of subgradients) of $f$ at~$\bm{\phi}$.  Given a real-valued sequence $(a_n)$ and a positive sequence $(b_n)$, we write $a_n = \tilde{O}(b_n)$ if there exist $C,\gamma > 0$ such that $a_n \leq C b_n \log^\gamma (1+n)$ for all $n \in \mathbb{N}$.

\section{Understanding the structure of the optimization problem}\label{sec:Setup}
Throughout this paper, we assume that $\bx_1,\ldots,\bx_n \in \mathbb{R}^d$
are distinct and that their convex hull $C_n:=\mathrm{conv}(\bx_1,\ldots,\bx_n)$ has nonempty interior, so that $n \geq d+1$ and $\Delta := \mathrm{vol}(C_n)>0$. This latter assumption ensures the existence and uniqueness of a minimizer of the objective function in~\eqref{Eq:sigma}~\cite[Theorem~2.2]{dumbgen2011approximation}.  Recall that we define the \textit{lower convex envelope} function~\citep{rockafellar1997convex} $\cef: \mathbb{R}^n \rightarrow \mathcal{C}_d$ by
\begin{equation}
  \label{Eq:cef}
  \cef[\bphi](\bx) \equiv \cef[(\phi_1,\ldots,\phi_n)](\bx) := \sup\bigl\{g(\bx):g\in\mathcal{C}_d, g(\bx_i)\leq \phi_i~\forall i \in [n]\bigr\}.
\end{equation}
As mentioned in the introduction, in computing the MLE, we seek
\begin{equation}
\label{Eq:OptProblem}
\bphi^* := \argmin_{\bm\phi \in \mathbb{R}^n} f(\bphi),
\end{equation}
where
\begin{equation}\label{eqn:MLE-phi}
    f(\bphi) := \frac1n\bm1^\top\bphi+\int_{C_n}\exp\{-\cef[\bm\phi](\bm x)\}\, d\bm x =: \frac1n\bm 1^\top\bm\phi+I(\bphi).
\end{equation}
Note that~\eqref{Eq:OptProblem} can be viewed as a stochastic optimization problem by writing
\begin{equation}\label{eqn:MLE-StocOpt}
    f(\bphi)=\E F(\bphi,\bm\xi),
\end{equation}
where $\bm\xi$ is uniformly distributed on $C_n$ and where, for $\bx \in C_n$,
\begin{equation}\label{eqn:F}
    F(\bphi,\bx):=\frac1n\bm1^\top\bphi+\Delta e^{-\cef[\bphi](\bx)}.
  \end{equation}

Let $\bm X := [\bx_1 \, \cdots \, \bx_n]^\top \in \mathbb{R}^{n \times d}$, and for $\bx \in \mathbb{R}^d$, let $E(\bx) := \bigl\{\bm\alpha\in\R^n:\bm X^\top\bm\alpha = \bx, \bm1_n^\top\bm\alpha=1, \bm\alpha\geq 0\bigr\}$ denote the set of all weight vectors for which $\bx$ can be written as a weighted convex combination of $\bx_1,\ldots,\bx_n$.  Thus $E(\bx)$ is a compact, convex subset of $\R^n$.  The $\cef$ function is given by a linear program (LP)~\citep{koenker2010quasi,axelrod2019polynomial}:
\begin{equation}
  \cef[\bphi](\bx) = \inf_{\bm\alpha\in E(\bx)}~\bm\alpha^\top\bphi.     \tag{\mbox{$Q_0$}} \label{eqn:cef-LP}
\end{equation}
If $\bx\notin C_n$, then $E(\bx) =\emptyset$, and, with the standard convention that $\inf \emptyset := \infty$, we see that~\eqref{eqn:cef-LP} agrees with~\eqref{Eq:cef}. 
From the LP formulation, it follows that $\bphi \mapsto \cef[\bphi](\bx)$ is concave, for every $\bx \in \R^d$.

Given a pair $\bphi \in \mathbb{R}^n$ and $\bx \in C_n$, an optimal solution to~\eqref{eqn:cef-LP} may not be unique, in which case the map $\bm{\phi}\mapsto  \cef[\bphi](\bx)$ is not differentiable \cite[Proposition~B.25(b)]{bertsekas2016nonlinear}.  Noting that the infimum in~\eqref{eqn:cef-LP} is attained whenever $\bx \in C_n$, let
\begin{align*}
  A[\bm\phi](\bx) &:= \mathrm{conv}\bigl(\bigl\{\bm\alpha \in E(\bx): \bm\alpha^\top \bphi = \cef[\bphi](\bx)\bigr\}\bigr) \\
  &\phantom{:}= \bigl\{\bm\alpha \in E(\bx): \bm\alpha^\top \bphi = \cef[\bphi](\bx)\bigr\}.
\end{align*}
Danskin's theorem \cite[Proposition B.25(b)]{bertsekas2016nonlinear} applied to $-\cef[\bphi](\bx)$ then yields that for each $\bx \in C_n$, the subdifferential of $F(\bphi,\bx)$ with respect to $\bphi$ is given by 
\begin{equation}
    \partial F(\bphi,\bx):=\biggl\{\frac1n\bm 1-\Delta e^{-\cef[\bphi](\bx)}\bm \alpha:\bm\alpha \in A[\bm\phi](\bx)\biggr\}. 
\end{equation}
Since both $f$ and $F(\cdot,\bx)$ are finite convex functions on $\mathbb{R}^n$ (for each fixed $\bx \in C_n$ in the latter case), by \citet[Proposition~2.3.6(b) and Theorem~2.7.2]{clarke1990optimization}, the subdifferential of $f$ at $\bphi \in \mathbb{R}^n$ is given by
\begin{equation}\label{eqn:subgradient-phi}
    \partial f(\bphi):= \bigl\{\E\bm G(\bphi,\bm\xi):\bm G(\bphi,\bx) \in \partial F(\bphi,\bx) \text{ for each } \bx \in C_n\bigr\}.
\end{equation}
Observe that given any $\bphi \in \mathbb{R}^n$, the function $\bx \mapsto -\mathrm{cef}\bigl[\bphi + \log I(\bphi)\bm1\bigr](\bx)$ (where $I(\bphi)$ is the integral defined in~\eqref{eqn:MLE-phi}) is a log-density.  It is also convenient to let $\bar{\bphi} \in \R^n$ be such that $\exp\{-\cef[\bar{\bphi}]\}$ is the uniform density on $C_n$, so that $f(\bar\bphi)=\log \Delta+1$.   Proposition~\ref{prop:bound-for-phis} below (an extension of \cite[Lemma~2]{axelrod2019polynomial}) provides uniform upper and lower bounds on this log-density, whenever the objective function~$f$ evaluated at $\bphi$ is at least as good as that at $\bar{\bphi}$.  In more statistical language, these bounds hold whenever the log-likelihood of the density $\exp\bigl\{-\mathrm{cef}\bigl[\bphi + \log I(\bphi)\bm1\bigr](\cdot)\bigr\}$ is at least as large as that of the uniform density on the convex hull of the data, so in particular, they must hold for the log-concave MLE (i.e.~when $\bphi = \bphi^*$).  Let 
$\phi^0 := (n-1) +d(n-1)\log \bigl(2n + 2nd \log(2nd)\bigr) +\log \Delta$ and $\phi_0 :=-1 - d \log \bigl(2n + 2nd \log(2nd)\bigr) +\log \Delta$.

\begin{proposition}\label{prop:bound-for-phis}
For any $\bphi \in \R^n$ such that $f(\bphi)\leq \log \Delta+1$, we have
$\phi_0\leq \phi_i+\log I(\bphi)\leq \phi^0$ for all $i \in [n]$.
\end{proposition}
The following corollary is an immediate consequence of Proposition~\ref{prop:bound-for-phis}.
\begin{cor}\label{cor:verify-D-assum}
  Suppose that $\bphi \in \mathbb{R}^n$ satisfies $I(\bphi)=1$ and $f(\bphi)\leq f(\bar\bphi) = \log \Delta+1$.  Then $\bphi^* \in \mathbb{R}^n$ defined in~\eqref{Eq:OptProblem} satisfies
  \[
    \norm{\bphi-\bphi^*}\leq \sqrt{n}(\phi^0-\phi_0).
  \]
\end{cor}
Corollary~\ref{cor:verify-D-assum} gives a sense in which any $\bphi \in \R^n$ for which the objective function is `good' cannot be too far from the optimizer $\bphi^*$; here, `good' means that the objective should be no larger than that of the uniform density on the convex hull of the data.  Moreover, an upper bound on the integral $I(\bphi)$ provides an upper bound on the norm of any subgradient $\bm g(\bphi)$ of $f$ at $\bphi$.
\begin{proposition}\label{prop:bound-subgradient}
Any subgradient $\bm g(\bphi) \in \R^n$ of $f$ at $\bphi \in \R^n$ satisfies
$\norm{\bm g(\bphi)}^2\leq \max \bigl\{1/n + 1/4,I(\bphi)^2\bigr\}.$
\end{proposition}

\section{Computing the log-concave MLE}\label{sec:computation}

As mentioned in the introduction, subgradient methods~\citep{shor1985minimization,polyak1987introduction} tend to be slow for minimizing the objective function $f$ defined in~\eqref{eqn:MLE-phi} \cite{cule2010maximum}.  Our alternative approach involves the minimizing the representation of $f$ given in~\eqref{eqn:MLE-StocOpt} via smoothing techniques, which offer superior computational guarantees and practical performance in our numerical experiments.

\subsection{Smoothing techniques}\label{sec:smoothing-techs}
We present two smoothing techniques to find the minimizer $\bphi^* \in \R^n$ of the nonsmooth convex optimization problem~\eqref{Eq:OptProblem}.
By Proposition~\ref{prop:bound-for-phis}, we have that $\bphi^* \in {\bm{\Phi}}$, where
\begin{equation}
    \bm{\Phi} := \{\bphi=(\phi_1,\ldots,\phi_n)\in\R^n:\phi_0 \leq \phi_i\leq \phi^0 \text{ for } i \in [n]\},
\end{equation}
with $\phi_0,\phi^0 \in \R$.  In what follows we present two smoothing techniques: one based on Nesterov smoothing~\citep{nesterov2005smooth} and the second on randomized smoothing~\citep{duchi2012randomized}. 

\subsubsection{Nesterov smoothing}

Recall that the non-differentiability in $f$ in~\eqref{eqn:MLE-phi} is due to the LP~\eqref{eqn:cef-LP} potentially having  multiple optimal solutions. Therefore, following~\cite{nesterov2005smooth}, we consider replacing this LP with the following quadratic program (QP): 
\begin{equation}\label{eqn:h-QP}
    q_u[\bphi](\bx) := \inf_{\bm\alpha \in E(\bx)}~\biggl(\bm{\alpha}^\top\bm{\phi}+\frac{u}2\norm{\bm\alpha-\bm\alpha_0}^2-\frac u2\biggr), \tag{\mbox{$Q_u$}}
\end{equation}
where $\bm\alpha_0:=(1/n)\bm 1\in\R^n$ is the center of $E(\bx)$, and where $u\geq 0$ is a regularization parameter that controls the extent of the quadratic regularization of the objective.
With this definition, we have $q_0[\bphi](\bx)=\cef[\bphi](\bx)$. For $u>0$, due to the strong convexity of the function $\bm{\alpha} \mapsto \bm{\alpha}^\top\bm{\phi}+ ({u}/2)\|\bm\alpha-\bm\alpha_0\|^2$ on the convex polytope $E(\bx)$, \eqref{eqn:h-QP} admits a unique solution that we denote by $\bm\alpha^*_u[\bphi](\bx)$. It follows again from Danskin's theorem that $\bphi \mapsto q_u[\bphi](\bx)$ is differentiable for such $u$, with gradient $\nabla_{\bphi} q_u[\bphi](\bx) = \bm\alpha^*_u[\bphi](\bx)$. 

Using $q_{u}[\bphi](\bx)$ instead of $q_{0}[\bphi](\bx)$ in~\eqref{eqn:MLE-phi}, 
we obtain a smooth objective $\bphi \mapsto \tilde{f}_u(\bphi)$, given by
\begin{align}\label{eqn:tildefu}
    \tilde{f}_u(\bphi) := \frac{1}{n}\bm1^\top\bm\phi+\int_{C_n}\exp\{-q_u[\bm\phi](\bm x)\}\dx=\E\tilde F_u(\bphi,\bm\xi),
\end{align}
where $\tilde{F}_u(\bphi,\bx):=({1}/{n})\bm1^\top\bphi+\Delta \exp\{-q_u[\bphi](\bx)\}$, and where $\bm\xi$ is again uniformly distributed on $C_n$.  
We may differentiate under the integral \citep[e.g.][Theorem~6.28]{klenke2014probability} to see that the partial derivatives of $\tilde{f}_u$ with respect to each component of $\bphi$ exist, and moreover they are continuous (because $\bphi \mapsto \bm\alpha^*_u[\bphi](\bx)$ is continuous by Proposition~\ref{prop:QP-properties}), so $\nabla_{\bphi} \tilde{f}_u(\bphi)= \E[\tilde{\bm G}_u(\bphi,\bm\xi)]$, where 
\begin{equation}\label{defn-G-hat-u-1}
\tilde{\bm G}_u(\bphi,\bx) := \nabla_{\bphi} \tilde{F}_u(\bphi,\bx)=\frac1n\bm 1-\Delta e^{-q_u[\bphi](\bx)}\alpha^*_u[\bphi](\bx).
\end{equation}

Proposition~\ref{prop:fhat-properties} below presents some properties of the smooth objective 
$\tilde{f}_u$. 
\begin{proposition}\label{prop:fhat-properties}
For any $\bphi\in{\bm{\Phi}}$,
we have 

\noindent(a) $0 \leq \tilde{f}_u(\bphi)-\tilde{f}_{u'}(\bphi)\leq \frac{u-u'}{2}e^{u'/2}I(\bphi)$ for $u' \in [0,u]$;

\noindent(b) For every $u \geq 0$, the function $\bphi \mapsto \tilde{f}_u(\bphi)$ is convex and $\Delta e^{-\phi_0+u/2}$-Lipschitz; 

\noindent(c) For every $u \geq 0$, the function $\bphi \mapsto \tilde{f}_u(\bphi)$ has $\Delta e^{-\phi_0+u/2}(1+u^{-1})$-Lipschitz gradient;

\noindent(d) $\E\bigl(\|{\tilde{\bm G}_u(\bphi,\bm\xi)-\nabla_{\bm\phi}\tilde{f}_u(\bphi)}\|^2\bigr) \leq (\Delta e^{-\phi_0+u/2})^2$ for every $u \geq 0$.

\end{proposition}

\subsubsection{Randomized smoothing} 

Our second smoothing technique is randomized smoothing \citep{lakshmanan2008decentralized,yousefian2012stochastic,duchi2012randomized}: we take the expectation of a random perturbation of the argument of $f$.  Specifically, for $u \geq 0$, let
\begin{equation}\label{eqn:barfu}
    \bar{f}_u(\bphi) :=\E f(\bphi+u\bm z),
\end{equation}
where $\bz$ is uniformly distributed on the unit $\ell_{2}$-ball in $\R^n$.  Thus, similar to Nesterov smoothing, $\bar{f}_0 = f$, and the amount of smoothing increases with $u$.  From a stochastic optimization viewpoint, we can write 
\[
\bar{f}_u(\bphi)=\E F(\bphi+u\bz,\bm\xi)~~~~\text{and}~~~\nabla_{\bphi} \bar{f}_u(\bphi)=\E \bm{G}(\bphi+u\bz,\bm\xi)
\]
where $\bm{G}(\bphi+u\bm v,\bx) \in \partial F(\bphi+u\bm v,\bx)$, and where the expectations are taken over independent random vectors $\bz$, distributed uniformly on the unit Euclidean ball in $\R^n$, and $\bm\xi$, distributed uniformly on $C_n$.  Here the gradient expression follows from, e.g.,~\citet[Lemma 3.3(a)]{lakshmanan2008decentralized}, ~\citet[Lemma 7]{yousefian2012stochastic}; since $F(\bphi+u\bm v,\bx)$ is differentiable almost everywhere with respect to $\bphi$, the expression for $\bar{f}_u(\bphi)$ does not depend on the choice of subgradient.



Proposition~\ref{prop:fbar-properties} below lists some properties of $\bar{f}_{u}$ and its gradient.  It extends \citet[Lemmas~7 and~8]{yousefian2012stochastic} by exploiting special properties of the objective function to sharpen the dependence of the bounds on $n$.  
\begin{proposition}\label{prop:fbar-properties}
For any $u \geq 0$ and $\bphi \in {\bm{\Phi}}$, we have

\noindent(a) $0\leq \bar{f}_u(\bphi)-f(\bphi)\leq I(\bphi)ue^u\sqrt{\frac{2\log n}{n+1}}$;

\noindent(b) $\bar{f}_{u'}(\bphi)\leq \bar{f}_u(\bphi)$ for any $u' \in [0,u]$;

\noindent(c) $\bphi \mapsto \bar{f}_u(\bphi)$ is convex and $\Delta e^{-\phi_0+u}$-Lipschitz;

\noindent(d) $\bphi \mapsto \bar{f}_u(\bphi)$ has $\Delta e^{-\phi_0+u}n^{1/2}/u$-Lipschitz gradient;

\noindent(e) $\E\bigl(\bigl\|\bm G(\bphi+u\bz,\bm\xi)-\nabla \bar{f}_u(\bphi)\bigr\|^2\bigr) \leq (\Delta e^{-\phi_0+u})^2$ whenever $\bm G(\bphi+u\bm v,\bx) \in \partial F(\bphi+u\bm v,\bx)$ for every $\bm v \in \R^n$ with $\|\bm v\| \leq 1$ and $\bx \in C_n$.

\end{proposition}


\subsection{Stochastic first-order methods for smoothing sequences} 

Our proposed algorithm for computing the log-concave MLE is given in Algorithm~\ref{algo:smoothing}.  It relies on the choice of a smoothing sequence of $f$, which may be constructed using Nesterov or randomized smoothing, for instance.  For a non-negative sequence $(u_t)_{t \in \mathbb{N}_0}$, this smoothing sequence is denoted by $(\ell_{u_t})_{t \in \mathbb{N}_0}$, where $\ell_{u_t}:=\tilde f_{u_t}$ is given by~\eqref{eqn:tildefu} or $\ell_{u_t}:=\bar f_{u_t}$ is given by~\eqref{eqn:barfu}.  In Algorithm~\ref{algo:smoothing}, $P_{\B\Phi}:\mathbb{R}^n \rightarrow \B\Phi$ denotes the projection operator onto the closed convex set $\B\Phi$, which is essentially a threshold clipping operator. In fact, Algorithm~\ref{algo:smoothing} is a modification of an algorithm due to~\cite{duchi2012randomized}, and can be regarded as an accelerated version of the dual averaging scheme~\citep{nesterov2009primal} applied to $(\ell_{u_t})$. 
\begin{algorithm}[!ht]
\begin{algorithmic}[1]
\Require Smoothing sequence $(\ell_{u_t})$ whose gradients have Lipschitz constants $(L_t)_{t \in \mathbb{N}_0}$; initialization $\bm\phi_0 \in \mathbb{R}^n$; learning rate sequence $(\eta_t)_{t \in \mathbb{N}}$ of positive real numbers; number of iterations $T \in \mathbb{N}$
\State $\bphi_0^{(x)}=\bphi_0^{(y)}=\bphi_0^{(z)}=\bphi_0$, $\bm s_t = \bm 0 \in \R^n$, $\theta_0 = 1$
\For {$t=0,\ldots,T-1$}
\State Compute an approximation $\bm g_t$ of $\nabla_{\bphi} \ell_{u_t}\bigl(\bphi_t^{(y)}\bigr)$; see Section~\ref{SubSubSec:Gradient}
\State $\bm{s}_{t+1}=\bm s_t+\bm g_t/\theta_t$ 
\State $\theta_{t+1}=\frac{2}{1+\sqrt{1+4/\theta_t^2}}$
\State $\bphi_{t+1}^{(z)} = P_{\B\Phi}\bigl(\bphi_0- \frac{\bm{s}_t}{L_{t+1}+\eta_{t+1}/\theta_{t+1}}\bigr)$   
\State $\bphi_{t+1}^{(x)} = (1-\theta_t)\bphi_t^{(x)}+\theta_t\bphi_{t+1}^{(z)}$
\State $\bphi_{t+1}^{(y)} = (1-\theta_{t+1})\bphi_{t+1}^{(x)}+\theta_{t+1}\bphi_{t+1}^{(z)}$
\EndFor
\State \Return $\bphi_{t+1}^{(x)}$
\end{algorithmic}
\caption{Accelerated stochastic dual averaging on a smoothing sequence with increasing grids}
\label{algo:smoothing}
\end{algorithm}

\subsubsection{Approximating the gradient of the smoothing sequence}
\label{SubSubSec:Gradient}

In Line 3 of Algorithm~\ref{algo:smoothing}, we need to compute an approximation of the gradient $\nabla_{\bphi} \ell_{u}$, for a general $u \geq 0$.  A key step in this process is to approximate the integral $I(\cdot)$, as well as a subgradient of $I$, at an arbitrary $\bm\phi \in \mathbb{R}^n$.  \cite{cule2010maximum} provide explicit formulae for these quantities, based on a triangulation of $C_n$, using tools from computational geometry.  For practical purposes, \cite{cule2008auxiliary} apply a Taylor expansion to approximate the analytic expression.  The \texttt{R} package \code{LogConcDEAD} \citep{JSSv029i02} uses this method to evaluate the exact integral at each iteration, but since this is time-consuming, we will only use this method at the final stage of our proposed algorithm as a polishing step\footnote{{Once our algorithm terminates at $\tilde{\bphi}_T$, say, we evaluate the integral $I(\tilde{\bphi}_T)$ in the same way as \cite{cule2010maximum}.  Our final output, then is $\bphi_T := \tilde{\bphi}_T + \log I(\tilde{\bphi}_T)\bm{1}$; this final step not only improves the objective function, but also guarantees that $\exp[-\mathrm{cef}[\bphi_T](\cdot)]$ is a log-concave density.}}.

An alternative approach is to use numerical integration\footnote{Yet another option involves sampling from a log-concave density~\citep{cule2010maximum,dalalyan2017theoretical}. \cite{axelrod2019polynomial}~discuss interesting polynomial-time sampling methods to approximate $I(\cdot)$, but as noted by~\cite{rathke2019fast}, these methods may not be practically efficient, and we do not pursue them here.}.  Among deterministic schemes, \cite{rathke2019fast} observed empirically that the simple Riemann sum with uniform weights appears to perform the best among several multi-dimensional integration techniques.  Random (Monte Carlo) approaches to approximate the integral are also possible: given a collection of grid points $\mathcal{S}=\{\bm\xi_1,\ldots,\bm\xi_m\}$, we approximate the integral as $I_{\mathcal{S}}(\bphi) := ({\Delta}/{m})\sum_{\ell=1}^m \exp\{-\cef[\bphi](\bm\xi_\ell)\}.$ This leads to an approximation of the objective $f$ given by
\begin{equation}\label{f-s-approximate}
f(\bphi) \approx \frac1n\bm1^\top\bphi +  I_{\mathcal{S}}(\bphi) =: f_{\mathcal{S}}(\bphi).
\end{equation}
Since $f_{\mathcal{S}}$ is a finite, convex function on $\mathbb{R}^n$, it has a subgradient at each $\bm\phi \in \mathbb{R}^n$, given by
\[
\bm g_{\mathcal{S}}(\bphi) := \frac{1}{m}\sum_{\ell=1}^m \bm G(\bphi,\bm\xi_\ell).
\]
As the effective domain of $\cef[\bphi](\cdot)$ is $C_{n}$, we consider grid points $\mathcal{S} \subseteq C_{n}$.

 
We now illustrate how these ideas allow us to approximate the gradient of the smoothing sequence, and initially consider Nesterov smoothing, with $\ell_u=\tilde{f}_u$.  If  
$\mathcal{S}=\{\bm \xi_1,\ldots,\bm \xi_m\} \subseteq C_n$ denotes a collection of grid points (either deterministic or Monte Carlo based), then $\nabla_{\bphi} \ell_{u}$ can be approximated by $\tilde{\bm{g}}_{u,\mathcal{S}}$, where
\begin{equation}
  \tilde{\bm{g}}_{u,\mathcal{S}}(\bphi) := \frac1n\bm1-\frac{\Delta}{m}\sum_{j=1}^me^{-q_u[\bphi](\bm\xi_j)}\alpha^*_u[\bphi](\bm\xi_j).\label{eqn:gtilde-S}
\end{equation}
In fact, we distinguish the cases of deterministic and random $\mathcal{S}$ by writing this approximation as $\tilde{\bm{g}}_{u,\mathcal{S}}^{\mathrm{D}}$ and $\tilde{\bm{g}}_{u,\mathcal{S}}^{\mathrm{R}}$ respectively.

For the randomized smoothing method with $\ell_u=\bar{f}_u$, the approximation is slightly more involved.  Given $m$ grid points $\mathcal{S}=\{\bm \xi_1,\ldots,\bm \xi_m\} \subseteq C_n$ (again either deterministic or random), and independent random vectors $\bm z_1,\ldots,\bm z_m$, each uniformly distributed on the unit Euclidean ball in $\R^n$, we can approximate $\nabla_{\bphi} \ell_{u}(\bphi)$ by
\begin{equation}\bar{\bm g}_{u,\mathcal{S}}^{\circ}(\bphi) = \frac1n\bm{1}-\frac{\Delta}{m}\sum_{j=1}^me^{-\cef[\bphi+u\bm z_j](\bm\xi_j)}\alpha^*[\bphi+u\bm z_j](\bm\xi_j),\label{eqn:gbar-S}\end{equation}
with $\circ \in \{\mathrm{D},\mathrm{R}\}$ again distinguishing the cases of deterministic and random $\mathcal{S}$.

\subsubsection{Related literature}

As mentioned above, Algorithm~\ref{algo:smoothing} is an accelerated version of the dual averaging method of \cite{nesterov2009primal}, which to the best of our knowledge has not been studied in the context of log-concave density estimation previously.  Nevertheless, related ideas have been considered for other optimization problems~\citep[e.g.][]{xiao2010dual,duchi2012randomized}.  Relative to previous work, our approach is quite general, in that it applies to both of the smoothing techniques discussed in Section~\ref{sec:smoothing-techs}, and allows the use of both deterministic and random grids to approximate the gradients of the smoothing sequence.  Another key difference with earlier work is that we allow the grid $\mathcal{S}$ to depend on $t$, so we write it as $\mathcal{S}_t$, with $m_t := |\mathcal{S}_t|$; in particular, inspired by both our theoretical results and numerical experience, we take $(m_{t})$ to be a suitable increasing sequence.

\section{Theoretical analysis of optimization error of Algorithm~\ref{algo:smoothing}}
\label{sec:theory}

We have seen in Propositions~\ref{prop:fhat-properties} and \ref{prop:fbar-properties} that the two smooth functions $\tilde{f}_{u}$ and $\bar{f}_{u}$ enjoy similar properties --- according to  Proposition~\ref{prop:fhat-properties}(a) to (c) and Proposition~\ref{prop:fbar-properties}(a) to (d), both $\tilde{f}_{u}$ and $\bar{f}_{u}$
satisfy the following assumption:
\begin{assumption}[Assumptions on smoothing sequence] \label{assump:smoothing-seq}
There exists $r \geq 0$ such that for any $\bphi\in {\bm{\Phi}}$,

\noindent(a) we can find $B_0 > 0$ with $f(\bphi)\leq \ell_{u}(\bphi)\leq f(\bphi)+B_0I(\bphi)u$ for all $u \in [0,r]$;

\noindent(b) $\ell_{u'}(\bphi)\leq \ell_u(\bphi)$ for all $u' \in [0,u]$;

\noindent(c) for each $u \in [0,r]$, the function $\bphi \mapsto \ell_u(\bphi)$ is convex and has ${B_1}/{u}$-Lipschitz gradient, for some $B_1>0$.
\end{assumption}

Recall from Section~\ref{sec:computation} that we have four possible choices corresponding to a combination of the smoothing and integral approximation methods, as summarized in Table~\ref{table:summary}.
\begin{table}[!h]
    \centering
    \caption{Summary of options for smoothing and gradient approximation methods.}
    \label{table:summary}
    \scalebox{0.92}{
    \begin{tabular}{rcc|rcc}
    \toprule
         Options& Smoothing & Approximation&Options& Smoothing & Approximation \\
         \midrule
        1 & $\breve{f}_u$ &$\tilde{\bm g}_{u,\mathcal{S}}^{\mathrm{D}}$ &3 & $\bar{f}_u$ & $\bar{\bm g}_{u,\mathcal{S}}^{\mathrm{D}}$ \\
        2 & $\breve{f}_u$ &$\tilde{\bm g}_{u,\mathcal{S}}^{\mathrm{R}}$ &      4 & $\bar{f}_u$ & $\bar{\bm g}_{u,\mathcal{S}}^{\mathrm{R}}$\\
        \bottomrule
    \end{tabular}}
\end{table}

Once we select an option, in line 3 of Algorithm~\ref{algo:smoothing}, we can take 
\[
  \bm{g}_t=\check{\bm g}_{u_t,\mathcal{S}_t}^{\circ}(\bphi_t^{(y)}),
\]
where $\check{~} \in \{\breve{~},~\bar{~}\}$ and $\circ \in \{\mathrm{D},\mathrm{R}\}$.  To encompass all four approximation choices in Line~3 of Algorithm~\ref{algo:smoothing}, we make the following assumption on the gradient approximation error $\bm e_t:= \bm g_t - \nabla_{\bphi} \ell_{u_t}(\bphi_t^{(y)})$:
\begin{assumption}[Gradient approximation error]\label{assump:error}
There exists $\sigma > 0$ such that 
\begin{equation}\label{cond-ass-1}
\E\bigl(\norm{\bm{e}_t}^2|\mathcal{F}_{t-1}\bigr) \leq  {\sigma^2}/{m_t}~~~~\text{for all } t \in \mathbb{N}_0,
\end{equation}
where $\mathcal{F}_{t-1}$ denotes the $\sigma$-algebra generated by all random sources up to iteration $t-1$ (with $\mathcal{F}_{-1}$ denoting the trivial $\sigma$-algebra).  
\end{assumption}

When $\mathcal{S}$ is a Monte Carlo random grid (options 2 and 4), the approximate gradient $\bm g_t$ is an average of $m_t$ independent and identically distributed random vectors, each being an unbiased estimator of $\nabla \ell_{u_t}(\bphi_t^{(y)})$. Hence,~\eqref{cond-ass-1} holds true with 
$\sigma^2$ determined by the bounds in Proposition~\ref{prop:fhat-properties}(d) (option 2) and Proposition~\ref{prop:fbar-properties}(e) (option 4).  For a deterministic Riemann grid $\mathcal{S}$ and Nesterov's smoothing technique (option 1), $\bm e_t$ is deterministic, and arises from using $\tilde{\bm g}_{u,\mathcal{S}}(\bphi)$ in~\eqref{eqn:gtilde-S} to approximate $\nabla_{\bphi} \tilde{f}_u(\bphi) =\E[\tilde{\bm G}_u(\bm\phi,\bm\xi)]$. For the deterministic Riemann grid and randomized smoothing (option 3), the error $\bm{e}_{t}$ can be decomposed into a random estimation error term (induced by $\bm z_1,\ldots,\bm z_{m_t}$) and a deterministic approximation error term (induced by $\bm\xi_1,\ldots,\bm\xi_{m_t}$) as follows:
\begin{align*}
    \bm e_t
    &=\frac{1}{m_t}\sum_{j=1}^{m_t}\bigl(\bm G(\bphi_t^{(y)}+u_t\bz_j,\bm\xi_j)-\E[\bm G(\bphi_t^{(y)}+u_t\bz,\bm\xi_j)|\mathcal{F}_{t-1}]\bigr) \\
    &\hspace{1cm}+\biggl(\frac{1}{m_t}\sum_{j=1}^{m_t}\E[\bm G(\bphi_t^{(y)}+u_t\bz,\bm\xi_j)|\mathcal{F}_{t-1}]-\E[\bm G(\bphi_t^{(y)}+u_t\bz,\bm\xi)|\mathcal{F}_{t-1}]\biggr).
\end{align*}
It can be shown using this decomposition that $\mathbb{E}(\|\bm e_t\|^2|\mathcal{F}_{t-1}) = O(1/m_t)$ under regularity conditions.




Theorem~\ref{thm:smoothing-guarantee-expectation} below establishes our desired computational guarantees for Algorithm~\ref{algo:smoothing}.  We write $D:=\sup_{\bphi,\tilde{\bphi}\in{\bm{\Phi}}} \| \bphi-\tilde{\bphi} \|$ for the diameter of $\bm \Phi$.
\begin{theorem}\label{thm:smoothing-guarantee-expectation}
Suppose that Assumptions~\ref{assump:smoothing-seq} and~\ref{assump:error} hold, and define the sequence $(\theta_t)_{t \in \mathbb{N}_0}$ by $\theta_0 := 1$ and $\theta_{t+1}:= 2\bigl(1+\sqrt{1+4/\theta_t^2}\bigr)^{-1}$ for $t \in \mathbb{N}_0$.  Let $u > 0$, let $u_t := \theta_t u$ and take $L_t=B_1/u_t$ and $\eta_t=\eta$ for all~$t \in \mathbb{N}_0$ as input parameters to Algorithm~\ref{algo:smoothing}.  Writing $M_T^{(1)}:=\sqrt{\sum_{t=0}^{T-1}m_t^{-1}}$ and $M_T^{(1/2)}:=\sum_{t=0}^{T-1}m_t^{-1/2}$, we have for any $\bphi \in {\bm{\Phi}}$ that
\begin{equation}\label{thm1-state-1}
\E[f(\bphi_T^{(x)})]-f(\bphi)\leq \frac{B_1D^2}{Tu}+\frac{4B_0I(\bphi)u}{T}+\frac{\eta D^2}{T}+\frac{\sigma^2(M_T^{(1)})^2}{T\eta}+\frac{2D\sigma M_T^{(1/2)}}{T}.
\end{equation}
In particular, taking $\bphi=\bphi^*$, and choosing $u=({D}/{2})\sqrt{B_1/B_0}$ and $\eta=({\sigma M_T^{(1)}})/{D}$, we obtain
\begin{equation}\label{thm1-state-2}
\varepsilon_{T} := \E[f(\bphi_T^{(x)})]-f(\bphi^*)\leq \frac{4\sqrt{B_0B_1}D}{T}+\frac{2\sigma DM_T^{(1)}}{T}+\frac{2D\sigma M_T^{(1/2)}}{T}.
\end{equation}
Moreover, if we further assume that $\E(\bm{e}_t|\mathcal{F}_{t-1})=\bm0$ (e.g.~by using options 2 and 4), then we can remove the last term of both inequalities above.
\end{theorem}
For related general results that control the expected optimization error for smoothed objective functions, see, e.g., \cite{nesterov2005smooth}, \cite{tseng2008accelerated}, \cite{xiao2010dual}, \cite{duchi2012randomized}.  With deterministic grids (corresponding to options 1 and 3), if we take $|\mathcal S_{t}| = m$ for all $t$, then $M_T^{(1/2)}=T/\sqrt{m}$, and the upper bound in~\eqref{thm1-state-2} does not converge to zero as $T \rightarrow \infty$.  On the other hand, if we take $|\mathcal{S}_t|=t^2$, for example, then $\sup_{T \in \mathbb{N}} M_T^{(1)} < \infty$ and $M_T^{(1/2)}= \tilde{O}(1)$, and we find that $\varepsilon_{T} = \tilde{O}(1/T)$.  For random grids (options 2 and 4), if we take $|\mathcal S_{t}| = m$ for all $t$, then $M_T^{(1)}=\sqrt{T/m}$ and we recover the $\varepsilon_{T} = O(1/\sqrt{T})$ rate for stochastic subgradient methods \citep{polyak1987introduction}.  This can be improved to $\varepsilon_{T} =\tilde{O}(1/T)$ with $m_t = t$, or even $\varepsilon_{T} = O(1/T)$ if we choose $(m_t)_t$ such that $\sum_{t=0}^\infty m_t^{-1} < \infty$.  
A direct application of the theory of \cite{duchi2012randomized} would yield an error rate of $\varepsilon_{T} = O(n^{1/4}/T + 1/\sqrt{T})$.  On the other hand, Theorem~\ref{thm:smoothing-guarantee-expectation} shows that, owing to the increasing sequence of grid sizes used to approximate the gradients in Step~3 of Algorithm~\ref{algo:smoothing}, we can improve this rate to $\tilde{O}(1/T)$. Note however, 
that this improvement is in terms of the number of iterations $T$, and not the total number of stochastic oracle queries (equivalently, the total number of LPs~\eqref{eqn:cef-LP}), which is given by $T_{\mathrm{query}}:=\sum_{t=0}^{T-1}m_t$. \cite{agarwal2012information} and \cite{nemirovsky1983problem} have shown that 
the optimal expected number of stochastic oracle queries is of order $1/\sqrt{T_{\mathrm{query}}}$, which is attained by the algorithm of~\cite{duchi2012randomized}.
For our framework, by taking $m_t=t$, we have $T_{\mathrm{query}}=\sum_{t=0}^{T-1} m_t=\tilde{O}(T^2)$, so after $T_{\mathrm{query}}$ stochastic oracle queries, our algorithm also attains the optimal error on the objective function scale, up to a logarithmic factor.  Other advantages of our algorithm and the theoretical guarantees provided by Theorem~\ref{thm:smoothing-guarantee-expectation} relative to the earlier contributions of \cite{duchi2012randomized} are that we do not require an upper bound on $I(\bphi)$ and are able to provide a unified framework that includes Nesterov smoothing and an alternative gradient approximation approach by numerical integration in addition to randomized smoothing scheme with stochastic gradients.  Moreover, we can exploit the specific structure of the log-concave density estimation problem to provide much better Lipschitz constants for the randomized smoothing sequence than would be obtained using the generic constants of \cite{duchi2012randomized}.  For example, our upper bound in Proposition~\ref{prop:fbar-properties}(a) is of order $O(n^{-1/2}\log^{1/2}n)$, whereas a naive application of the general theory of \cite{duchi2012randomized} would only yield a bound of $O(1)$.  A further improvement in our bound comes from the fact that it now involves $I(\bphi)$ directly, as opposed to an upper bound on this quantity.

In Theorem~\ref{thm:smoothing-guarantee-expectation}, the computational guarantee depends upon $B_0,B_1,\sigma$ in Assumptions~\ref{assump:smoothing-seq} and~\ref{assump:error}.  In light of Propositions~\ref{prop:fhat-properties} and~\ref{prop:fbar-properties}, Table~\ref{table:smoothing-comparison} illustrates how these quantities, and hence the corresponding guarantees, differ according to whether we use Nesterov smoothing or randomized smoothing.

The randomized smoothing procedure requires solving LPs, whereas Nesterov's smoothing technique requires solving QPs. While both of these problems are highly structured and can be solved efficiently by off-the-shelf solvers \citep[e.g.,][]{gurobi}, we found the LP solution times to be faster than those for the QP. Additional computational details are discussed in Section~\ref{sec:compute}.

\begin{table}[h]
    \centering
        \caption{Comparion of constants in Assumption~\ref{assump:smoothing-seq} for different smoothing schemes with $u \in [0,r]$. Here, $\sigma$ corresponds to random grid points (options 2 and 4), the optimal $\eta$ is taken to be proportional to $\sigma$, the optimal $u$ is proportional to $\sqrt{B_1/B_0}$, we take $C_1= \sqrt{\Delta e^{-\phi_0}}$; $\sqrt{B_0B_1}$ determines the first term in the error rate.}
    \label{table:smoothing-comparison}
    \resizebox{0.95\textwidth}{!}{
    \begin{tabular}{l|ccccc}
    \toprule
         &$B_0$&$B_1$&$\sigma~(\eta\propto\sigma)$&$u~(\propto\sqrt{B_1/B_0})$ &$\sqrt{B_0B_1}$ \\
         \midrule
         Nesterov&$1/2$&$\Delta e^{-\phi_0+r/2}(r+1)$&$\Delta e^{-\phi_0+r/2}$&\multicolumn{2}{c}{$C_1e^{r/4}\sqrt{(r+1)}O(1)$}\\
         Randomized& $\sqrt{2n^{-1}\log n}e^r$&$\Delta e^{-\phi_0+r}\sqrt{n}$&$\Delta e^{-\phi_0+r}$&$C_1\tilde{O}(\sqrt{n})$&$C_1e^r\tilde{O}(1)$\\
         \bottomrule
    \end{tabular}}
\end{table}

Note that Theorem~\ref{thm:smoothing-guarantee-expectation} presents error bounds in expectation, though for option 1, since we use Nesterov's smoothing technique and the Riemann sum approximation of the integral, the guarantee in Theorem~\ref{thm:smoothing-guarantee-expectation} holds without the need to take an expectation.  Theorem~\ref{thm:smoothing-guarantee-highprob} below presents corresponding high-probability guarantees. For simplicity, we present results for options 2 and 4, which rely on the following assumption:
\begin{assumption}\label{assump:error-subGaussian}
Assume that $\E(\bm{e}_t|\mathcal{F}_{t-1})=\bm{0}$ and that $\mathbb{E}\bigl(e^{\|\bm e_t\|^2/\sigma_t^2} \mid \mathcal{F}_{t-1}\bigr) \leq e$, where $\sigma_t=\sigma/\sqrt{m_t}$. 
\end{assumption}
\begin{theorem}\label{thm:smoothing-guarantee-highprob}
  Suppose that Assumptions~\ref{assump:smoothing-seq} and~\ref{assump:error-subGaussian} hold, and define the sequence $(\theta_t)_{t \in \mathbb{N}_0}$ by $\theta_0 := 1$ and $\theta_{t+1}:= 2\bigl(1+\sqrt{1+4/\theta_t^2}\bigr)^{-1}$ for $t \in \mathbb{N}_0$.  Let $u > 0$, let $u_t := \theta_t u$ and take $L_t=B_1/u_t$ and $\eta_t=\eta$ for all~$t \in \mathbb{N}_0$ as input parameters to Algorithm~\ref{algo:smoothing}.  Writing $M_T^{(2)}:=\sqrt{\sum_{t=0}^{T-1}m_t^{-2}}$ and $M_T^{(1)}:=\sqrt{\sum_{t=0}^{T-1}m_t^{-1}}$, and choosing $u=({D}/{2})\sqrt{B_1/B_0}$ and $\eta=({\sigma M_T^{(1)}})/{D}$ as in Theorem~\ref{thm:smoothing-guarantee-expectation}, for any $\delta \in (0,1)$, we have
with probability at least $1-\delta$ that
\begin{align*}
      f(\bphi_T^{(x)})-f(\bphi^*)&\leq \frac{2\sqrt{B_0B_1}D}{T}+\frac{\sigma DM_T^{(1)}}{T}+\frac{4\sigma DM_T^{(1)}\sqrt{\log\frac2\delta}}{T}\\
      &\hspace{2.5cm}+\frac{4\sigma D\max\bigl\{M_T^{(2)}\sqrt{2e\log\frac2\delta},m_0^{-1}\log\frac2\delta\bigr\}}{M_T^{(1)} T}.
\end{align*}

\end{theorem}

For option 3, we would need to consider the approximation error from the Riemann sum, and the final error rate would include additional $O(1/T)$ terms. We omit the details for brevity. 


Finally in this section, we relate the error of the objective to the error in terms of $\bphi$, as measured through the squared $L_2$ distance between the corresponding lower convex envelope functions. 
\begin{theorem}\label{thm:error-in-cef}
For any $\bphi\in{\bm{\Phi}}$, we have
\begin{equation}
    \int_{C_n}\bigl\{\cef[\bphi](\bx)-\cef[\bphi^\ast](\bx)\bigr\}^2\,d\bx\leq 2e^{\phi^0}\bigl\{f(\bphi)-f(\bphi^\ast)\bigr\}.
\end{equation}
\end{theorem}

\section{Beyond log-concave density estimation}
\label{Sec:sconcave}
In this section, we extend our computational framework beyond the log-concave density family, through the notion of $s$-concave densities. 
For $s\in \R$, define domains $\mathcal{D}_s$ and $\psi_s:\mathcal{D}_s \rightarrow \R$ by
\[
\mathcal{D}_s :=\left\{\begin{array}{ll}[0,\infty)&\text{ if } s<0\\(-\infty,\infty)&\text{ if } s= 0\\(-\infty,0]&\text{ if } s>0.\end{array}\right.~~~~\text{and}~~~
 \psi_s(y) := \left\{\begin{array}{ll}y^{1/s}&\text{ if } s<0\\e^{-y}&\text{ if } s= 0\\(-y)^{1/s}&\text{ if } s>0.\end{array}\right.
\]
\begin{definition}[$s$-concave density, \cite{seregin2010nonparametric}]\label{def:s-concave}
For $s\in \R$, 
the class $\mathcal{P}_s(\R^d)$ of $s$-concave density functions on $\R^d$ is given by
\begin{align*}
  &\mathcal{P}_s(\R^d) \\
  &\hspace{0.2cm}:= \biggl\{p(\cdot):p =\psi_s\circ\varphi \ \mbox{for some} \ \varphi\in \mathcal{C}_d \ \mbox{with} \ \mathrm{Im}(\varphi) \subseteq \mathcal{D}_s \cup \{\infty\}, \int_{\R^d}p =1 \biggr\}.
\end{align*}
\end{definition}
For $s=0$, the family of $s$-concave densities reduces to the family of log-concave densities.   Moreover, for $s_1 < s_2$, we have $\mathcal{P}_{s_2}(\R^d)\subseteq \mathcal{P}_{s_1}(\R^d)$ \cite[p.~86]{dharmadhikari1988unimodality}.  The $s$-concave density family introduces additional modelling flexibility, in particular allowing much heavier tails when $s < 0$ than the log-concave family, but we note that there is no guidance available in the literature on how to choose $s$.

For the problem of $s$-concave density estimation, we discuss two estimation methods, both of which have been previously considered in the literature, but for which there has been limited algorithmic development.  The first is based on the maximum likelihood principle (Section~\ref{sec:s-concave-MLE}), while the other is based on minimizing a R\'enyi divergence (Section~\ref{sec:s-concave-renyi}).  


\subsection{Computation of the \texorpdfstring{$s$}{s}-concave maximum likelihood estimator}\label{sec:s-concave-MLE}

\cite{seregin2010nonparametric} proved that a maximum likelihood estimator over $\mathcal{P}_s(\R^d)$ exists with probability one for $s\in(-1/d,\infty)$ and $n > \max\bigl(\frac{dr}{r-d},d\bigr)$, where $r := -1/s$, and does not exist if $s<-1/d$.  \cite{doss2016global} provide some statistical properties of this estimator when $d=1$.  The maximum likelihood estimation problem is to compute
\begin{equation}\label{eqn:MLE-psi-functional}
\hat{p}_n:=\argmax_{p\in\mathcal{P}_s(\R^d)} \sum_{i=1}^n\log p(\bx_i),
\end{equation}
or equivalently,
\begin{equation}\label{eqn:MLE-psi-varphi}
\argmax_{\varphi \in\mathcal{C}_d: \mathrm{Im}(\varphi) \subseteq \mathcal{D}_s\cup\{\infty\}}~\frac{1}{n}\sum_{i=1}^n\log \psi_s\circ\varphi(\bx_i) \quad\text{subject to}\quad \int_{\R^d}\psi_s\circ\varphi(\bx)\dx=1.    
\end{equation}
We establish the following theorem:
\begin{theorem}\label{thm:s-concave-MLE}
Let $s \in [0,1]$ and suppose that the convex hull $C_n$ of the data is $d$-dimensional (so that the $s$-concave MLE $\hat{p}_n$ exists and is unique). Then computing $\hat{p}_n$ in~\eqref{eqn:MLE-psi-functional} is equivalent to the convex minimization problem of computing
\begin{equation}
\bphi^* := \argmin_{\bphi = (\phi_1,\ldots,\phi_n) \in \mathcal{D}_s^n}\biggl\{-\frac{1}{n}\sum_{i=1}^n\log \psi_s(\phi_i)+ \int_{C_n}\psi_s\bigl(\cef[\bphi](\bx)\bigr)\, d\bx\biggr\},\label{eqn:MLE-psi-phi}\end{equation}
in the sense that $\hat{p}_n = \psi_s \circ \cef[\bphi^*]$.
\end{theorem}

\begin{remark}
The equivalence result in Theorem~\ref{thm:s-concave-MLE} holds for any $s$ (outside $[0,1]$) as long as the $s$-concave MLE exists. 
However, when $s\in[0,1]$, \eqref{eqn:MLE-psi-phi} is a convex optimization problem.  The family of $s$-concave densities with $s<0$ appears to be more useful from a statistical viewpoint as it allows for heavier tails than log-concave densities, but the MLE cannot be then computed via convex optimization.  Nevertheless, the entropy minimization methods discussed in Section~\ref{sec:s-concave-renyi} can be used to obtain $s$-concave density estimates for $s > -1$. 
\end{remark}



\subsection{Quasi-concave density estimation}\label{sec:s-concave-renyi}
Another route to estimate an $s$-concave density (or even a more general class) is via the following problem:
\begin{equation}\label{eqn:quasi-concave}\check\varphi := \argmin_{\varphi\in\mathcal{C}_d:\mathrm{dom}(\varphi) = C_n} \biggl\{\frac{1}{n}\sum_{i=1}^n\varphi(\bx_i)+\int_{C_n}\Psi\bigl(\varphi(\bx)\bigr)\dx\biggr\},\end{equation}
where $\Psi:\R \rightarrow (-\infty,\infty]$ is a decreasing, proper convex function.  When $\Psi(y) = e^{-y}$, \eqref{eqn:quasi-concave} is equivalent to the MLE for log-concave density estimation \eqref{Eq:MLE}, by \citet[Theorem~1]{cule2010maximum}. This problem, proposed by \cite{koenker2010quasi}, is called {\emph{quasi-concave}} density estimation. \citet[Theorem~4.1]{koenker2010quasi} show that under some assumptions on $\Psi$, there exists a solution to~\eqref{eqn:quasi-concave}, and if $\Psi$ is strictly convex, then the solution is unique. Furthermore, if $\Psi$ is differentiable on the interior of its domain, then the optimal solution to the dual of \eqref{eqn:quasi-concave} is a probability density $p$ such that $p=-\Psi'(\varphi)$, and the dual problem can be regarded as minimizing different distances or entropies (depending on~$\Psi$) between the empirical distribution of the data and $p$.  In particular, when $\beta \geq 1$ and $\Psi(y) = \mathbbm{1}_{\{y \leq 0\}}(-y)^\beta/\beta$, and when $\beta < 0$ and $\Psi(y) = -y^\beta/\beta$ for $y \geq 0$ (with $\Psi(y) = \infty$ otherwise), the dual problem of~\eqref{eqn:quasi-concave} is essentially minimizing the R\'enyi divergence and we have the primal-dual relationship  $p=|\varphi|^{\beta-1}$.  In fact, this amounts to estimating an $s$-concave density via R\'enyi divergence minimization with $\beta=1+1/s$ and $s \in (-1,\infty) \setminus \{0\}$.  We therefore consider the problem 
\begin{equation}\label{eqn:quasi-concave-s}
    \min_{\substack{\varphi\in\mathcal{C}_d:\mathrm{dom}(\varphi) = C_n\\\mathrm{Im}(\varphi) \subseteq \mathcal{D}_s}} \biggl\{\frac{1}{n}\sum_{i=1}^n\varphi(\bx_i)+\frac{1}{|1+1/s|}\int_{C_n}|\varphi(\bx)|^{1+1/s}\dx\biggr\}.
\end{equation}
The proof of Theorem~\ref{thm:quasi-concave} is similar to that of Theorem~\ref{thm:s-concave-MLE}, and is omitted for brevity. 
\begin{theorem}\label{thm:quasi-concave}
Given a decreasing proper convex function $\Psi$, the quasi-concave density estimation problem \eqref{eqn:quasi-concave} is equivalent to the following convex problem: 
\begin{equation}\bphi^* := \argmin_{\bphi\in\mathcal{D}_s^n}\biggl\{\frac{1}{n}\bm{1}^\top\bphi+ \int_{C_n}\Psi\bigl(\cef[\bphi](\bx)\bigr)\dx\biggr\},\label{eqn:quasi-concave-phi}\end{equation}
in the sense that $\check{\varphi} = \cef[\bphi^*]$, with corresponding density estimator $\tilde{p}_n = -\Psi' \circ \cef[\bphi^*]$.  
\end{theorem}
The objective in~\eqref{eqn:quasi-concave-phi} is convex, so our computational framework can be applied to solve this problem.

\section{Computational experiments on simulated data}\label{sec:compute}
In this section, we present numerical experiments to study the different variants of our algorithm
and compare them with existing methods based on convex optimization for the log-concave MLE.  Our results are based on large-scale synthetic datasets with $n \in \{5{,}000,10{,}000\}$ observations generated from standard $d$-dimensional normal and Laplace distributions with $d=4$. 
Code for our experiments is available from the github repository \texttt{LogConcComp} available at: 
\begin{center}
\url{https://github.com/wenyuC94/LogConcComp}. 
\end{center}

All computations were carried out on the MIT Supercloud Cluster \citep{reuther2018interactive} on an Intel Xeon Platinum 8260 machine, with 24 CPUs and 24GB of RAM. Our algorithms were written in Python; we used Gurobi~\citep{gurobi} to solve the LPs and QPs.

Our first comparison method is that of \cite{cule2010maximum}, implemented in the \texttt{R} package \texttt{LogConcDEAD} \citep{JSSv029i02}, and denoted by CSS.  The CSS algorithm terminates when $\|\bphi^{(t)} - \bphi^{(t-1)}\|_\infty \leq \tau$, and we consider $\tau \in \{10^{-2},10^{-3},10^{-4}\}$.   Our other competing approach is the randomized smoothing method of \cite{duchi2012randomized}, with random grids of a fixed grid size, which we denote here by RS-RF-$m$, with $m$ being the grid size.  To the best of our knowledge, this method has not been used to compute the log-concave MLE previously.

We denote the different variants of our algorithm as $\text{Alg-$V$}$, where $\text{Alg}\in\{\text{RS,NS}\}$ represents Algorithm~\ref{algo:smoothing} with Randomized smoothing and Nesterov smoothing, and $V\in\{\text{DI},\text{RI}\}$ represents whether we use deterministic or random grids of increasing grid sizes to approximate the gradient. Further details of our input parameters are given in Appendix~\ref{app:expt-setting}.

Figure~\ref{fig:profile} presents the relative objective error, defined for an algorithm with iterates $\bphi_1,\ldots,\bphi_t$ as
\begin{equation}
\label{Eq:RelObj}
\texttt{relobj}(t) := \biggl|\frac{\min_{s\in[t]}f(\bphi_s)-f(\bphi^*)}{f(\bphi^*)}\biggr|,
\end{equation}
against time (in minutes) and number of iterations.  In the definition of the relative objective error in~\eqref{Eq:RelObj} above, $\bphi^*$ is taken as the CSS solution with tolerance $\tau=10^{-4}$.  The figure shows that randomized smoothing appears to outperform Nesterov smoothing in terms of the time taken to reach a desired relative objective error, since the former solves an LP~\eqref{eqn:cef-LP}, whereas the latter has to solve a QP~\eqref{eqn:h-QP}; the number of iterations taken by the different methods is, however, similar.  There is no clear winner between randomized and deterministic grids, and both appear to perform well.

Table~\ref{table:obj} compares our proposed methods against the CSS solutions with different tolerances, in terms of running time, final objective function, and distances of the algorithm outputs to the optimal solution~$\bphi^*$ and the truth $\bphi^{\text{truth}}$.  We find that all of our proposals yield marked improvements in running time compared with the CSS solution: with $n=10{,}000$, $d=4$ and $\tau = 10^{-4}$, CSS takes more than 20 hours for all of the data sets we considered, whereas the RS-DI variant is around 50 times faster.  The CSS solution may have a slightly improved objective function value on termination, but as shown in Table~\ref{table:obj}, all of our algorithms achieve an optimization error that is small by comparison with the statistical error, and from a statistical perspective, there is no point in seeking to reduce the optimization error further than this.  Table~\ref{table:stats-dist} shows that the distances $\|\bphi^* - \bphi^{\mathrm{truth}}\|/n^{1/2}$ are well concentrated around their means (i.e.~do not vary greatly over different random samples drawn from the underlying distribution), which provides further reassurance that our solutions are sufficiently accurate for practical purposes.   On the other hand, the CSS solution with tolerance $10^{-3}$ is not always sufficiently reliable in terms of its statistical accuracy, e.g.~for a Laplace distribution with $n=5{,}000$.  Our further experiments on real data sets reported in Appendix~\ref{subsec:additional-expts} provide qualitatively similar conclusions.  


Finally, Figure~\ref{fig:profile-RS} compares our proposed multistage increasing grid sizes (RS-DI/RS-RI) (see Tables~5 and~6) with the fixed grid size (RS-RF) proposed by \cite{duchi2012randomized}, under the randomized smoothing setting.  We see that the benefits of using the increasing grid sizes as described by our theory carry over to improved practical performance, both in terms of running time and number of iterations.

\begin{table}[t!]
\centering
\caption{Comparison of our proposed methods with the CSS solution \citep{cule2010maximum} and RS-RF \citep{duchi2012randomized}.  On a single dataset, we ran 5 repetitions of each algorithm with different random seeds and report the median statistics.  Here, \texttt{obj} and \texttt{relobj} denote the objective and relative objective error, respectively, \texttt{runtime} denotes the running time (in minutes), \texttt{dopt} and \texttt{dtruth} denote the (Euclidean) distances between the algorithm outputs and the optimal solution and the truth, respectively, \texttt{iter} denotes the number of iterations, \texttt{tO} denotes the total number of oracles (grid points), \texttt{aO} denotes the average number of oracles (grid points) per iteration, and \texttt{hO} denotes the harmonic average of grid sizes (which equals $T/(M_T^{(1)})^2$).  For CSS, \texttt{param} is the tolerance $\tau$; for RS-RF, \texttt{param} is the (fixed) grid size $m$. Here `-' means the running time of the corresponding algorithm exceeded 20 hours.}
\label{table:obj}
Normal, $n=5{,}000,d=4$\\
\resizebox{0.65\textwidth}{!}{\begin{tabular}{rrrrrrrrrrr}
\toprule
 algo & \texttt{param} &    \texttt{obj} & \texttt{relobj} & \texttt{runtime} &   \texttt{dopt} & \texttt{dtruth} & \texttt{iter} & \texttt{tO} & \texttt{aO} & \texttt{hO}\\
\midrule
   \multirow{3}{*}{CSS} &  1e-2 & 6.5209 & 1.1e-03 &   10.15 & 0.1955 & 0.2788 &        &        &        \\
   &  1e-3 & 6.5146 & 9.8e-05 &  110.04 & 0.0612 & 0.2465 &        &        &        \\
   &  1e-4 & 6.5140 & 0.0e-00 &  829.55 & 0.0000 & 0.2454 &        &        &        \\\midrule
RS-DI &  None & 6.5144 & 7.0e-05 &   16.04 & 0.0227 & 0.2499 & 128 & 6.18M & 48.31K & 20.23K \\
RS-RI &  None & 6.5150 & 1.6e-04 &   31.05 & 0.0289 & 0.2502 & 128 & 6.88M & 53.75K & 32.00K \\
NS-DI &  None & 6.5144 & 7.1e-05 &   89.94 & 0.0259 & 0.2497 & 128  &  6.18M & 48.31K & 20.23K \\
NS-RI &  None & 6.5149 & 1.5e-04 &  102.23 & 0.0312 & 0.2502  & 128 &  6.88M & 53.75K & 32.00K \\\midrule
 \multirow{5}{*}{RS-RF}&  5000& 6.5174 & 5.2e-04 &   30.22 & 0.0575 & 0.2552 & 1024& 5.12M &  5.00K &  5.00K \\
 &10000 & 6.5168 & 4.4e-04 &   25.48 & 0.0429 & 0.2508 & 512 & 5.12M & 10.00K & 10.00K \\
  & 20000 & 6.5164 & 3.7e-04 &   23.05 & 0.0552 & 0.2567 & 256&  5.12M & 20.00K & 20.00K \\
 & 40000 & 6.5158 & 2.9e-04 &   21.31 & 0.0344 & 0.2496 & 128&  5.12M & 40.00K & 40.00K \\
 & 80000 & 6.5150 & 1.6e-04 &   42.25 & 0.0288 & 0.2499 & 128& 10.24M & 80.00K & 80.00K \\
\bottomrule
\end{tabular}}\\
Laplace, $n=5{,}000,d=4$\\
\resizebox{0.65\textwidth}{!}{\begin{tabular}{rrrrrrrrrrr}
\toprule
 algo & \texttt{param} &    \texttt{obj} & \texttt{ relobj} & \texttt{runtime} &   \texttt{dopt} & \texttt{dtruth} & \texttt{iter}& \texttt{tO} & \texttt{aO}& \texttt{hO}\\
\midrule
   \multirow{3}{*}{CSS} &  1e-2 & 7.9183 & 3.0e-02 &   30.77 & 2.5100 & 2.5449 &        &        &        \\
   &  1e-3 & 7.6994 & 1.1e-03 &  387.34 & 0.5985 & 0.6514 &        &        &        \\
   &  1e-4 & 7.6908 & 0.0e-00 & 1007.80 & 0.0000 & 0.2552 &        &        &        \\\midrule
RS-DI &  None & 7.6988 & 1.0e-03 &   17.64 & 0.0592 & 0.2304 & 128&  7.24M & 56.54K & 34.84K \\
RS-RI &  None & 7.6939 & 4.0e-04 &   31.32 & 0.0424 & 0.2632 & 128& 6.88M & 53.75K & 32.00K \\
NS-DI &  None & 7.6989 & 1.0e-03 &  109.68 & 0.0640 & 0.2259 &  128& 7.24M & 56.54K & 34.84K \\
NS-RI &  None & 7.6943 & 4.5e-04 &  107.57 & 0.0362 & 0.2601 &  128& 6.88M & 53.75K & 32.00K \\\midrule
 \multirow{5}{*}{RS-RF} &  5000 & 7.7048 & 1.8e-03 &   31.43 & 0.0628 & 0.2732 & 1024&  5.12M &  5.00K &  5.00K \\
& 10000 & 7.7059 & 2.0e-03 &   27.33 & 0.0621 & 0.2745 & 512& 5.12M & 10.00K & 10.00K \\
 & 20000 & 7.6986 & 1.0e-03 &   24.53 & 0.0527 & 0.2696 &  256&5.12M & 20.00K & 20.00K \\
 & 40000 & 7.6979 & 9.2e-04 &   22.68 & 0.0852 & 0.2776 &  128&5.12M & 40.00K & 40.00K \\
& 80000 & 7.6949 & 5.4e-04 &   43.27 & 0.0349 & 0.2614 & 128& 10.24M & 80.00K & 80.00K\\
\bottomrule
\end{tabular}}\\
Normal, $n=10{,}000,d=4$\\
\resizebox{0.65\textwidth}{!}{\begin{tabular}{rrrrrrrrrrr}
\toprule
 algo & \texttt{param} &    \texttt{obj} & \texttt{ relobj} & \texttt{runtime} &   \texttt{dopt} & \texttt{dtruth} & \texttt{iter}& \texttt{tO} & \texttt{aO} & \texttt{hO}\\
\midrule
   \multirow{3}{*}{CSS} &  1e-2 & 6.5634 & 4.1e-04 &   24.72 & 0.1018 & 0.1911 &        &        &        \\
   &  1e-3 & 6.5612 & 7.3e-05 &  181.01 & 0.0462 & 0.1854 &        &        &        \\
   &  1e-4 & 6.5607 & 0.0e-00 & - & 0.0000 & 0.1859 &        &        &        \\\midrule
RS-DI &  None & 6.5621 & 2.1e-04 &   35.40 & 0.0411 & 0.1939 & 128&  6.67M & 52.14K & 21.85K \\
RS-RI &  None & 6.5626 & 3.0e-04 &   65.18 & 0.0443 & 0.1939 & 128& 6.88M & 53.75K & 32.00K \\
NS-DI &  None & 6.5620 & 2.0e-04 &  207.32 & 0.0429 & 0.1953 & 128& 6.67M & 52.14K & 21.85K \\
NS-RI &  None & 6.5625 & 2.8e-04 &  215.51 & 0.0452 & 0.1959 & 128& 6.88M & 53.75K & 32.00K \\\midrule
\multirow{5}{*}{RS-RF} &  {5000} & 6.5690 & 1.3e-03 &   64.80 & 0.1097 & 0.2205 &  1024& 5.12M &  5.00K &  5.00K \\
 & {10000} & 6.5704 & 1.5e-03 &   56.26 & 0.0470 & 0.1854 & 512& 5.12M & 10.00K & 10.00K \\
& {20000} & 6.5656 & 7.5e-04 &   50.17 & 0.0412 & 0.1890 & 256& 5.12M & 20.00K & 20.00K \\
 & {40000} & 6.5638 & 4.7e-04 &   46.24 & 0.0478 & 0.1948 & 128& 5.12M & 40.00K & 40.00K \\
 & 80000 & 6.5627 & 3.0e-04 &   89.52 & 0.0446 & 0.1948 & 128&10.24M & 80.00K & 80.00K\\
\bottomrule
\end{tabular}}\\
Laplace, $n=10{,}000,d=4$\\
\resizebox{0.65\textwidth}{!}{\begin{tabular}{rrrrrrrrrrr}
\toprule
 algo & \texttt{param} &    \texttt{obj} & \texttt{ relobj} & \texttt{runtime} &   \texttt{dopt} & \texttt{dtruth} & {\texttt{iter}} & {\texttt{tO}} & {\texttt{aO}} & {\texttt{hO}}\\
\midrule
   \multirow{3}{*}{CSS} &  1e-2 & 8.1796 & 5.8e-02 &   57.46 & 3.9044 & 3.9328 &        &        &        \\
   &  1e-3 & 7.7327 & 4.6e-04 & - & 0.3470 & 0.4081 &        &        &        \\
   &  1e-4 & 7.7292 & 0.0e-00 & - & 0.0000 & 0.2025 &        &        &        \\\midrule
RS-DI &  None & 7.7401 & 1.4e-03 &   42.70 & 0.0886 & 0.1825 & 128& 8.14M & 63.60K & 39.20K \\
RS-RI &  None & 7.7370 & 1.0e-03 &   65.67 & 0.0753 & 0.2295 & 128& 6.88M & 53.75K & 32.00K \\
NS-DI &  None & 7.7399 & 1.4e-03 &  263.40 & 0.0801 & 0.1724 & 128& 8.14M & 63.60K & 39.20K \\
NS-RI &  None & 7.7365 & 9.4e-04 &  225.80 & 0.0480 & 0.2159 & 128& 6.88M & 53.75K & 32.00K \\\midrule
\multirow{5}{*}{RS-RF} &  {5000} & 7.7541 & 3.2e-03 &   64.93 & 0.1051 & 0.2499 & 1024& 5.12M &  5.00K &  5.00K \\
 & { 10000} & 7.7543 & 3.2e-03 &   57.87 & 0.0918 & 0.2435 & 512& 5.12M & 10.00K & 10.00K \\
& {20000} & 7.7468 & 2.3e-03 &   51.25 & 0.1157 & 0.2529 & 256& 5.12M & 20.00K & 20.00K \\
& {40000} & 7.7511 & 2.8e-03 &   46.31 & 0.0907 & 0.2423 & 128& 5.12M & 40.00K & 40.00K \\
 & 80000 & 7.7378 & 1.1e-03 &   89.13 & 0.0529 & 0.2230 & 128&10.24M & 80.00K & 80.00K \\
\bottomrule
\end{tabular}}
\end{table}

\begin{table}[ht]
    \centering
    \caption{Statistics of the distance between the optimal solution and truth. For each type of data set, we drew 40 random samples of the sizes given, and computed the log-concave MLE by CSS with tolerance $10^{-4}$.}
    \label{table:stats-dist}
    \scalebox{0.9}{
    \begin{tabular}{lrrrrrrr}
\toprule
{} &   mean &    std.error &    min &    25\% &    50\% &    75\% &    max \\
\midrule
normal ($n=5{,}000$)   & 0.2565 & 0.0093 & 0.2415 & 0.2480 & 0.2574 & 0.2629 & 0.2745 \\
Laplace ($n=5{,}000$)  & 0.2590 & 0.0114 & 0.2366 & 0.2508 & 0.2578 & 0.2676 & 0.2825 \\
\bottomrule
\end{tabular}}
\end{table}

\begin{figure}[ht]
    \centering
\resizebox{0.92\textwidth}{!}{\begin{tabular}{r c c}
\multicolumn{3}{c}{ {Objective Profiles for Normal, $n=5{,}000$, $d=4$ }}\\
\rotatebox{90}{  { {~~~~~~~~~~~~~$\texttt{relobj}$}}}&\includegraphics[width =0.48 \textwidth,trim = 1cm 0cm 2.5cm 2cm,clip ]{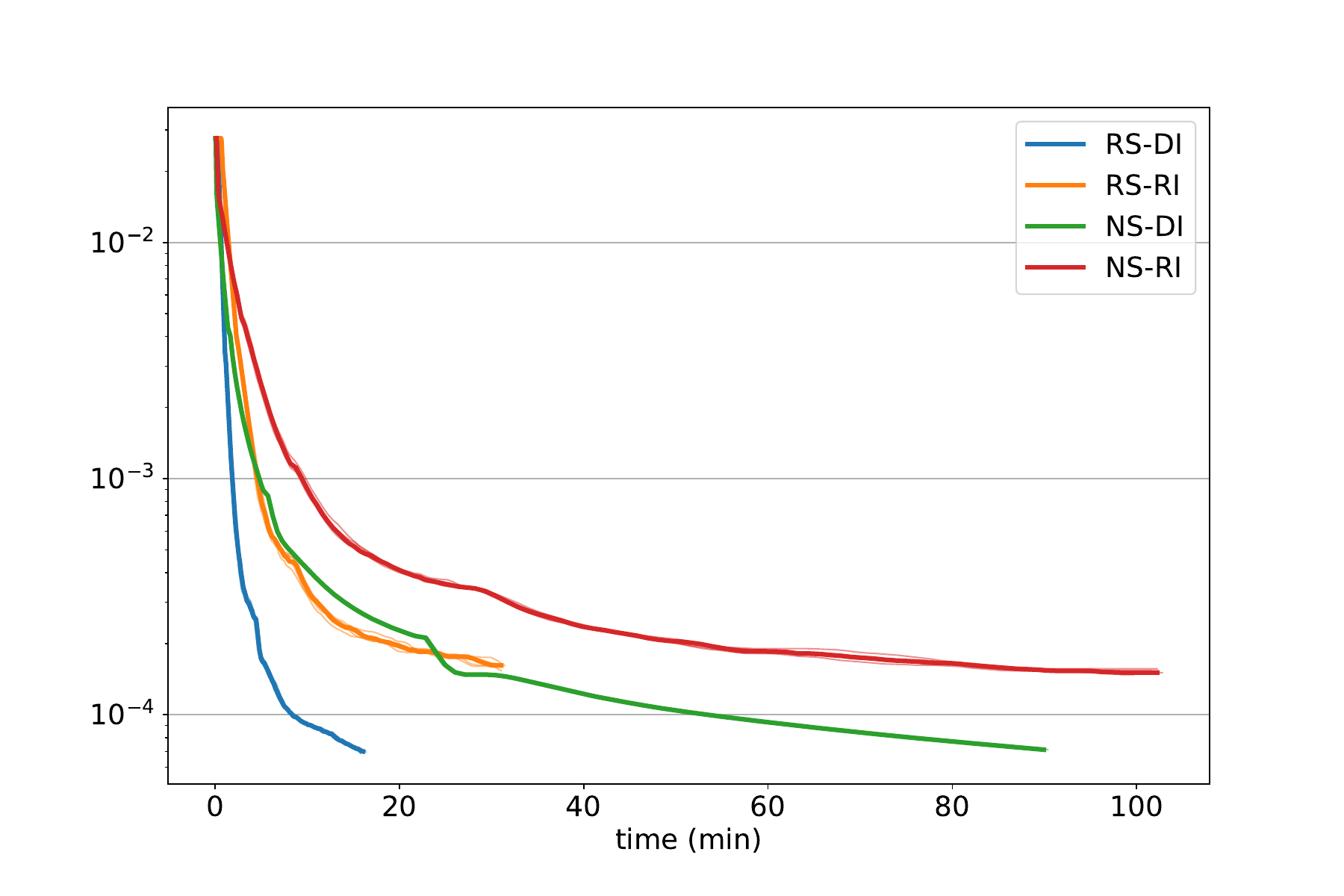}&\includegraphics[width =0.48 \textwidth,trim = 1cm 0cm 2.5cm 2cm,clip ]{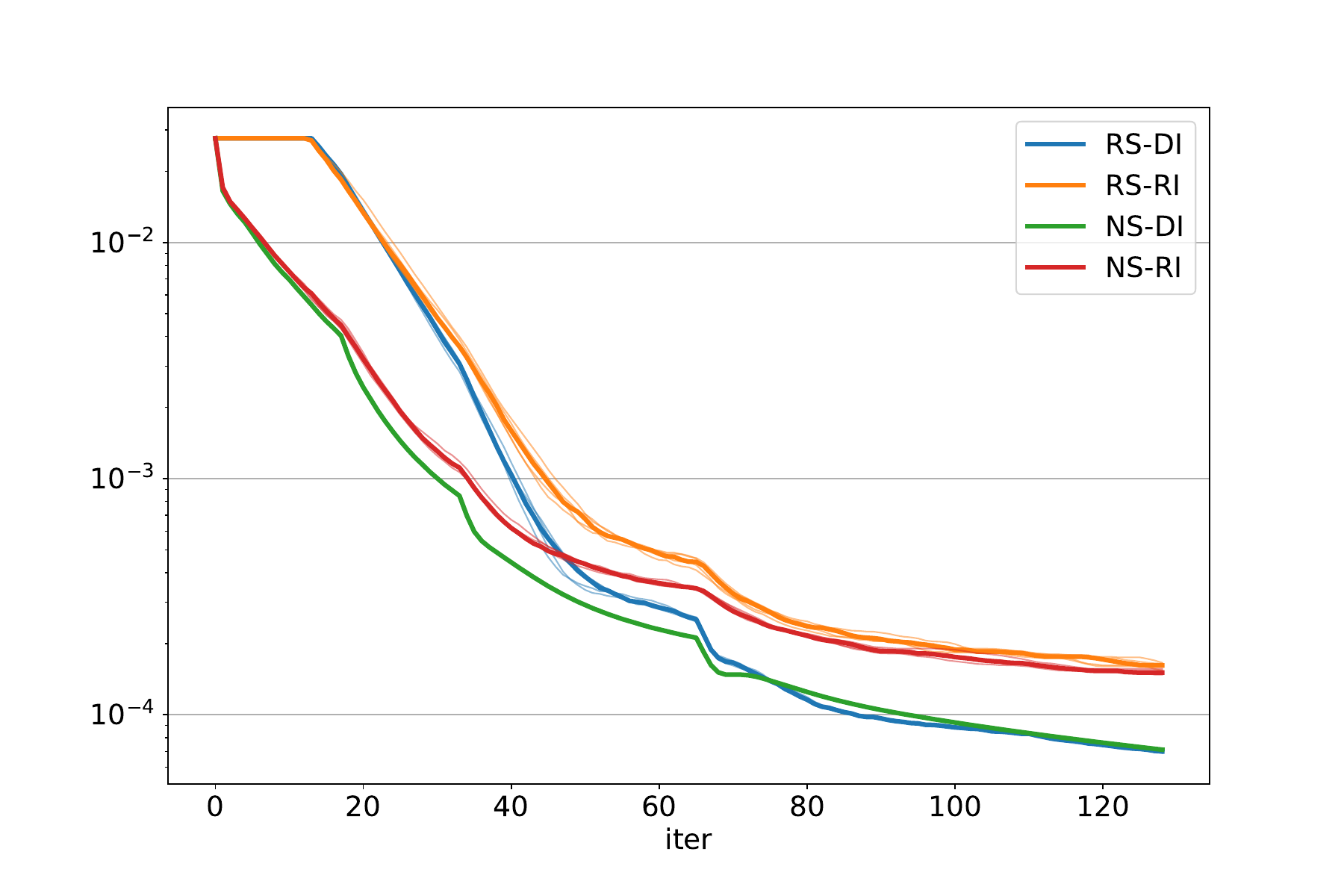}\\
 \multicolumn{3}{c}{ {Objective Profiles for Laplace, $n=5{,}000$, $d=4$}} \\
\rotatebox{90}{  { {~~~~~~~~~~~~~$\texttt{relobj}$}}}&\includegraphics[width =0.48 \textwidth,trim = 1cm 0cm 2.5cm 2cm,clip ]{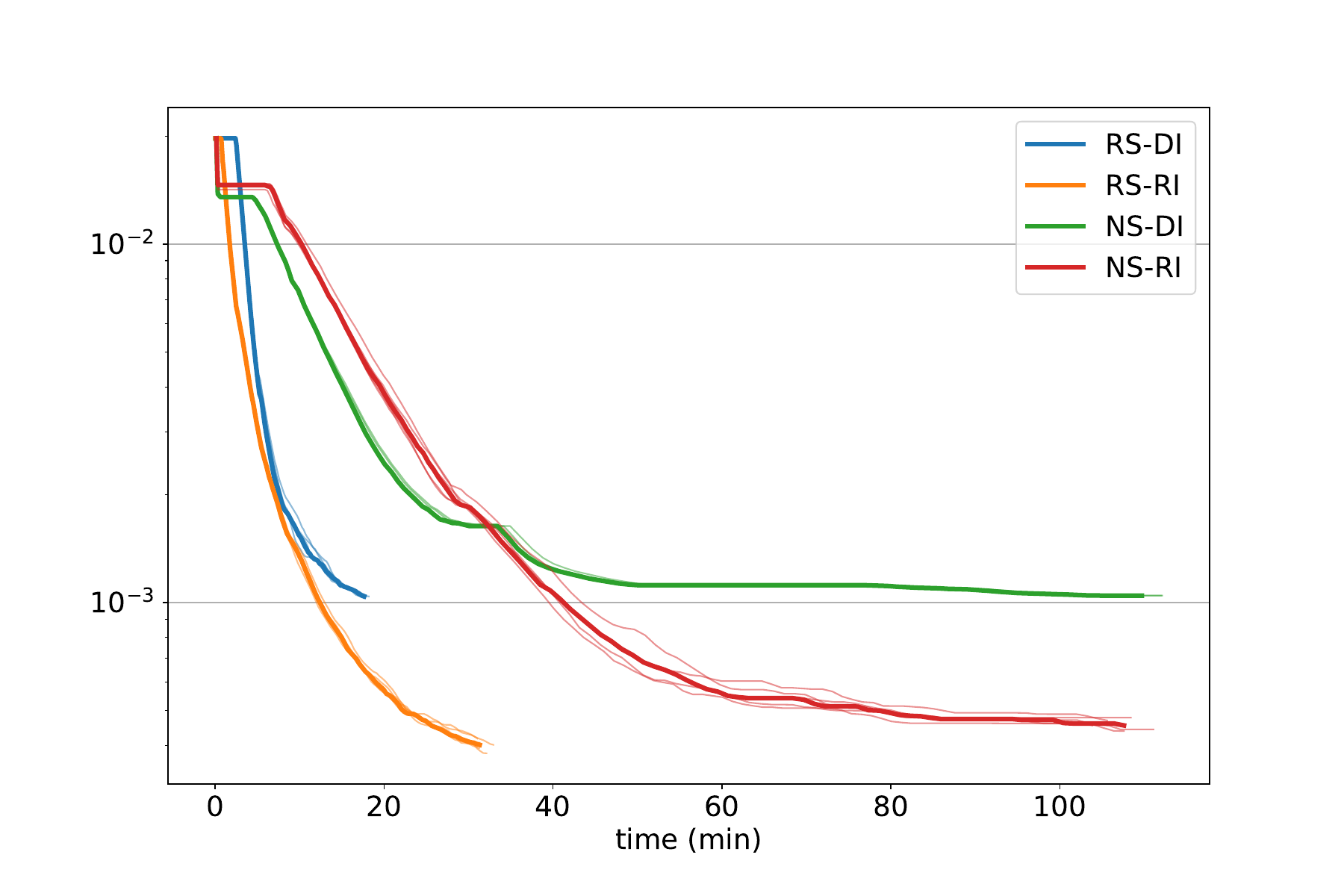}&\includegraphics[width =0.48 \textwidth,trim = 1cm 0cm 2.5cm 2cm,clip ]{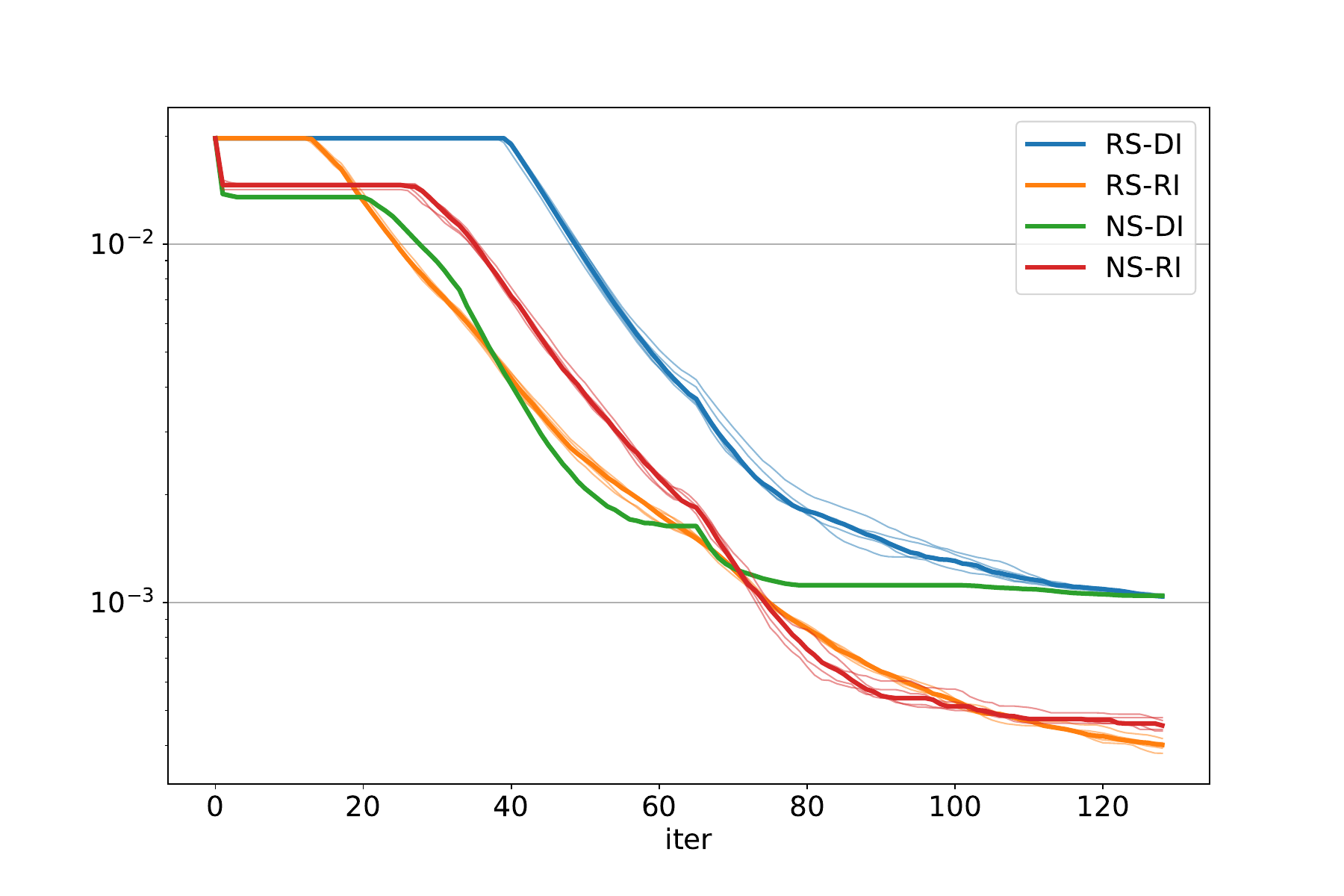}\\
 \multicolumn{3}{c}{ {Objective Profiles for Normal, $n=10{,}000$, $d=4$}} \\
\rotatebox{90}{  { {~~~~~~~~~~~~~$\texttt{relobj}$}}}&\includegraphics[width =0.48 \textwidth,trim = 1cm 0cm 2.5cm 2cm,clip ]{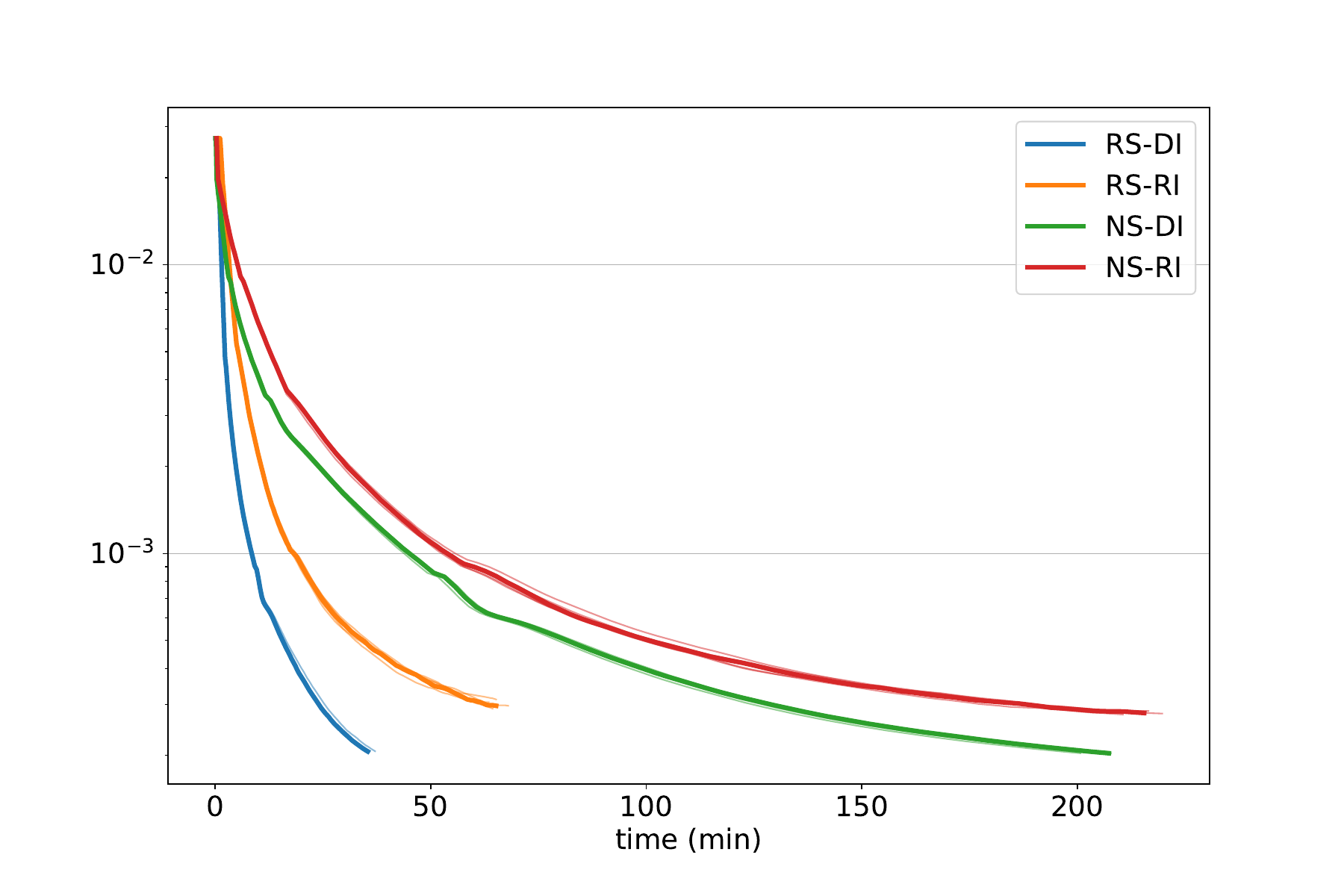}&\includegraphics[width =0.48 \textwidth,trim = 1cm 0cm 2.5cm 2cm,clip ]{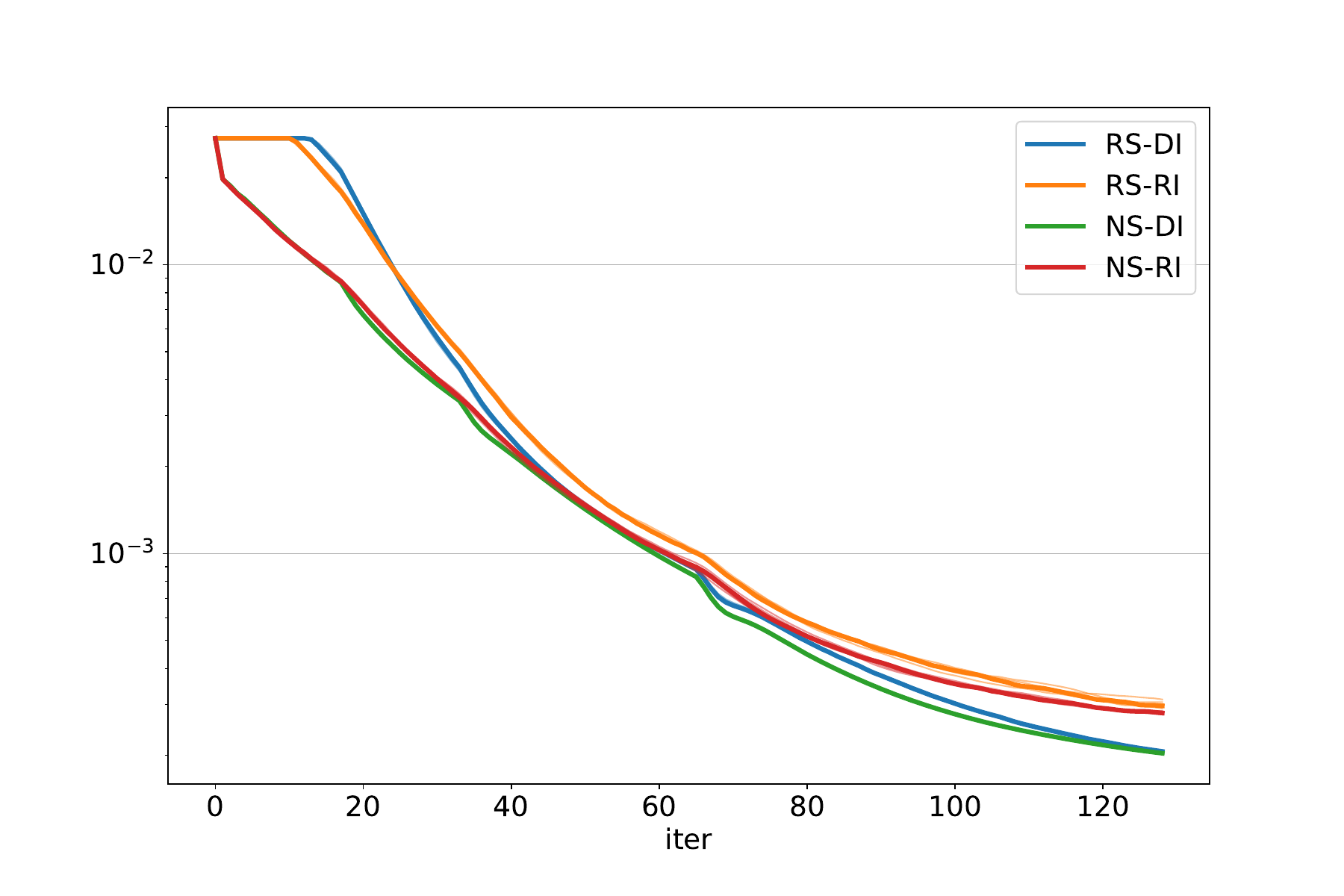}\\
 \multicolumn{3}{c}{ {Objective Profiles for Laplace, $n=10{,}000$, $d=4$}} \\
\rotatebox{90}{  { {~~~~~~~~~~~~~$\texttt{relobj}$}}}&\includegraphics[width =0.48 \textwidth,trim = 1cm 0cm 2.5cm 2cm,clip ]{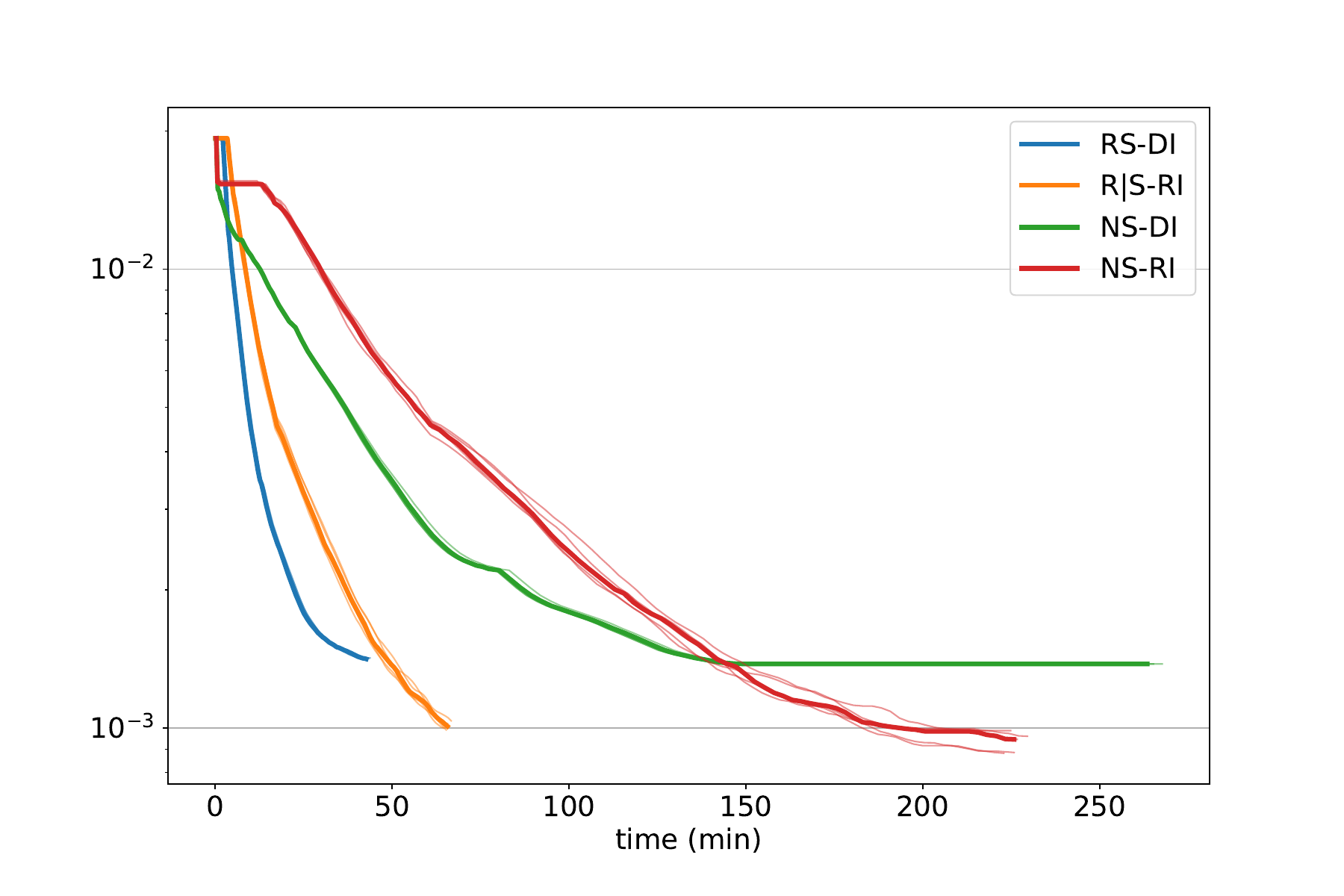}&\includegraphics[width =0.48 \textwidth,trim = 1cm 0cm 2.5cm 2cm,clip ]{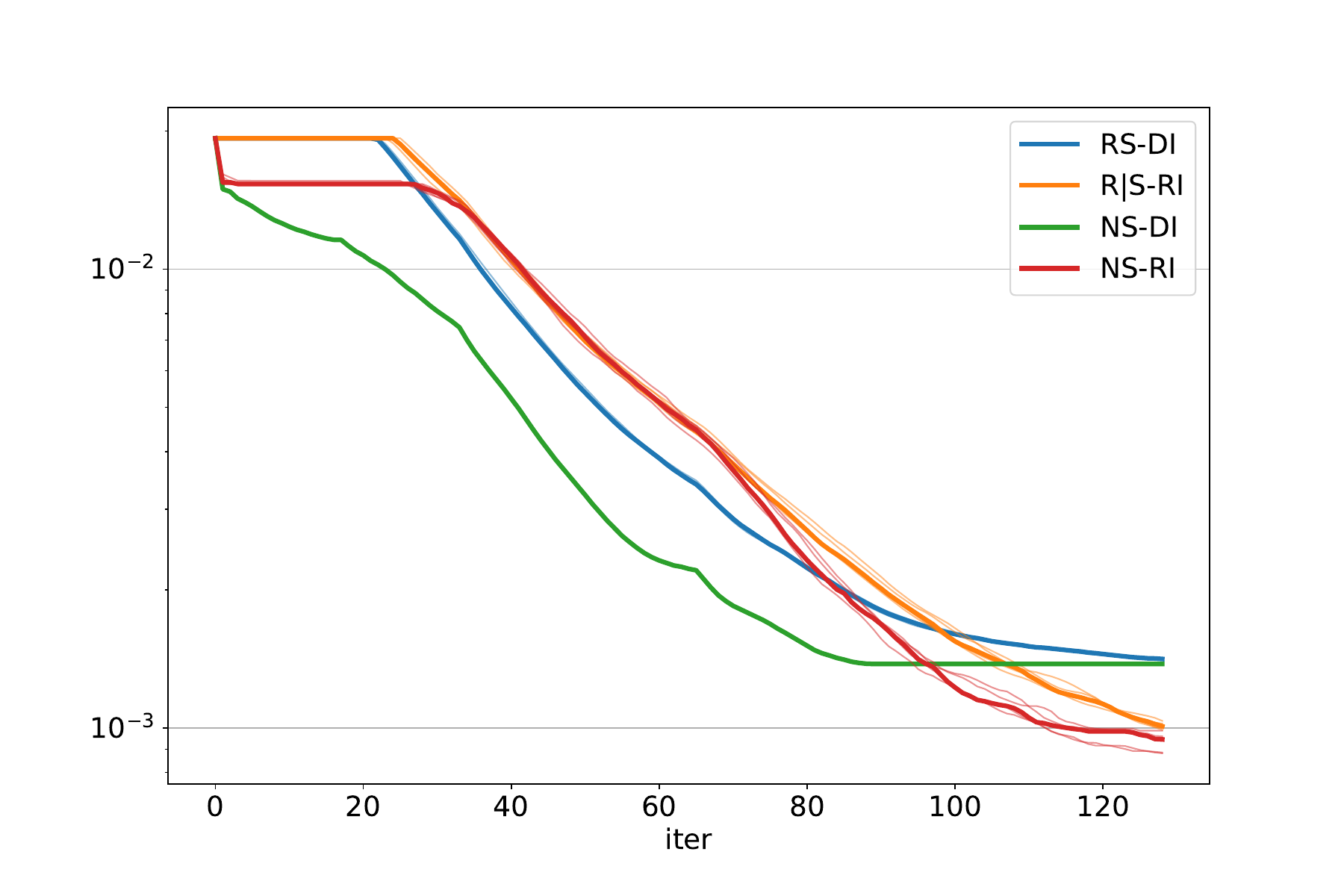}
    \end{tabular}
    }
    \caption{{\small{Plots on a log-scale of Relative Objective 
    versus time (mins) [left panel] and number of iterations [right panel].  For each of our four synthetic data sets, we ran five repetitions of each algorithm, so each bold line corresponds to the median of the profiles of the corresponding algorithm, and each thin line corresponds to the profile of one repetition.  For the right panel, we show the profiles up to 128 iterations.
}}}
    \label{fig:profile}
\end{figure}

\begin{figure}[ht]
    \centering
\resizebox{0.95\textwidth}{!}{\begin{tabular}{r c c}
\multicolumn{3}{c}{ {Objective Profiles for Normal, $n=5{,}000$, $d=4$ }}\\
\rotatebox{90}{  { {~~~~~~~~~~~~~$\texttt{relobj}$}}}&\includegraphics[width =0.48 \textwidth,trim = 1cm 0cm 2.5cm 2cm,clip ]{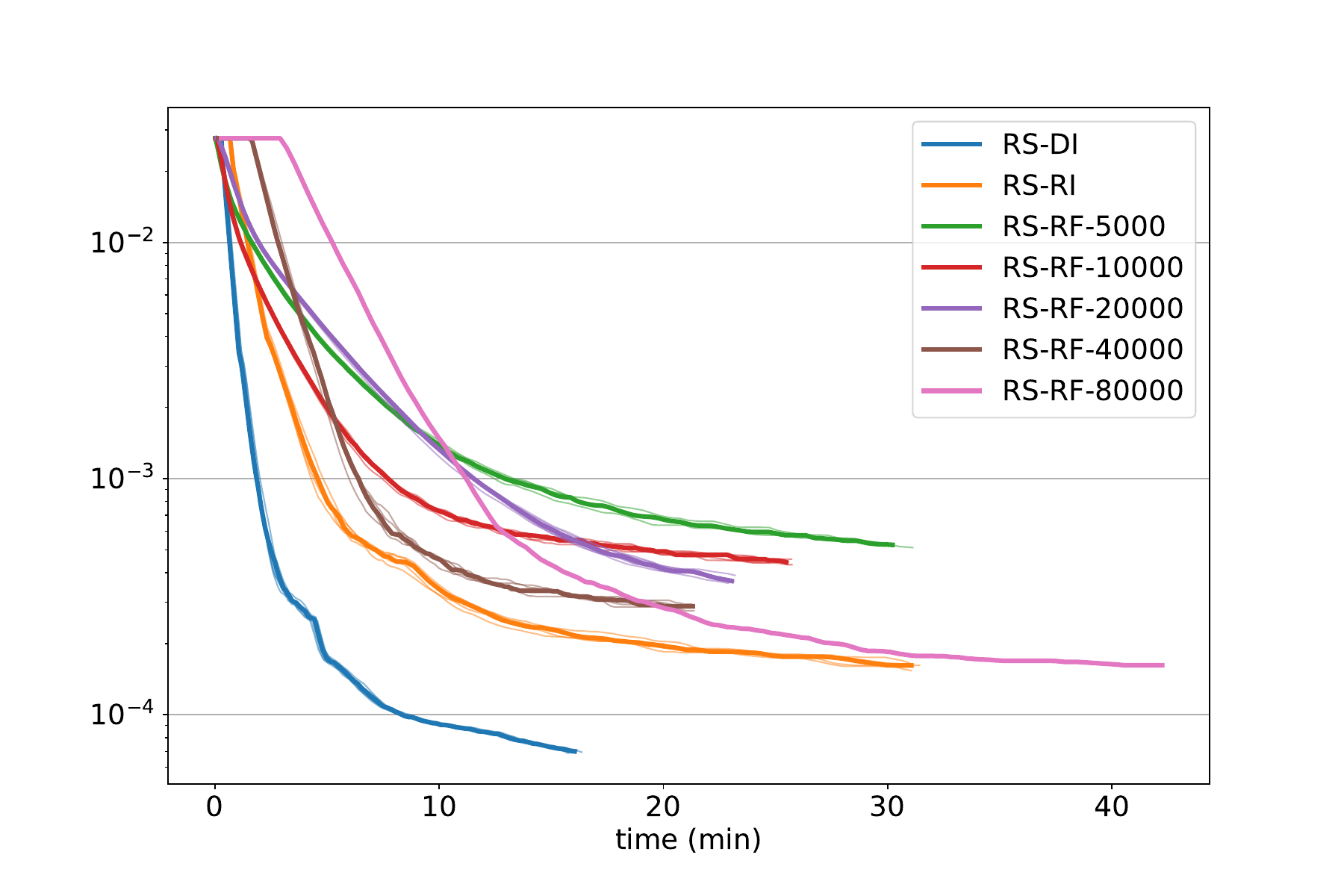}&\includegraphics[width =0.48 \textwidth,trim = 1cm 0cm 2.5cm 2cm,clip ]{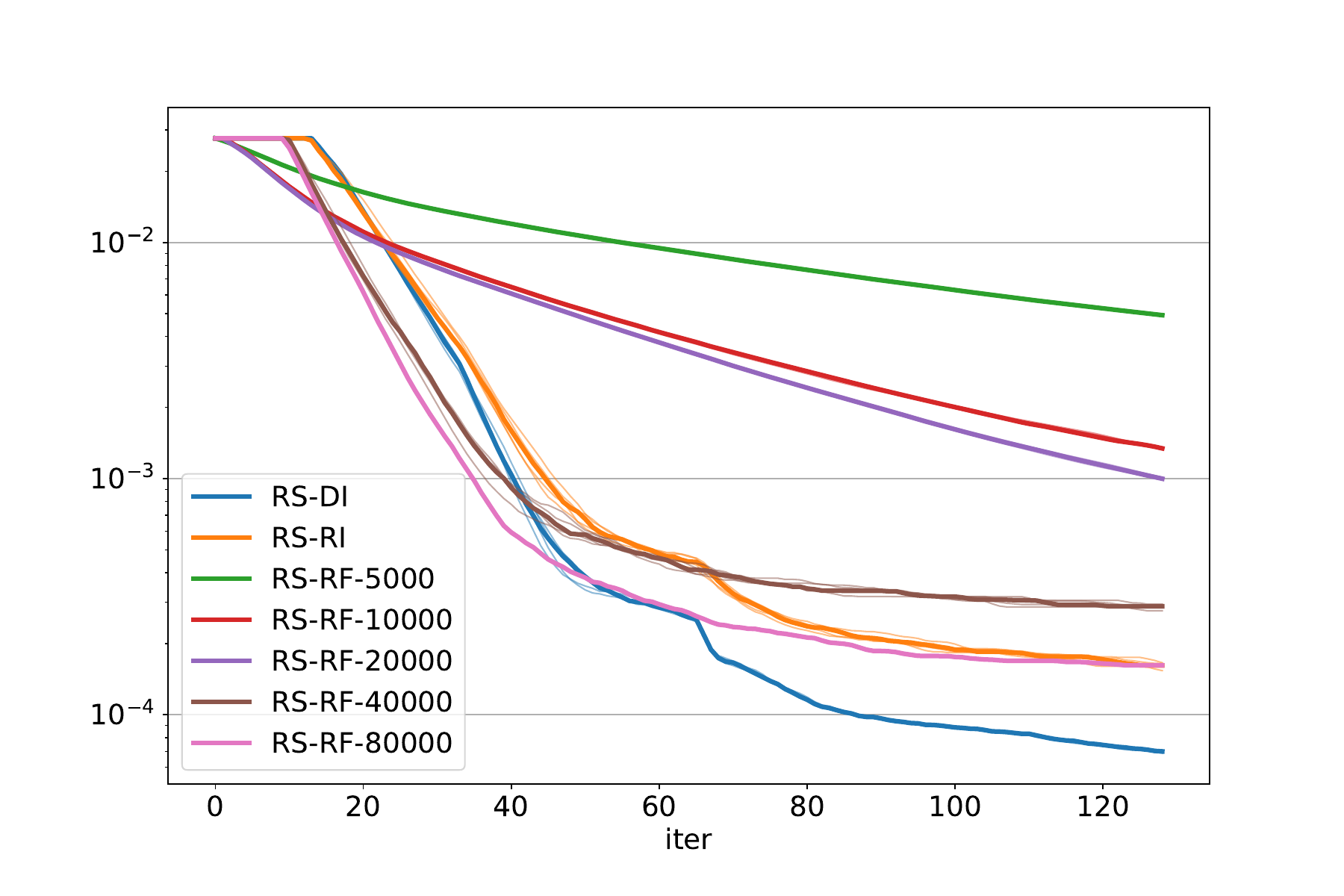}\\
\multicolumn{3}{c}{ {Objective Profiles for Laplace, $n=5{,}000$, $d=4$ }}\\
\rotatebox{90}{  { {~~~~~~~~~~~~~$\texttt{relobj}$}}}&\includegraphics[width =0.48 \textwidth,trim = 1cm 0cm 2.5cm 2cm,clip ]{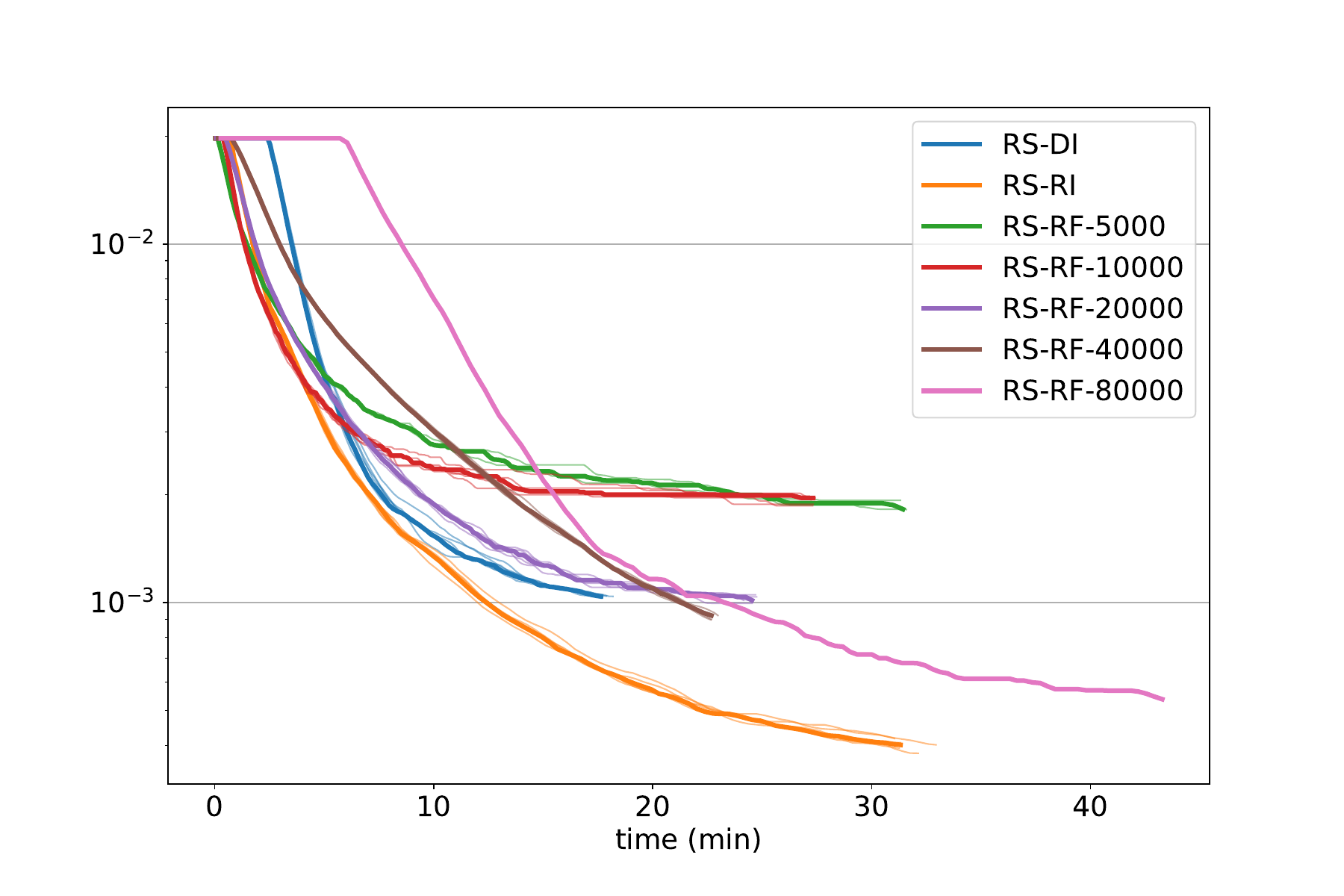}&\includegraphics[width =0.48 \textwidth,trim = 1cm 0cm 2.5cm 2cm,clip ]{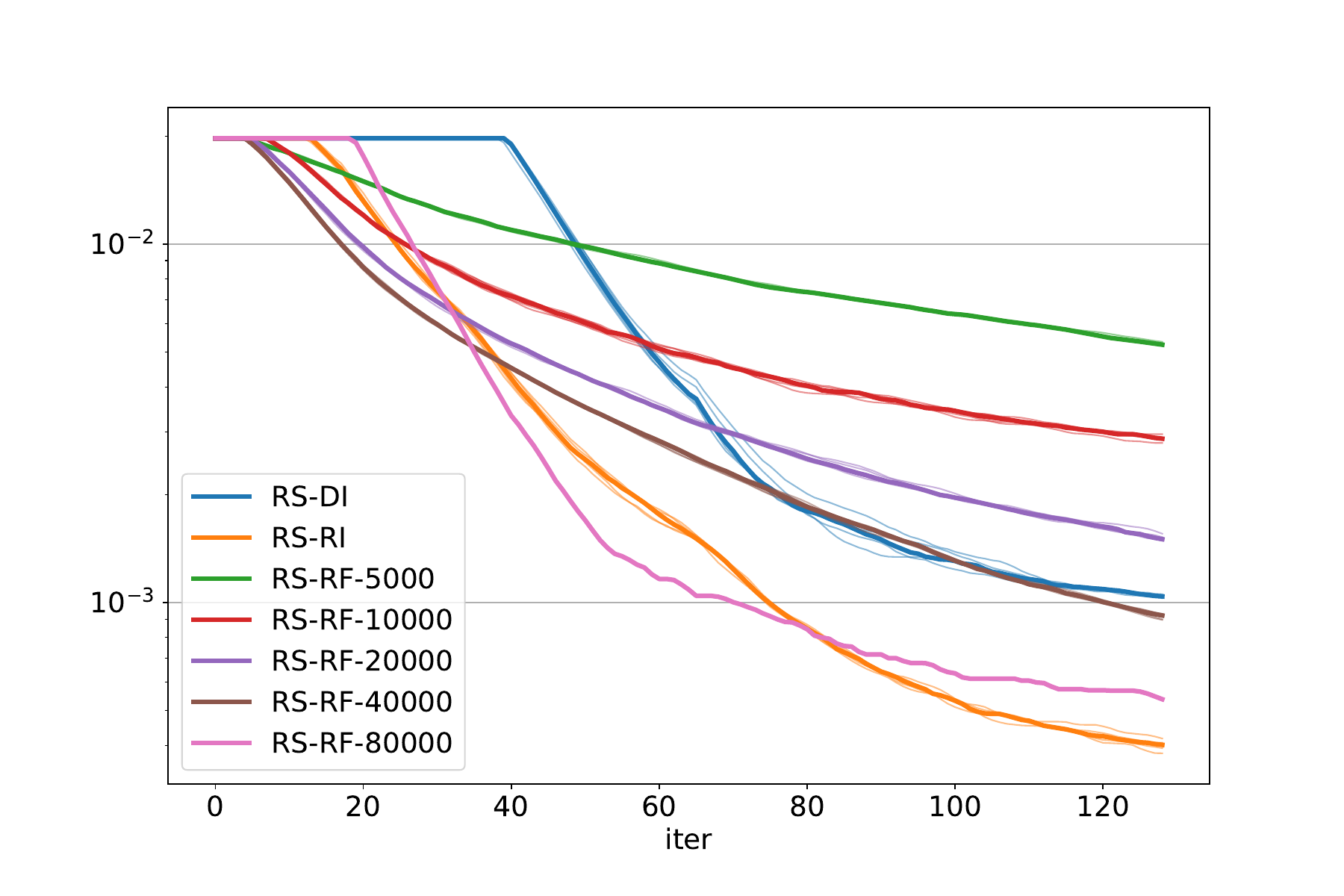}
\\
 \multicolumn{3}{c}{ {Objective Profiles for Normal, $n=10{,}000$, $d=4$}} \\
\rotatebox{90}{  { {~~~~~~~~~~~~~$\texttt{relobj}$}}}&\includegraphics[width =0.48 \textwidth,trim = 1cm 0cm 2.5cm 2cm,clip ]{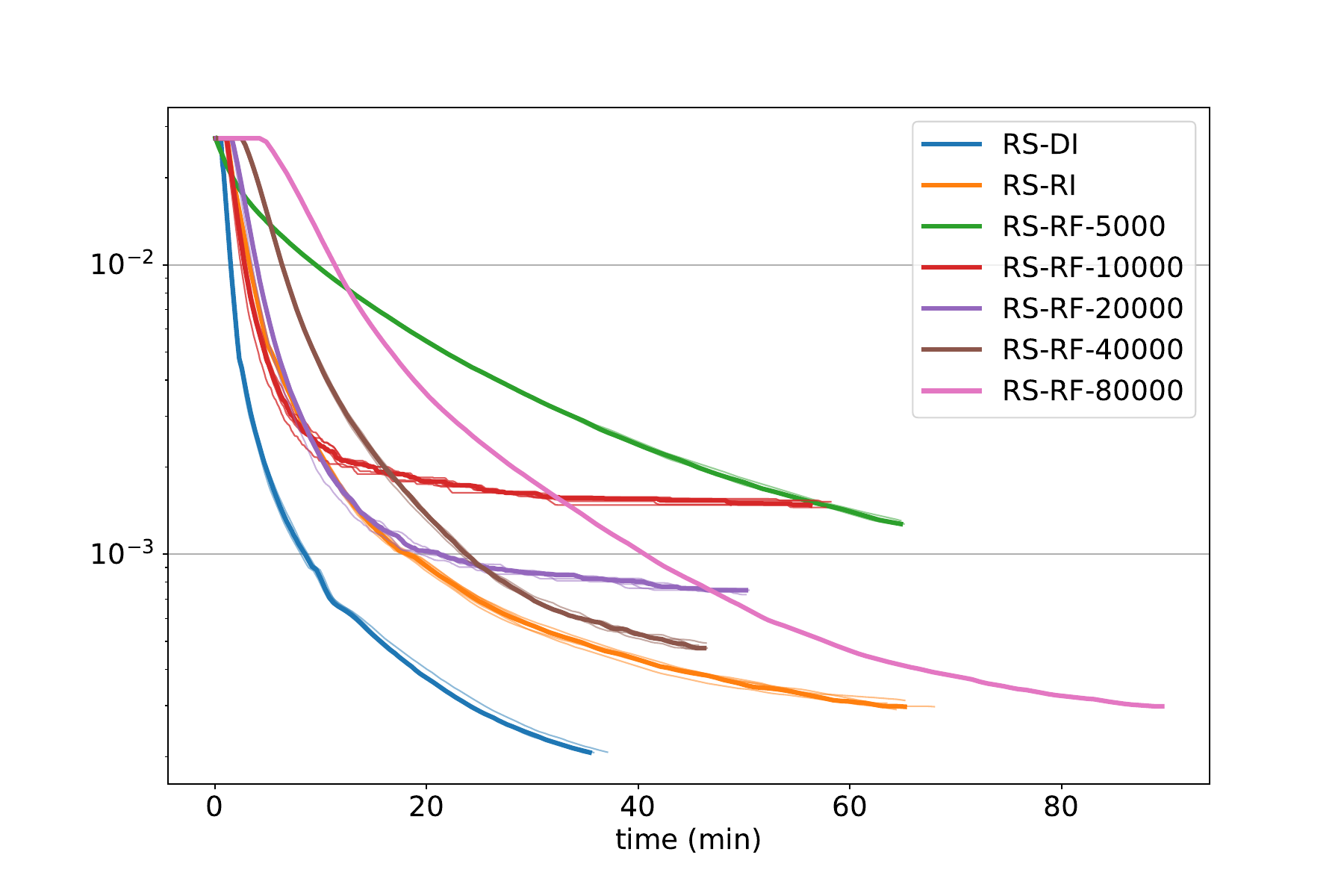}&\includegraphics[width =0.48 \textwidth,trim = 1cm 0cm 2.5cm 2cm,clip ]{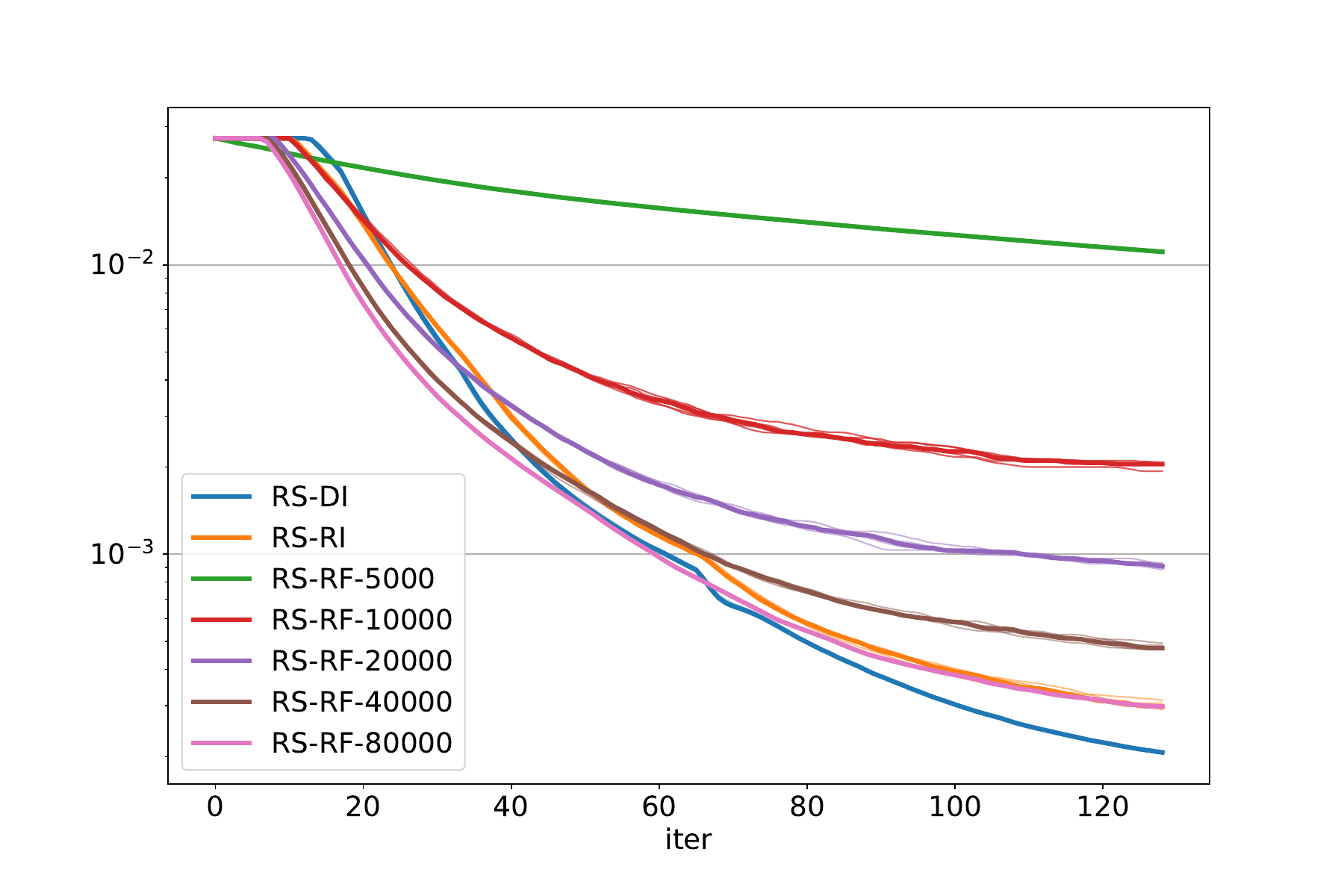}
\\
 \multicolumn{3}{c}{ {Objective Profiles for Laplace, $n=10{,}000$, $d=4$}} \\
\rotatebox{90}{  { {~~~~~~~~~~~~~$\texttt{relobj}$}}}&\includegraphics[width =0.48 \textwidth,trim = 1cm 0cm 2.5cm 2cm,clip ]{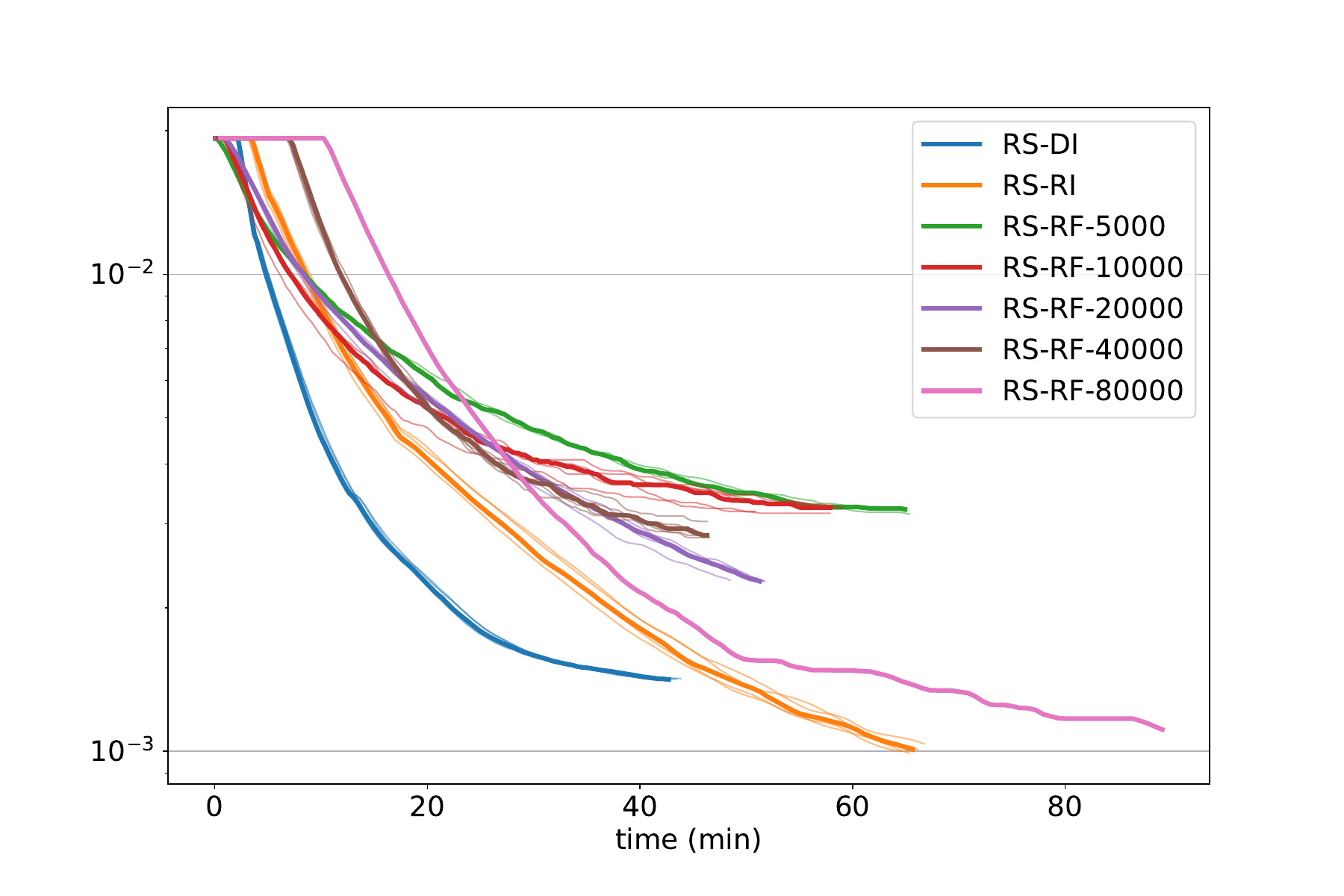}&\includegraphics[width =0.48 \textwidth,trim = 1cm 0cm 2.5cm 2cm,clip ]{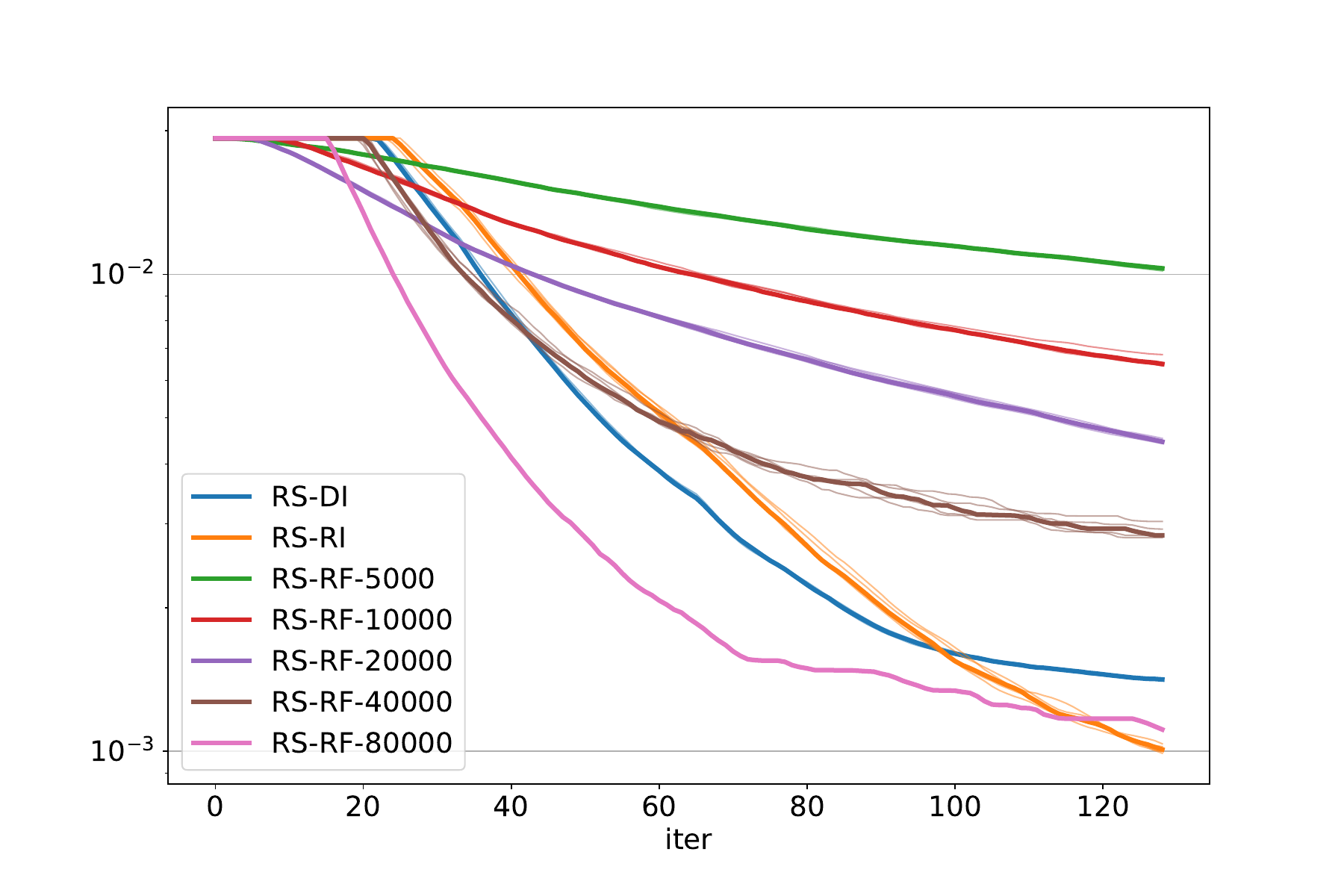}
    \end{tabular}
    }
    \caption{{\small{Plots on a log-scale of Relative Objective 
    versus time (mins) [left panel] and number of iterations [right panel].
}}}
    \label{fig:profile-RS}
\end{figure}

\bibliographystyle{apalike}
\bibliography{references}

\begin{appendix}

 \section{Additional implementational and experimental details}\label{sec:addl-expts}


\subsection{Initialization: non-convex method}\label{subsec:init}
\cite{cule2010maximum} show that the negative log-density $-\log\hat{p}_n(\cdot)$ of the log-concave MLE is a piecewise-affine convex function over its domain $C_n$.   This allows us to parametrize these functions as $\varphi(\bx):=\max_{j \in [m]}\{\bm{a}_j^\top\bx+b_j\}$ for $x \in C_n$, where $\bm{a}_1,\ldots,\bm{a}_m \in \R^d$ and $b_1,\ldots,b_m \in \R$.  We can then reformulate the problem as $\min_{\bm{a}_1,\ldots,\bm{a}_m \in \R^d,b_1,\ldots,b_m \in \R} f_0(\bm{a},\bm{b})$, with non-convex objective
\begin{equation}
    f_0(\bm{a},\bm{b}) :=\frac{1}{n}\sum_{i=1}^n\max_{j\in [m]}\{\bm{a}_j^\top\bx_i+b_j\}+\int_{C_n}e^{-\max_{j\in [m]}\{\bm{a}_j^\top\bx+b_j\}}\dx.\label{eqn:nonconvex}
\end{equation}
To approximate the integral, we use the same simple Riemann grid points mentioned in Section~\ref{SubSubSec:Gradient}.  Subgradients of the objective \eqref{eqn:nonconvex} are straightforward to compute via the chain rule and the subgradient of the maximum function (see, e.g., \cite{bertsekas2009convex}).  After standardizing each coordinate of our data to have mean zero and unit variance, which does not affect the final outcome due to the affine equivariance of the log-concave MLE \cite[Remark~2.4]{dumbgen2011approximation}, we generate $m=10$ initializing hyperplanes from a standard $(d+1)$-dimensional Gaussian distribution.  We then obtain the initializer for our main algorithm by applying a vanilla subgradient method to the objective~\eqref{eqn:nonconvex} \citep{shor1985minimization,polyak1987introduction} with stepsize $t^{-1/2}$ at the $t$th iteration, terminating when the difference in the objective function at successive iterations drops below $10^{-4}$, or after 100 iterations, whichever is the sooner.  This technique is related to the non-convex method for log-concave density estimation proposed by \cite{rathke2019fast}, who considered a smoothed version of~\eqref{eqn:nonconvex}.  Our goal here is only to seek a good initializer rather than the global optimum, and we found that the approach described above was effective in this respect, as well as being  faster to compute than the method of \cite{rathke2019fast}.

\subsection{Final polishing step}
\label{Sec:Polish}

As mentioned in Section~\ref{SubSubSec:Gradient}, once our algorithm terminates at $\tilde{\bphi}_T$, say, we evaluate the integral $I(\tilde{\bphi}_T)$ in the same way as \cite{cule2010maximum}.  Our final output, then is $\bphi_T := \tilde{\bphi}_T + \log I(\tilde{\bphi}_T)\bm{1}$; this final step not only improves the objective function, but also guarantees that $\exp[-\mathrm{cef}[\bphi_T](\cdot)]$ is a log-concave density.  This can be shown by following the same arguments as in Steps 2-3 in the proof of Theorem~\ref{thm:s-concave-MLE}. 


\subsection{Input parameter settings}\label{app:expt-setting}

According to Theorem~\ref{thm:smoothing-guarantee-expectation}, we should take $u=\frac{D}{2}\sqrt{\frac{B_1}{B_0}}$. By Table~\ref{table:smoothing-comparison}, for randomized smoothing, this is approximately $\frac D2C_1\sqrt{n}$, where $C_1=\sqrt{\Delta e^{-\phi_0}}$.  In our experiments for randomized smoothing, we chose $u=Dn^{1/4}/2$, while for Nesterov smoothing, we chose $u= D/2$.  According to Theorem~\ref{thm:smoothing-guarantee-expectation}, $\eta = \sigma M_T^{(1)}/D$, where we took $\sigma=10^{-4}$ for RS-RI and RS-DI, and $\sigma=10^{-3}$ for NS-RI and NS-DI. For the competing RS-RF-$m$ method, we present the better of the results from $\sigma \in \{10^{-3},10^{-4}\}$. 

\begin{table}[ht]
    \centering
\caption{Examples of increasing grid size ($|\mathcal S_t| = m_{t}$) schemes to achieve $\tilde{O}(1/T)$ rate (i.e. $M_T^{(1/2)}=\tilde O(1)$ for deterministic $\mathcal{S}_t$ and $M_T^{(1)}=\tilde O(1)$ for random $\mathcal{S}_t$). Here, $C$ and $C_1$ are positive constants. For the multi-stage scheme, $a\geq 1$ denotes the current stage number.}
    \label{table:examples-MT}
    {\scalebox{0.95}{\begin{tabular}{c|ccc}
    \toprule
        Schemes& Grid Sizes& $M_T^{(1/2)}=\tilde O(1)$&$M_T^{(1)}=\tilde O(1)$  \\
        \midrule
        Exponential & $m_t=C\rho^t$&$\rho >1$&$\rho >1$\\\midrule
        Polynomial&$m_t=Ct^\beta$&$\beta \geq 2$&$\beta\geq 1$\\\midrule
        \multirow{2}{*}{Multi-stage} & $m_t=C\rho^a$ for &\multirow{2}{*}{$\rho_1^2\leq \rho$}&\multirow{2}{*}{$\rho_1\leq \rho$}\\
        &$t\in\bigl[\mathbbm{1}_{\{a \geq 1\}} + C_1\rho_1^{a-1}\mathbbm{1}_{\{a \geq 2\}},C_1\rho_1^a\bigr]$&&\\
        \bottomrule
    \end{tabular}}}
\end{table}

To illustrate the increasing grid size strategy we take in the experiments, we first present in Table~\ref{table:examples-MT} some potential schemes to achieve the $\tilde O(1/T)$ error rate on the objective function scale. In our experiments, we used the multi-stage increasing grid size scheme with $C_1=8$ and $\rho_1=2$.  For the random grid (RI), we take $C=5{,}000$ and $\rho = 2$. For the deterministic grid (DI), we first choose an axis-aligned grid with $m_{0,t}$ points in each dimension that encloses the convex hull $C_n$ of the data.  Then~$m_t$ is the number of these grid points that fall within $C_n$.  Table~\ref{table:grid-sizes} provides an illustration of this multi-stage strategy used in the numerical experiments for a Laplace distribution with $n=5{,}000$ and $d=4$.  Code for the other settings is available in the github repository \texttt{LogConcComp}.
\begin{table}[ht]
    \centering
\caption{Summary of increasing grid size strategy (illustrated with $n=5{,}000$ observations from a Laplace distribution in four dimensions). We take a  four stage grid strategy and 128 iterations in total, with stage lengths shown in second line. For deterministic grids (denoted by DI), we use $m_{0,t}$ to determine the grid size (third line in the table), and the fourth line of the table is the corresponding grid size. For random grids (denoted by RI), the fifth line is the grid size of random sample.}
    \label{table:grid-sizes}
    \begin{tabular}{c|c|cccc}\toprule
    \multicolumn{2}{l}{Stage number $a$}&1&2&3&4\\
    \multicolumn{2}{l}{Stage length}&16&16&32&64\\
    \midrule
        DI&$m_{0,t}$ & 18&22&26&30 \\
        DI&$m_t$ & 10{,}656&23{,}582&45{,}969&81{,}558 \\
        RI&$m_t$& 10{,}000&20{,}000&40{,}000&80{,}000\\
        \bottomrule
    \end{tabular}
\end{table}

\subsection{Experimental results on real data sets}\label{subsec:additional-expts}

We provide additional simulation results on three real data sets:
\begin{itemize}
    \item \textbf{Stock returns:} The Stock returns real data consist of daily returns of four stocks\footnote{International Business Machines Corporation (IBM.US), JPMorgan Chase \& Co. (JPM.US), Caterpillar Inc. (CAT.US), 3M Company (MMM.US)} over $n=10{,}000$ randomly sampled days between 1970 and 2010, normalized so that each dimension has unit variance.  The real data are available at \url{https://stooq.com/db/h/}.
    \item \textbf{Census:} The Census real data consist of percentages of the population of different age groups (18-24, 25-44, 45-64 and 65+) for $n=10{,}000$ randomly sampled Census tracts based on the 2015-2019 5-year ACS (American Community Survery)\footnote{pct\_Pop\_18\_24\_ACS\_15\_19, pct\_Pop\_25\_44\_ACS\_15\_19, pct\_Pop\_45\_64\_ACS\_15\_19, pct\_Pop\_65plus\_ACS\_15\_19}, and the data are normalized so that each dimension has unit variance. The data and description are available at \url{https://www.census.gov/topics/research/guidance/planning-databases/2021.html}.
    \item \textbf{Gas turbine:} The Gas turbine real data consist of 4 sensor measures\footnote{Ambient temperature (AT), Ambient pressure (AP), Carbon monoxide (CO), Nitrogen oxides (NOx).} aggregated over one hour from a gas turbine for $n=10{,}000$ hours between 2011 and 2015, normalized so that each dimension has unit variance. The data are available at \url{https://archive.ics.uci.edu/ml/datasets/Gas+Turbine+CO+and+NOx+Emission+Data+Set}.
\end{itemize}


Table~\ref{table:obj-real}, Figure~\ref{fig:profile-real} and Figure~\ref{fig:profile-RS-real} provide simulation results that correspond to those in Table~\ref{table:obj}, Figure~\ref{fig:profile} and Figure~\ref{fig:profile-RS} respectively, but for three three real data sets.  The table and figures reveal a qualitatively very similar story to that presented for the simulated data in Section~\ref{sec:compute}: the main conclusion is that our randomized smoothing approaches are significantly more computationally efficient than both the Nesterov smoothing and CSS methods. 

\begin{table}[t!]
\centering
\caption{Comparison of our proposed methods with the CSS solution \cite{cule2010maximum} and RS-RF \cite{duchi2012randomized}, but on 3 real datasets.  Details are given in the caption of Table~\ref{table:obj}.}
\label{table:obj-real}
Stock returns, $n=10{,}000,d=4$\\
\resizebox{0.8\textwidth}{!}{\begin{tabular}{rrrrrrrrrr}
\toprule
 algo & \texttt{param} &    \texttt{obj} & \texttt{ relobj} & \texttt{runtime} &   \texttt{dopt}  &  \texttt{iter} &{\texttt{tO}} & \texttt{aO} & \texttt{hO}\\
\midrule

\multirow{3}{*}{CSS}  &  1e-2 & 6.3395 & 6.7e-02 &  315.19 & 5.4659 &        &        &        \\
   &  1e-3 & 5.9458 & 8.7e-04 & - & 0.5130 &        &        &        \\
   &  1e-4 & 5.9406 & 0.0e-00 & - & 0.0000 &        &        &        \\\midrule
RS-DI &  None & 5.9589 & 3.1e-03 &   38.47 & 0.1428 & 128&  7.79M & 60.86K & 30.22K \\
RS-RI &  None & 5.9778 & 6.3e-03 &   59.65 & 0.2032 & 128& 6.88M & 53.75K & 32.00K \\
NS-DI &  None & 5.9506 & 1.7e-03 &  254.91 & 0.0792 & 128& 7.79M & 60.86K & 30.22K \\
NS-RI &  None & 5.9672 & 4.5e-03 &  228.82 & 0.1362 & 128& 6.88M & 53.75K & 32.00K \\\midrule
\multirow{5}{*}{RS-RF}  &  5000 & 6.0003 & 1.0e-02 &   61.98 & 0.2015 & 1024& 5.12M &  5.00K &  5.00K \\
 & 10000 & 5.9886 & 8.1e-03 &   54.20 & 0.2354 & 512& 5.12M & 10.00K & 10.00K \\
 & 20000 & 5.9852 & 7.5e-03 &   49.16 & 0.2141 & 256& 5.12M & 20.00K & 20.00K \\
 & 40000 & 5.9940 & 9.0e-03 &   46.25 & 0.3194 & 128& 5.12M & 40.00K & 40.00K \\
& 80000 & 5.9665 & 4.4e-03 &   84.59 & 0.1055 &128& 10.24M & 80.00K & 80.00K \\
\bottomrule
\end{tabular}}\\
Census, $n=10{,}000,d=4$\\
\resizebox{0.8\textwidth}{!}{\begin{tabular}{rrrrrrrrrr}
\toprule
 algo & \texttt{param} &    \texttt{obj} & \texttt{ relobj} & \texttt{runtime} &   \texttt{dopt} & \texttt{iter} & {\texttt{tO}} & \texttt{aO} & \texttt{hO}\\
\midrule

\multirow{3}{*}{CSS}  &  1e-2 & 5.4458 & 9.4e-03 &   71.36 & 0.9222 &        &        &        \\
   &  1e-3 & 5.3953 & 1.2e-05 &  812.49 & 0.0098 &        &        &        \\
   &  1e-4 & 5.3952 & 0.0e-00 & - & 0.0000 &        &        &        \\\midrule
RS-DI &  None & 5.3995 & 8.0e-04 &   31.97 & 0.0478 & 128& 6.19M & 48.33K & 25.79K \\
RS-RI &  None & 5.4003 & 9.4e-04 &   63.32 & 0.0506 & 128& 6.88M & 53.75K & 32.00K \\
NS-DI &  None & 5.3992 & 7.3e-04 &  199.74 & 0.0453 & 128& 6.19M & 48.33K & 25.79K \\
NS-RI &  None & 5.3992 & 7.4e-04 &  223.34 & 0.0475 & 128& 6.88M & 53.75K & 32.00K \\\midrule
\multirow{5}{*}{RS-RF} &  5000 & 5.4100 & 2.7e-03 &   69.86 & 0.1047 & 1024& 5.12M &  5.00K &  5.00K \\
 & 10000 & 5.4093 & 2.6e-03 &   60.79 & 0.0796 & 512& 5.12M & 10.00K & 10.00K \\
& 20000 & 5.4074 & 2.3e-03 &   55.89 & 0.1236 & 256& 5.12M & 20.00K & 20.00K \\
 & 40000 & 5.4059 & 2.0e-03 &   49.04 & 0.0587 & 128& 5.12M & 40.00K & 40.00K \\
 & 80000 & 5.3998 & 8.6e-04 &   94.51 & 0.0477 & 128& 10.24M & 80.00K & 80.00K \\
\bottomrule
\end{tabular}}\\
Gas turbine, $n=10{,}000,d=4$\\
\resizebox{0.8\textwidth}{!}{\begin{tabular}{rrrrrrrrrr}
\toprule
 algo & \texttt{param} &    \texttt{obj} & \texttt{ relobj} & \texttt{runtime} &   \texttt{dopt} &\texttt{iter} & {\texttt{tO}} & \texttt{aO} & \texttt{hO}\\
\midrule

\multirow{3}{*}{CSS}  &  1e-2 & 5.5920 & 5.7e-03 &   95.59 & 0.5908 &        &        &        \\
   &  1e-3 & 5.5617 & 2.7e-04 & - & 0.0994 &        &        &        \\
   &  1e-4 & 5.5602 & 0.0e-00 & - & 0.0000 &        &        &        \\\midrule
RS-DI &  None & 5.5693 & 1.6e-03 &   34.14 & 0.0897 & 128& 7.08M & 55.28K & 25.53K \\
RS-RI &  None & 5.5633 & 5.7e-04 &   61.90 & 0.0493 & 128& 6.88M & 53.75K & 32.00K \\
NS-DI &  None & 5.5689 & 1.6e-03 &  230.47 & 0.0914 & 128& 7.08M & 55.28K & 25.53K \\
NS-RI &  None & 5.5622 & 3.7e-04 &  224.88 & 0.0499 & 128&  6.88M & 53.75K & 32.00K \\\midrule
\multirow{5}{*}{RS-RF} &  5000 & 5.5694 & 1.7e-03 &   67.28 & 0.1111 & 1024& 5.12M &  5.00K &  5.00K \\
 & 10000 & 5.5670 & 1.2e-03 &   61.22 & 0.0794 & 512& 5.12M & 10.00K & 10.00K \\
 & 20000 & 5.5673 & 1.3e-03 &   53.08 & 0.0455 & 256& 5.12M & 20.00K & 20.00K \\
 & 40000 & 5.5657 & 9.9e-04 &   47.80 & 0.0547 & 128& 5.12M & 40.00K & 40.00K \\
 & 80000 & 5.5632 & 5.5e-04 &   92.44 & 0.0501 & 128& 10.24M & 80.00K & 80.00K \\
\bottomrule
\end{tabular}}
\end{table}

\begin{figure}[ht]
    \centering
\resizebox{0.95\textwidth}{!}{\begin{tabular}{r c c}
 \multicolumn{3}{c}{ {Objective Profiles for Stock returns, $n=10{,}000$, $d=4$}} \\
\rotatebox{90}{  { {~~~~~~~~~~~~~$\texttt{relobj}$}}}&\includegraphics[width =0.48 \textwidth,trim = 1cm 0cm 2.5cm 2cm,clip ]{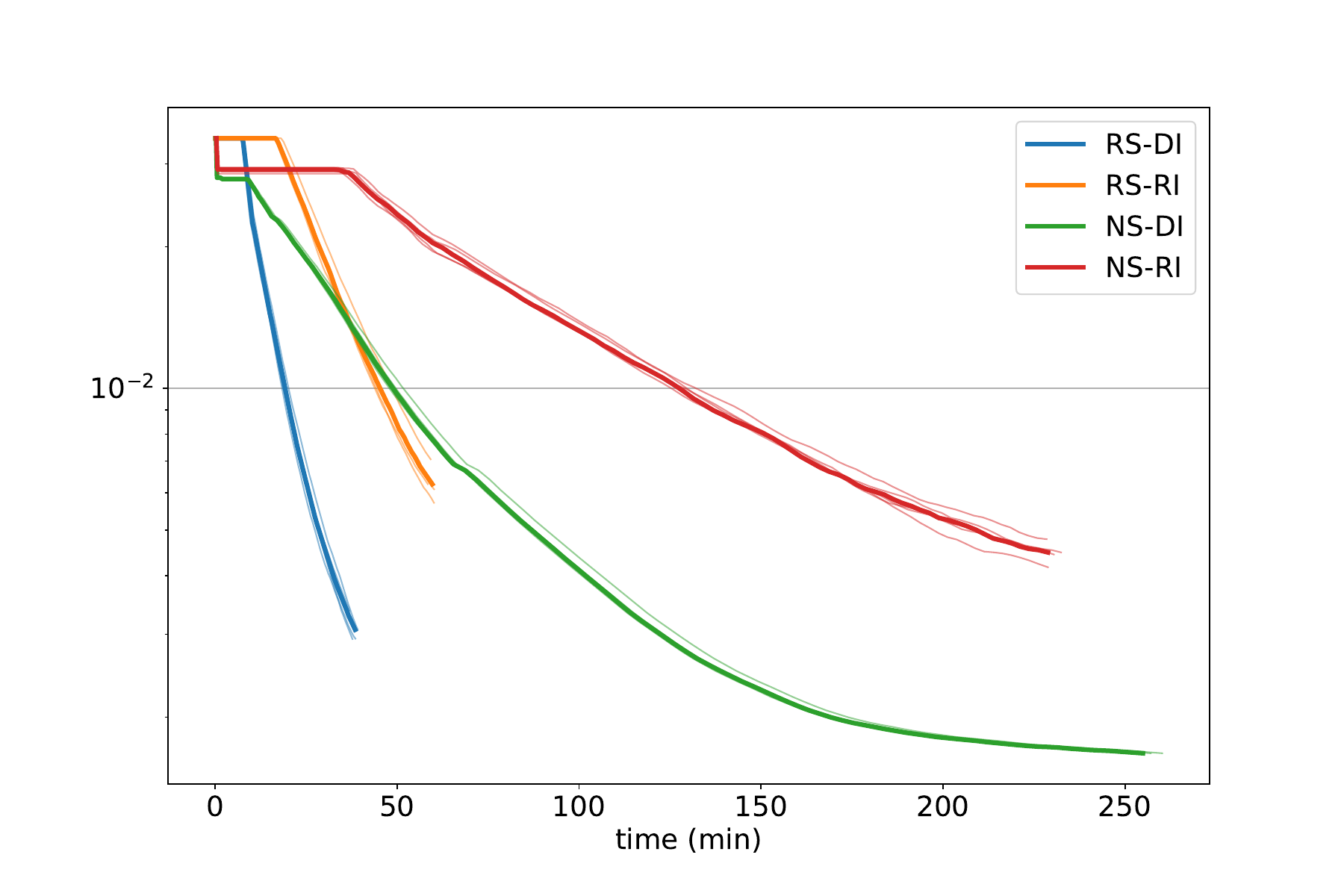}&\includegraphics[width =0.48 \textwidth,trim = 1cm 0cm 2.5cm 2cm,clip ]{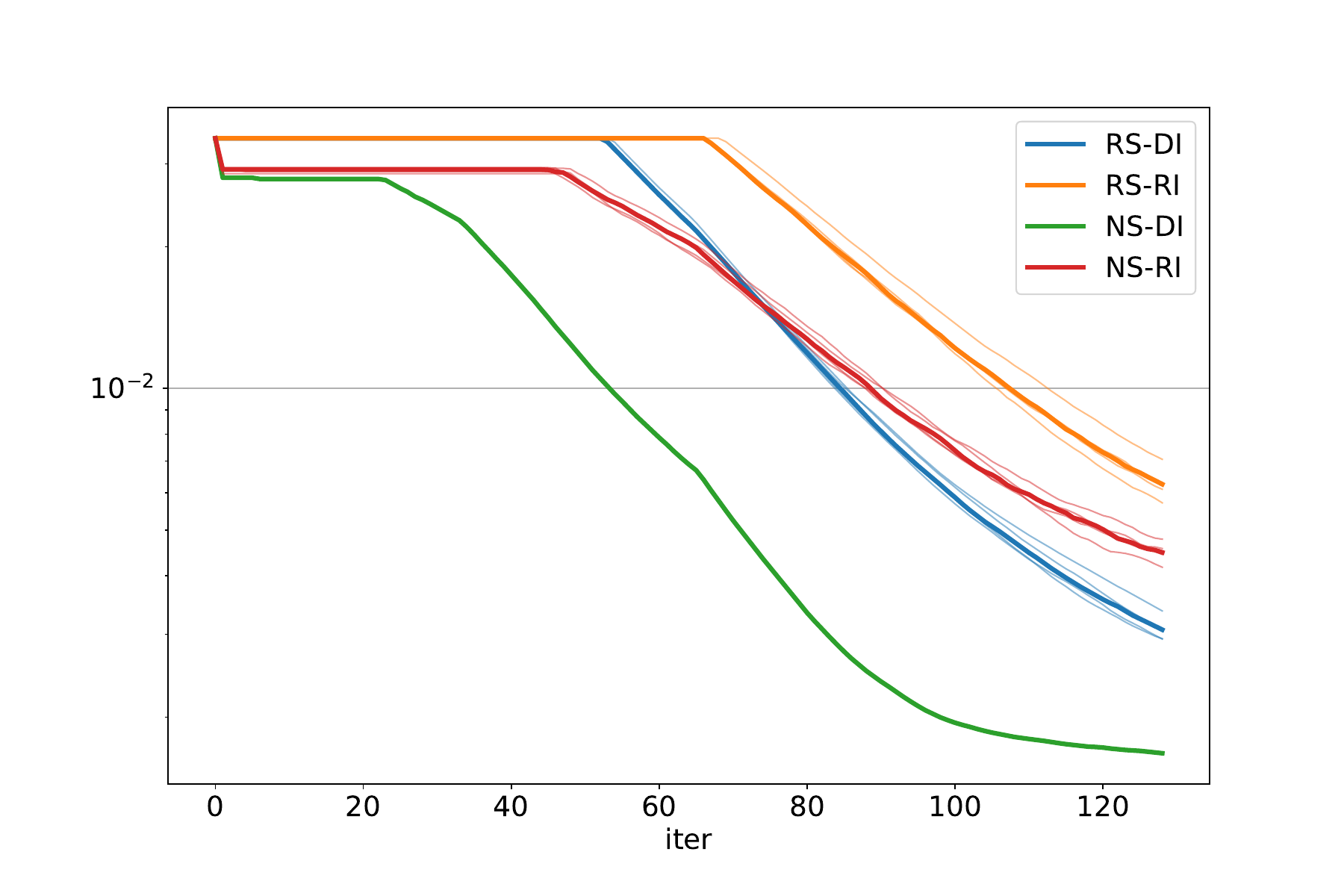}\\
\multicolumn{3}{c}{ {Objective Profiles for Census, $n=10{,}000$, $d=4$ }} \\
\rotatebox{90}{{{~~~~~~~~~~~~~~~~~$\texttt{relobj}$}}}&\includegraphics[width =0.48 \textwidth,trim = 1cm 0cm 2.5cm 2cm,clip ]{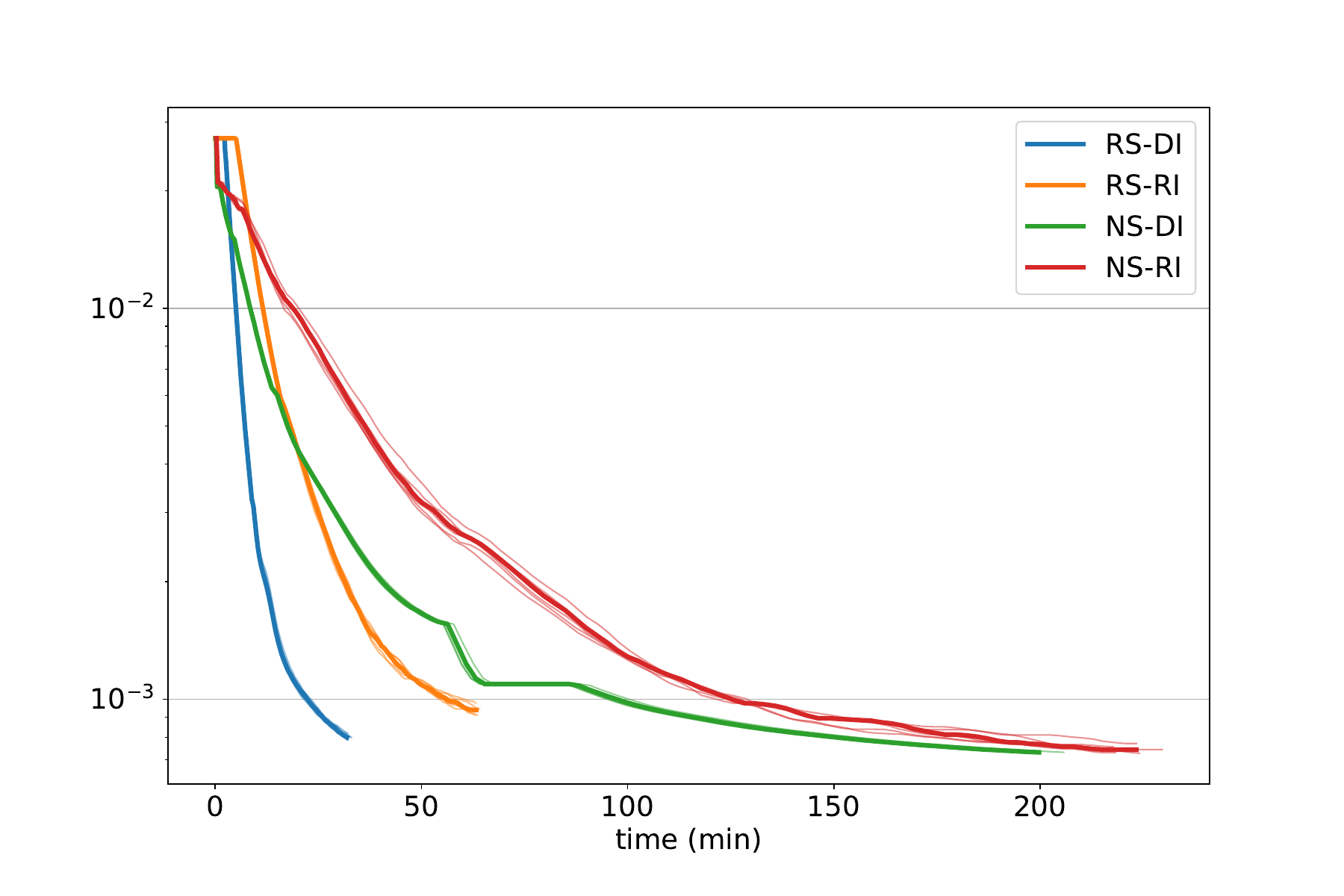}&\includegraphics[width =0.48 \textwidth,trim = 1cm 0cm 2.5cm 2cm,clip ]{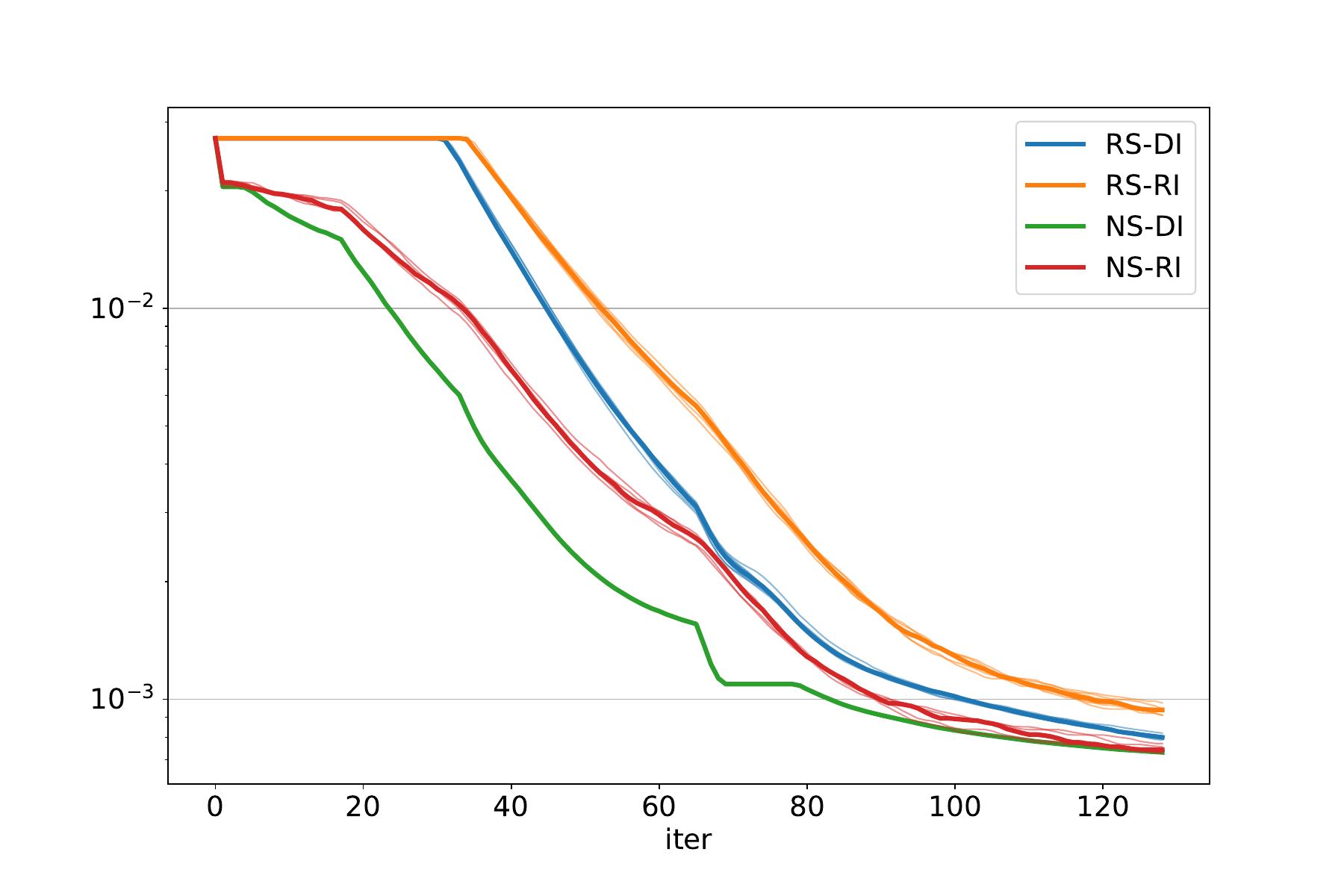}\\
\multicolumn{3}{c}{ {Objective Profiles for Gas turbine, $n=10{,}000$, $d=4$ }} \\
\rotatebox{90}{{{~~~~~~~~~~~~~~~~~$\texttt{relobj}$}}}&\includegraphics[width =0.48 \textwidth,trim = 1cm 0cm 2.5cm 2cm,clip ]{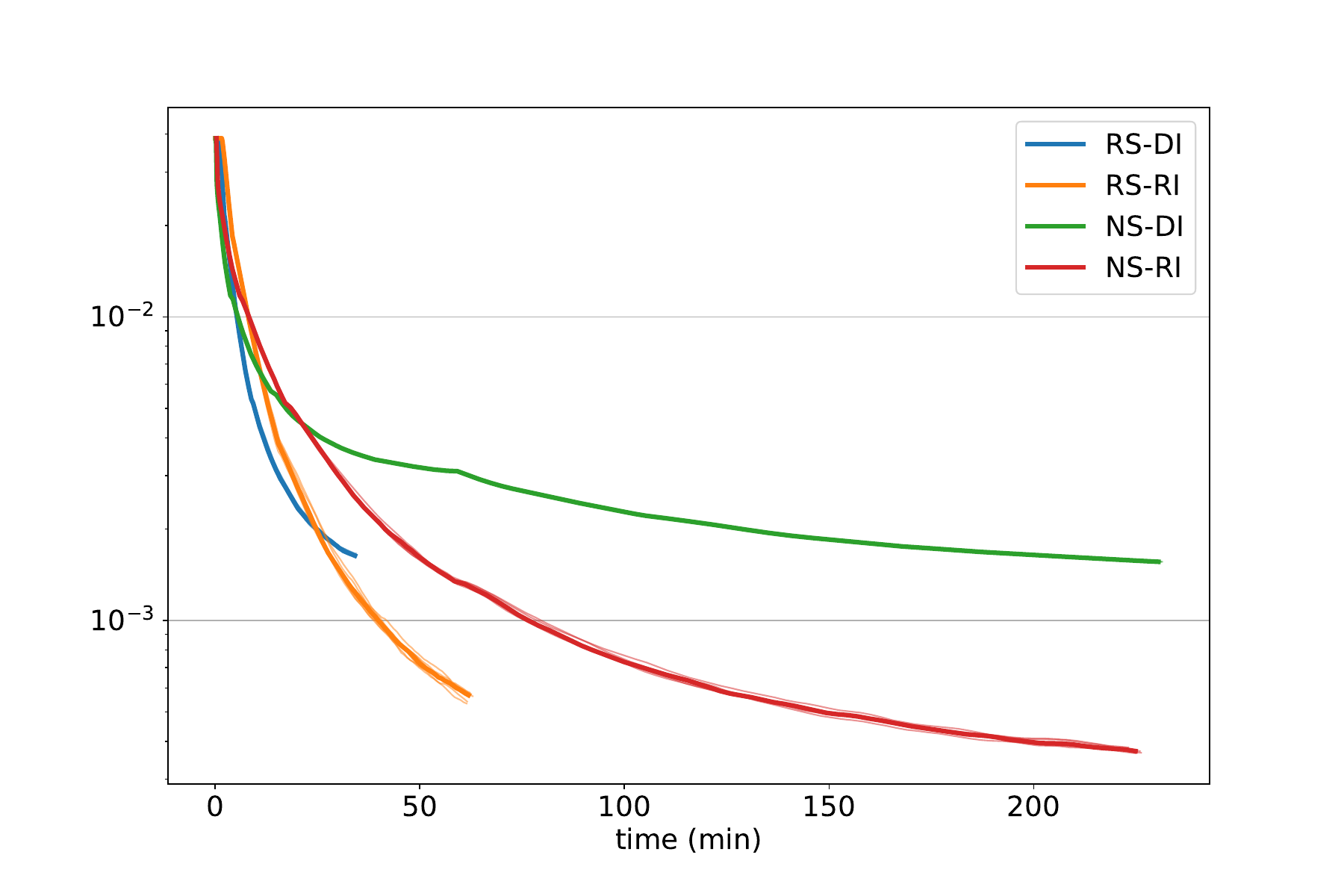}&\includegraphics[width =0.48 \textwidth,trim = 1cm 0cm 2.5cm 2cm,clip ]{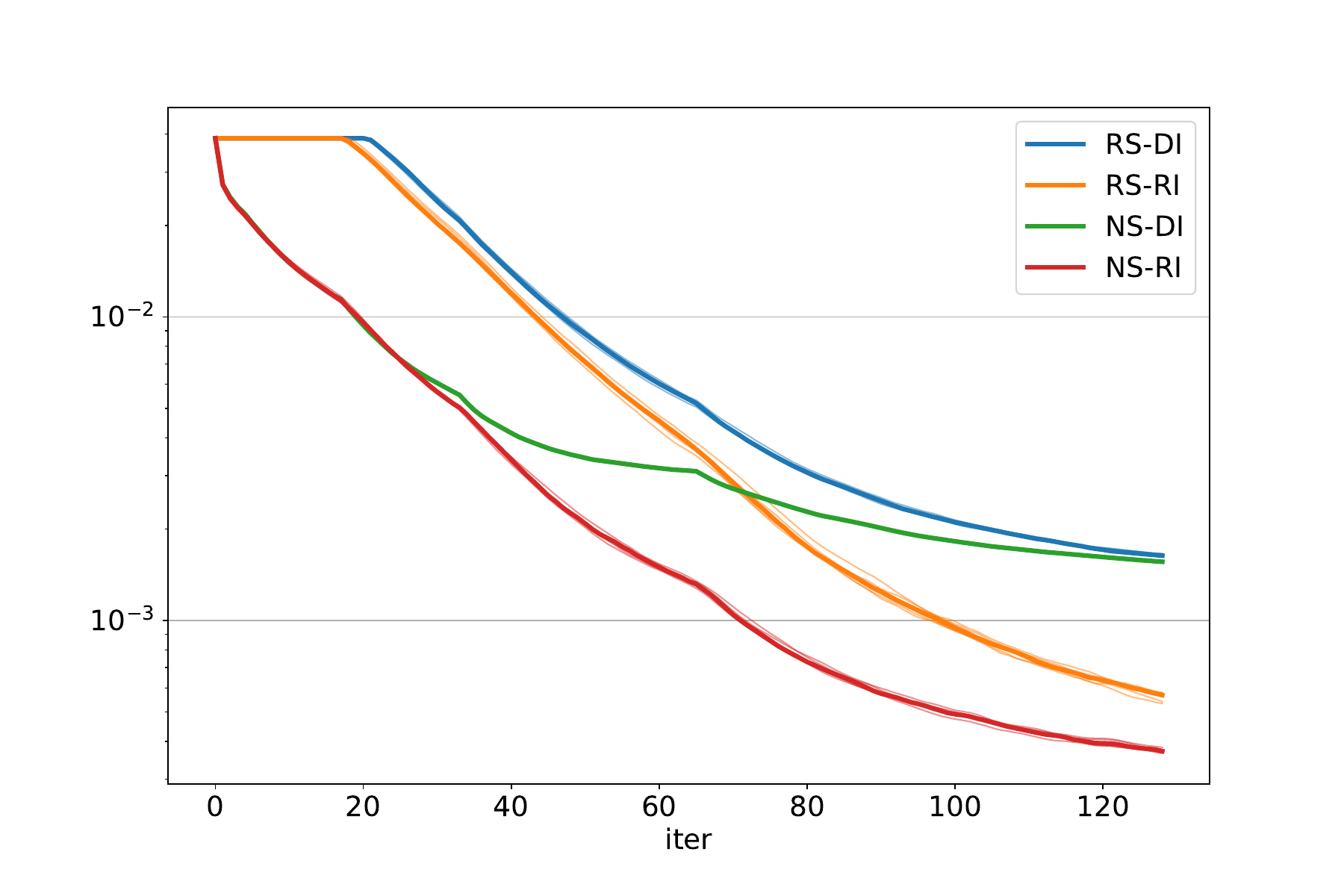}
    \end{tabular}
    }
    \caption{{\small{Additional plots on a log-scale of Relative Objective 
    versus time (mins) [left panel] and number of iterations [right panel]. Details are given in the caption of Figure~\ref{fig:profile}.
}}}
    \label{fig:profile-real}
\end{figure}

\begin{figure}[ht]
    \centering
\resizebox{0.95\textwidth}{!}{\begin{tabular}{r c c}
\multicolumn{3}{c}{ {Objective Profiles for Stock returns, $n=10{,}000$, $d=4$ }}\\
\rotatebox{90}{  { {~~~~~~~~~~~~~$\texttt{relobj}$}}}&\includegraphics[width =0.48 \textwidth,trim = 1cm 0cm 2.5cm 2cm,clip ]{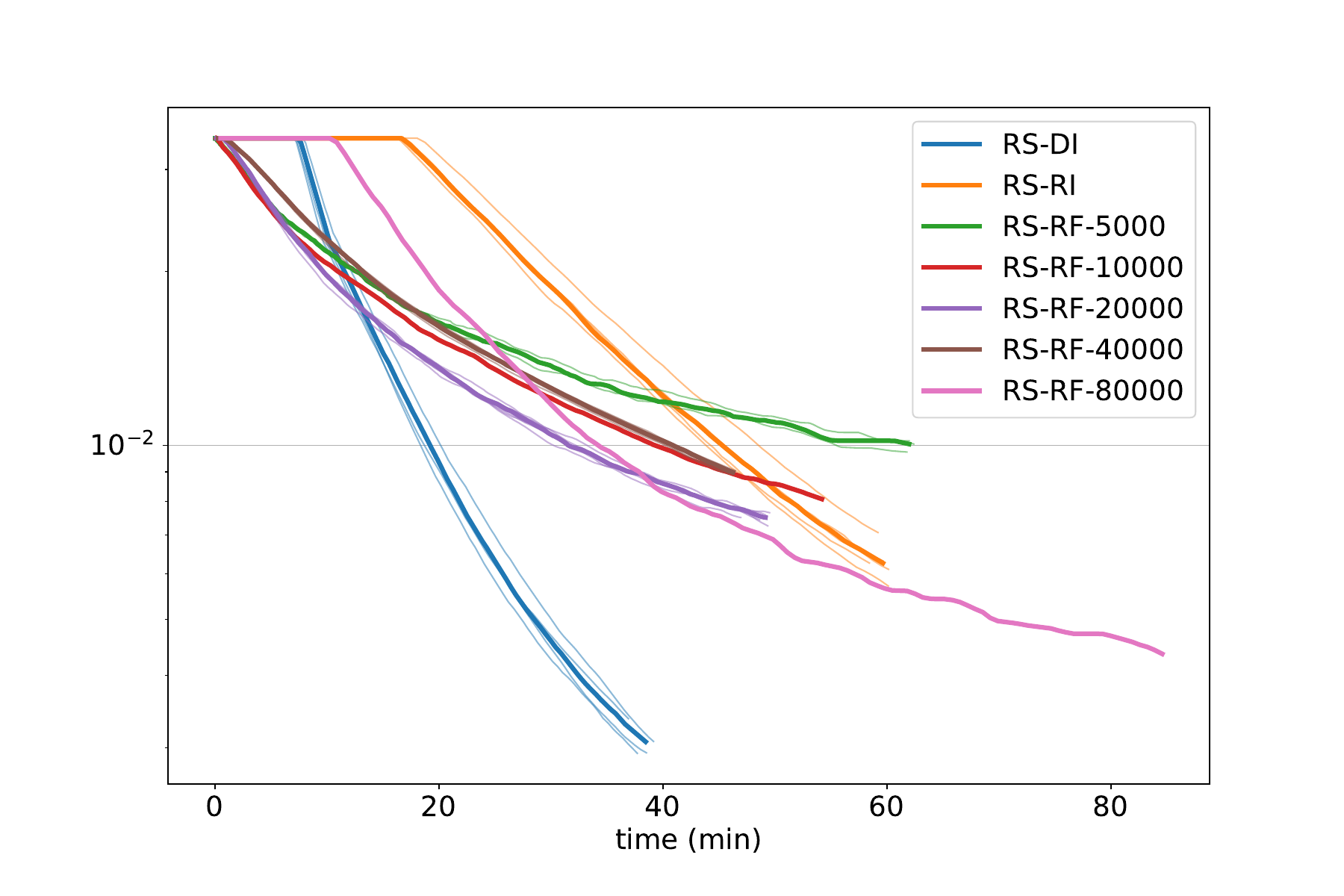}&\includegraphics[width =0.48 \textwidth,trim = 1cm 0cm 2.5cm 2cm,clip ]{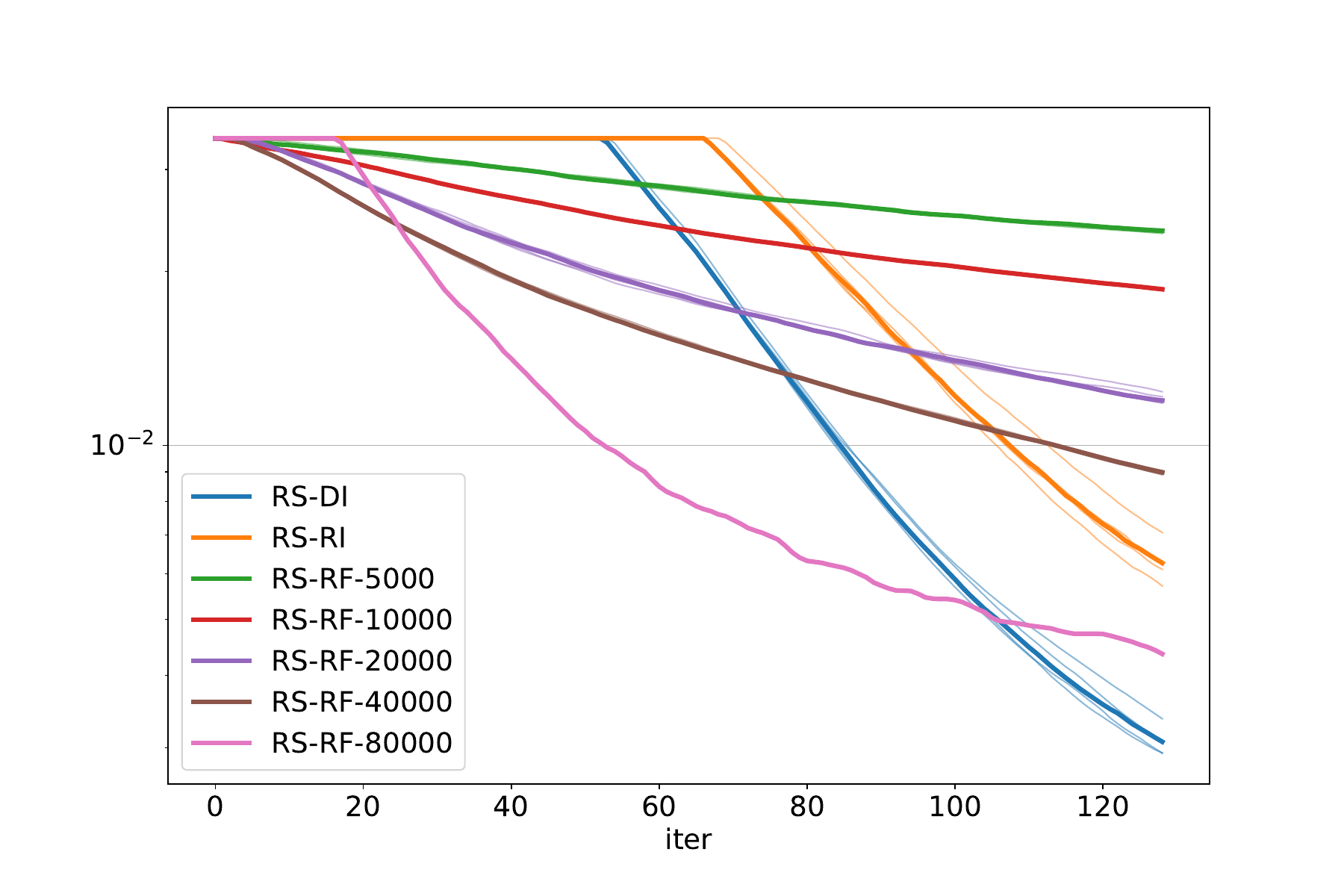}\\
 \multicolumn{3}{c}{ {Objective Profiles for Census, $n=10{,}000$, $d=4$}} \\
\rotatebox{90}{  { {~~~~~~~~~~~~~$\texttt{relobj}$}}}&\includegraphics[width =0.48 \textwidth,trim = 1cm 0cm 2.5cm 2cm,clip ]{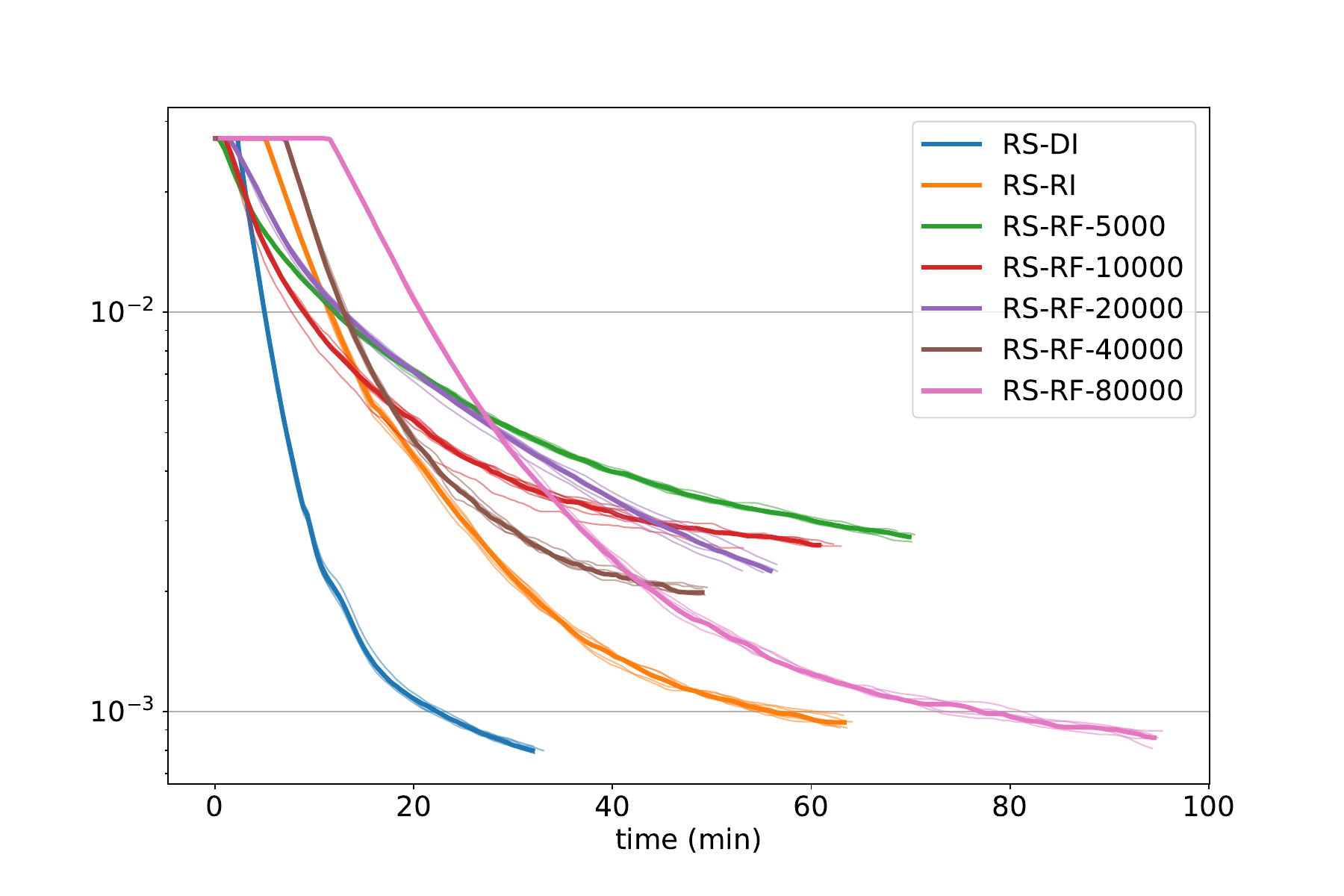}&\includegraphics[width =0.48 \textwidth,trim = 1cm 0cm 2.5cm 2cm,clip ]{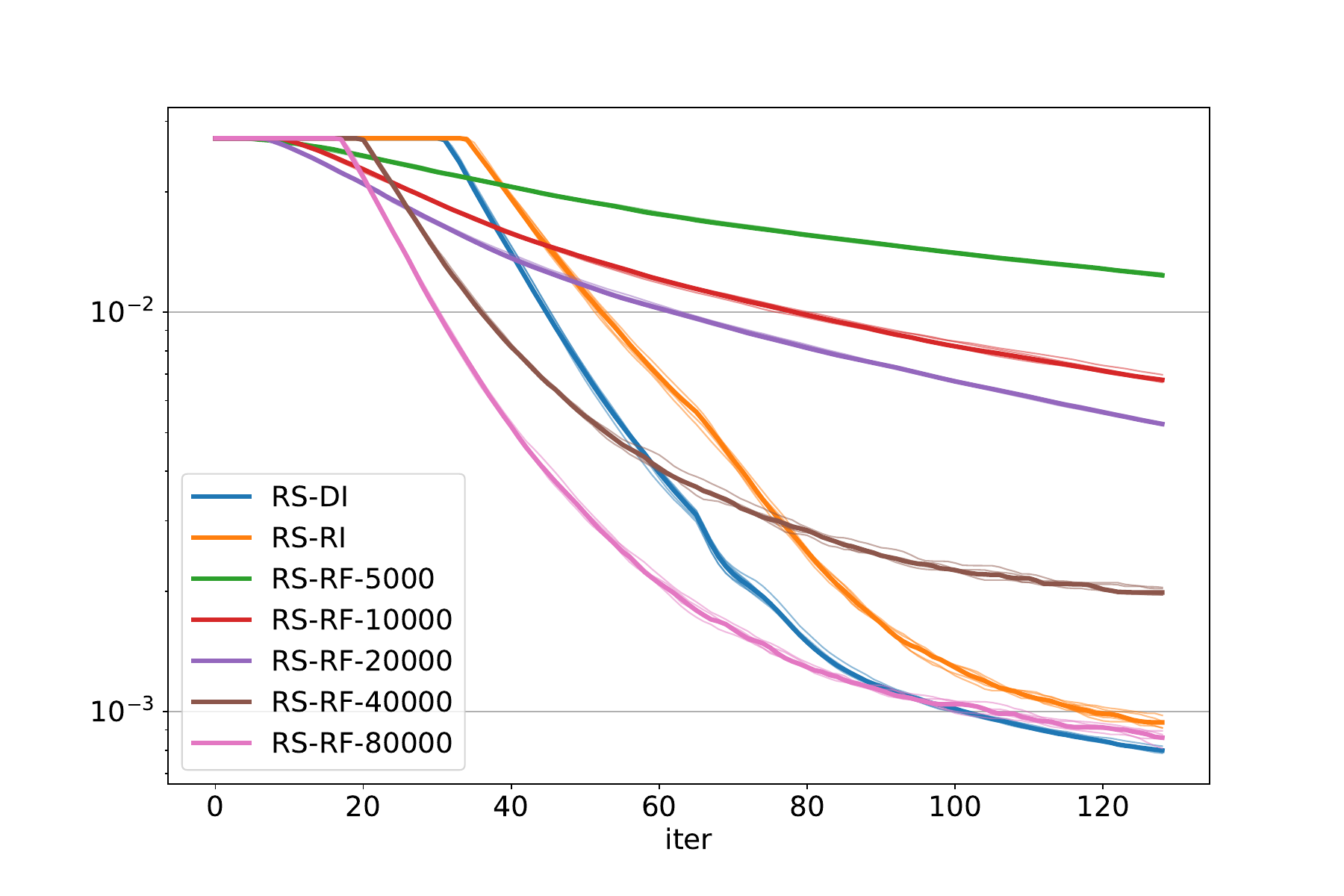}\\
\multicolumn{3}{c}{ {Objective Profiles for Gas turbine, $n=10{,}000$, $d=4$ }}\\
\rotatebox{90}{  { {~~~~~~~~~~~~~$\texttt{relobj}$}}}&\includegraphics[width =0.48 \textwidth,trim = 1cm 0cm 2.5cm 2cm,clip ]{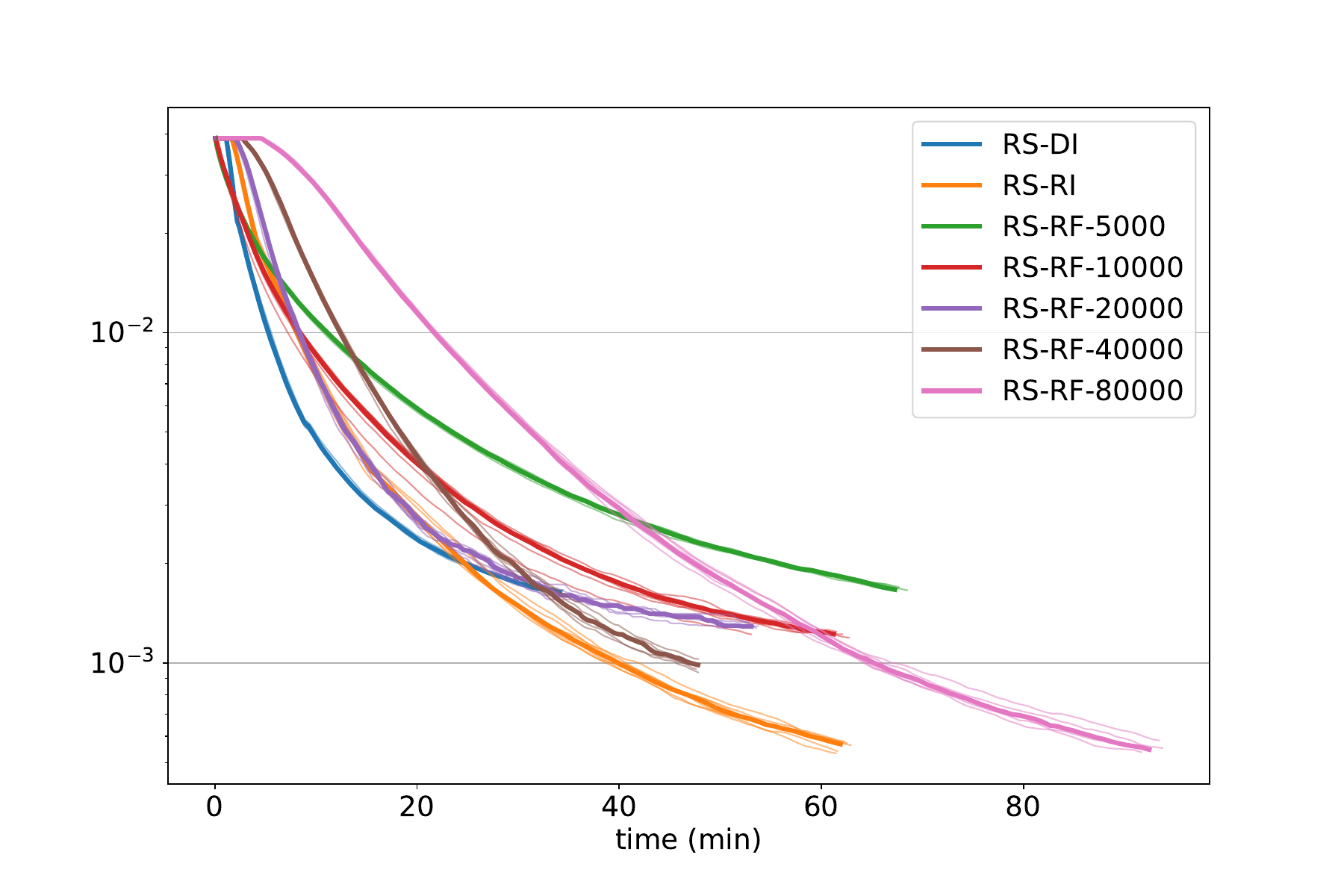}&\includegraphics[width =0.48 \textwidth,trim = 1cm 0cm 2.5cm 2cm,clip ]{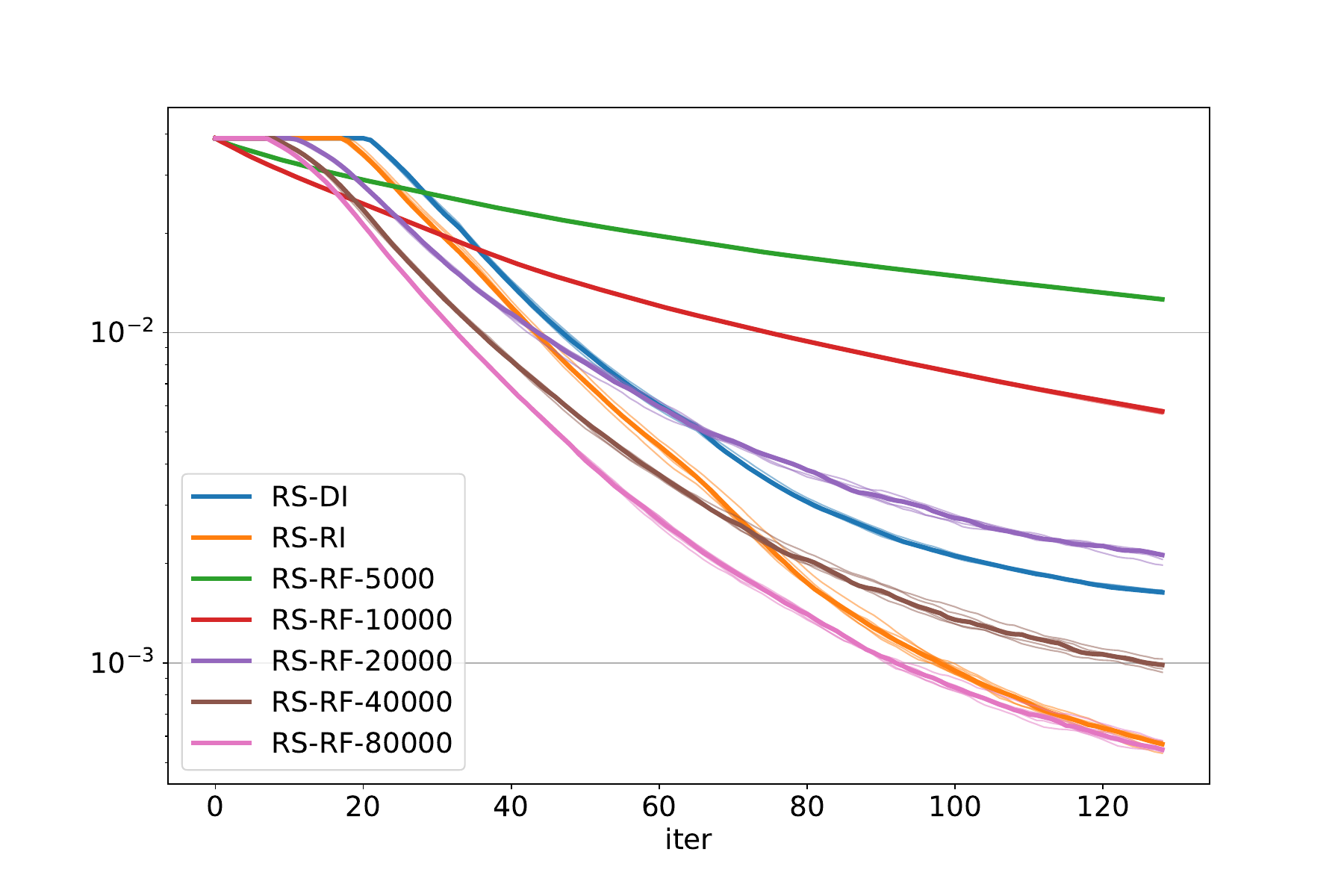}\\
    \end{tabular}
    }
    \caption{{\small{Plots on a log-scale of Relative Objective 
    versus time (mins) [left panel] and number of iterations [right panel].
}}}
    \label{fig:profile-RS-real}
\end{figure}

\clearpage
\section{Proofs}\label{sec:proofs}

\subsection{Proofs of Propositions~\ref{prop:bound-for-phis} and \ref{prop:bound-subgradient}}\label{subsec:proofs-bounds-on-phis}
The proof of Proposition~\ref{prop:bound-for-phis} is adapted from the proof of \citet[Lemma~2]{axelrod2019polynomial}, which in turn is based on \citet[Lemma~8]{carpenter2018near}. 
\begin{proof}[Proof of Proposition~\ref{prop:bound-for-phis}] The proof has three parts.

\noindent \textbf{Part 1}. We first prove that
$\phi^*_{i} \in [\phi_0,\phi^0]$ for all $i \in [n]$; or equivalently 
\begin{align}\label{eqn:bound-phi-eqn0}
    \begin{aligned}
        \log\hat{p}_n(\bm x_i)&\geq  -(n-1) - d(n-1)\log \bigl(2n + 2nd \log(2nd)\bigr) -\log \Delta\\
        \log\hat{p}_n(\bm x_i)&\leq 1+ d \log \bigl(2n + 2nd \log(2nd)\bigr) -\log \Delta
    \end{aligned}
\end{align}
for all $i \in [n]$.
Define
\[
  M:=\sup_{\bx \in \mathbb{R}^d}\hat{p}_n(\bx)=\max_{i\in[n]}\hat{p}_n(\bx_i)~~~\text{and}~~R := \frac{\max_{i\in[n]}\hat{p}_n(\bx_i)}{\min_{i\in[n]}\hat{p}_n(\bx_i)}.
\]
We proceed to obtain an upper bound on $R$.  To this end, let $\bar{p}_{n}$ denote the uniform density over $C_{n}$. If $\hat{p}_{n}(\bx_{i})=\bar{p}_{n}(\bx_{i})=1/\Delta$ for all $i$, then~\eqref{eqn:bound-phi-eqn0} holds. So we may assume that $\hat{p}_{n} \neq \bar{p}_{n}$, so that $R>1$ and $M>1/\Delta$.  For a density $p$ on $\mathbb{R}^d$ and for $t \in \mathbb{R}$, let $L_p(t) := \{\bx\in\R^d:p(\bx)\geq t\}$ denote the super-level set of $p$ at height $t$.  Since $\hat{p}_n$ is supported on $C_n$, and since $\hat{p}_n(\bx)\geq \min_{i\in[n]}\hat{p}_n(\bx_i)=M/R$ for $\bx \in C_n$, it follows by \citet[Lemma~8]{carpenter2018near}\footnote{In fact, the factor of $e$ is omitted in the statement of \citet[Lemma~8]{carpenter2018near}, but one can see from the authors' inequalities~(27) and~(28) that it should be present.} that when $R \geq e$,
\begin{equation}
\label{eqn:bound-phi-eqn1}
\Delta = \mathrm{vol}(C_n) \leq \mathrm{vol}\bigl(L_{\hat{p}_n}(M/R)\bigr) \leq \frac{e\log^d R}{M}.
\end{equation}
On the other hand, since $\inf_{\bx \in C_n}\hat{p}_n(\bx) = M/R$, we have $(M/R) \cdot \Delta \leq 1$, so for $R < e$, we have $M \leq R/\Delta < e/\Delta$.  We deduce that 
\begin{equation}
  \label{Eq:MUpperBound}
  M \leq \frac{e\log_+^d R}{\Delta},
\end{equation}
for all $R > 1$, where $\log_+(x) := 1 \vee \log x$.  
Now, by the optimality of $\hat{p}_n$, we have
\begin{align}
  \label{Eq:Opthatpn}
  n\log(1/\Delta)\! = \!\sum_{i=1}^n\log \bar{p}_n(\bx_i)\! \leq\! \sum_{i=1}^n\log \hat{p}_n(\bx_i) &\leq \min_{i \in [n]}\log \hat{p}_n(\bx_i) + (n-1)\max_{i \in [n]}\log \hat{p}_n(\bx_i) \nonumber \\
  &= \log (M/R)+(n-1)\log M,
\end{align}
so that $R \leq (M\Delta)^n$.  It follows that when $R \geq e$, we have from~\eqref{Eq:MUpperBound} and the fact that $\log y \leq y$ for $y > 0$ that
\begin{equation}
  \label{Eq:R}
  R = \frac{R^2}{R} \leq \frac{e^{2n} \log^{2nd} R}{R} \leq e^{2n}(2nd)^{2nd}.
\end{equation}
Since~\eqref{Eq:R} holds trivially when $R < e$, we may combine~\eqref{Eq:R} with~\eqref{Eq:MUpperBound} to obtain
\begin{equation}
\label{Eq:logMBound}
  \log M\leq 1+ d \log \bigl(2n + 2nd \log(2nd)\bigr) -\log \Delta
\end{equation}

Moreover, from~\eqref{Eq:Opthatpn} and~\eqref{Eq:logMBound}, we also have
\begin{align*}
  \log(M/R) &\geq - n \log \Delta - (n-1) \log M \\
  &\geq -(n-1) - d(n-1)\log \bigl(2n + 2nd \log(2nd)\bigr) -\log \Delta,
\end{align*}
as required.
%

\medskip

\noindent \textbf{Part 2}. 
Now we extend the above result to all $\bphi \in \R^n$ such that $I(\bphi)=1$ and  $f(\bphi)\leq f(\bar\bphi)$, where $\bar\bphi$ is defined just after Proposition~\ref{prop:bound-for-phis}.  The key observation here is that the proof of Part 1 applies to any density with log-likelihood at least that of the uniform distribution over~$C_{n}$.  In particular, for any $\bphi$ satisfying these conditions, 
the density $p \in \mathcal{F}_d$ given by $p(\bx)=\exp\{-\cef[\bphi](\bx)\}$ has log-likelihood at least that of the uniform distribution over~$C_{n}$, so
\begin{align*}
  -\phi^0\leq \min_{i \in [n]} \log p(\bx_i) \leq \sup_{\bx \in \R^d} \log p(\bx) \leq -\phi_0,
\end{align*}
as required.

\medskip

\noindent \textbf{Part 3}. We now consider the case for a general $\bphi \in \R^n$ with $f(\bphi)\leq f(\bar\bphi)$.  Let $\tilde{\bphi} := \bphi+\log I(\bphi)\bm1$, so that $I(\tilde{\bphi})=1$ and $\cef[\tilde{\bphi}](\cdot) =\cef[\bphi](\cdot)+\log(I(\bphi))$.  Furthermore,
\[
  f(\tilde{\bphi})= \frac{1}{n}\bm1^\top\bphi+\log I(\bphi)+1 \leq \frac{1}{n}\bm1^\top\bphi+I(\bphi)=f(\bphi)\leq f(\bar\bphi).
\]
The result therefore follows by Part~2.
\end{proof}

\begin{proof}[Proof of Proposition~\ref{prop:bound-subgradient}]
Recall our notation from Section~\ref{sec:Setup} that $\bm\alpha^*[\bphi](\bx)\in A[\bphi](\bx)$ denotes a solution to~\eqref{eqn:cef-LP} at $\bx \in C_n$.  Recall further from~\eqref{eqn:subgradient-phi} that any subgradient $\bm g(\bphi)$ of $f$ at $\bphi$ is of the form
\[
  \bm g(\bphi)= \frac{1}{n}\bm1- \bm{\gamma}, 
\]
where $\bm{\gamma}:=\Delta\E\bigl[\bm\alpha^*[\bphi](\bm\xi)\exp\{-\cef[\bphi](\bm\xi)\}\bigr]$ and $\bm\xi$ is uniformly distributed on $C_n$.
Since $\bm\alpha^*$ lies in the simplex $\{\bm\alpha\in\R^n:\bm\alpha\geq 0,\bm1^\top\bm\alpha=1 \},$
we have $\bm\gamma \geq 0$ and
\[
  \bm1^\top\bm\gamma = \Delta \E[\bm1^\top\alpha^*[\bphi](\bm\xi)\exp\{-\cef[\bphi](\bm\xi)\}]=\Delta\E[\exp\{-\cef[\bphi](\bm\xi)\}]=I(\bphi).
\]
In particular, $\norm{\bm\gamma}_1=I(\bphi)$, so 
\[
\|\bm g(\bphi)\|^2=\frac{1}{n}- \frac{2}{n}\bm1^\top\bm\gamma+\|\bm\gamma\|^2 \leq \frac1n - \frac2nI(\bphi)+ I(\bphi)^2.
\]
If $I(\bphi) \leq 1/2$, then $\norm{\bm g(\bphi)}^2\leq 1/4+ 1/n$; if $I(\phi)> 1/2$, then $\norm{\bm g(\bphi)}^2\leq I(\bphi)^2$. Therefore, $\norm{\bm g(\bphi)}^2\leq \max\bigl\{1/4+ 1/n,I(\bphi)^2\bigr\}$. 
\end{proof}

\subsection{Proofs of Proposition~\ref{prop:fhat-properties} and Proposition~\ref{prop:fbar-properties}}\label{proofs-props-properties}
The proof of Proposition~\ref{prop:fhat-properties} is based on the following properties of the quadratic program $q_u[\bphi](\bx)$ defined in~\eqref{eqn:h-QP}, as well as its unique optimizer $\bm{\alpha}_u^*[\bphi](\bx)$:
\begin{proposition}\label{prop:QP-properties}
  For $\bphi\in\bm{\Phi}$ and $\bx\in C_n$, we have
  
  \noindent(a) $\|\bm\alpha_u^*[\bphi](\bx)\|\leq 1$ for any $\bx\in C_n$ and $\bphi \in \bm\Phi$;
  
  \noindent(b) $q_{u'}[\bphi](\bx)+(u'-u)/2\leq q_u[\bphi](\bx)\leq q_{u'}[\bphi](\bx)$ for $u' \in [0,u]$;
  
  \noindent(c) $\phi_0-\frac{u}{2} \leq q_u[\bphi](\bx)\leq \phi^0 $ for all $u\geq 0$, $\bphi\in\bm{\Phi}$ and $\bx\in C_n$;  
  
  \noindent(d) $\|\bm\alpha_u^*[\tilde\bphi](\bx)-\bm\alpha_u^*[\bphi](\bx)\|\leq ({1}/{u})\|\tilde\bphi-\bphi\|$ for any $u > 0$, $\bphi,\tilde{\bphi}\in\bm{\bm{\Phi}}$, and any $\bx\in C_n$.
\end{proposition}
\begin{proof} The proof exploits ideas from~\cite{nesterov2005smooth}. For (a), observe that $\bm\alpha_u^*[\bphi](\bx)\in E(\bx) \subseteq \{\bm\alpha \in \R^n: \bm1_n^\top\bm\alpha=1, \bm\alpha\geq 0\bigr\}$, and this simplex is the convex hull of $n+1$ points that all lie in the closed unit Euclidean ball in $\R^n$.

  The lower bound in (b) follows immediately from the definition of the quadratic program in~\eqref{eqn:h-QP}.  For the upper bound, for $u' \in [0,u]$, we have
  \begin{align*}
  &~q_u[\bphi](\bx) - q_{u'}[\bphi](\bx)\\
     &= \bm\alpha^*_u[\bphi](\bx)^\top \bphi + \frac{u}{2}\bigl\|\bm\alpha_u^*[\bphi](\bx) - \bm\alpha_0\bigr\|^2 - \bm\alpha^*_{u'}[\bphi](\bx)^\top \bphi - \frac{u'}{2}\bigl\|\bm\alpha_{u'}^*[\bphi](\bx) - \bm\alpha_0\bigr\|^2+\frac{u'-u}{2} \\
    &\leq \frac{u-u'}{2}\biggl(\bigl\|\bm\alpha_{u'}^*[\bphi](\bx) - \bm\alpha_0\bigr\|^2-1\biggr) \leq \frac{u-u'}{2}\biggl(1 - \frac{2}{n} + \frac{1}{n}-1\biggr) \leq 0.
  \end{align*}

  (c) For all $u\geq 0$, $\bphi\in\bm{\Phi}$ and $\bx\in C_n$, we have
  \begin{align*}
    q_u[\bphi](\bx) = \inf_{\bm \alpha \in E(\bx)}\biggl(\bm\alpha^\top \bphi + \frac{u}{2}\|\bm\alpha - \bm\alpha_0\|^2-\frac{u}2\biggr) &\geq \inf_{\bm \alpha \in E(\bx)} \bm\alpha^\top \bphi-\frac{u}2 \\
    &\geq \phi_0 \inf_{\bm \alpha \in E(\bx)} \bm \alpha^\top \bm1_n-\frac{u}2 = \phi_0-\frac{u}2.
  \end{align*}
  Similarly,
  \[
    q_u[\bphi](\bx) \leq \phi^0 \sup_{\bm \alpha \in E(\bx)} \bm\alpha^\top \bm1_n + \frac{u}{2}\sup_{\bm \alpha \in E(\bx)}\|\bm\alpha - \bm\alpha_0\|^2-\frac{u}2 \leq \phi^0.
  \]

(d) Observe that
\[
  \bm\alpha_u^*[\bphi](\bx) = \argmin_{\bm \alpha \in E(\bx)}~ \biggl(\bm\alpha^\top \bphi + \frac{u}{2}\|\bm\alpha - \bm\alpha_0\|^2\biggr) = \argmin_{\bm \alpha \in E(\bx)}~ \biggl\|\bm\alpha - \biggl(\bm\alpha_0 - \frac{\bphi}{u}\biggr)\biggr\|^2,
\]
so $\bm\alpha_u^*[\bphi](\bx)$ is the Euclidean projection of $\bm\alpha_0- ({\bphi}/{u})$ onto $E(\bx)$.  Since this projection is an $\ell_2$-contraction, we deduce that
\[
  \bigl\|\bm\alpha_u^*[\tilde{\bphi}](\bx) - \bm\alpha_u^*[\bphi](\bx)\bigr\| \leq \biggl\|\bm\alpha_0- \frac{\tilde{\bphi}}{u} - \biggl(\bm\alpha_0- \frac{\bphi}{u}\biggr)\biggr\| = \frac{1}{u}\|\tilde{\bphi} - \bphi\|,
  \] 
as required.
\end{proof}
\begin{proof}[Proof of Proposition~\ref{prop:fhat-properties}]
  (a) For $u \geq 0$ and $\bphi \in \bm\Phi$, let
  \[
    \tilde{I}_u(\bphi) := \int_{C_n} e^{-q_u[\bphi](\bx)} \, d\bx = \Delta \mathbb{E}(e^{-q_u[\bphi](\bm\xi)}),
  \]
  where $\bm\xi$ is uniformly distributed on $C_n$, so that $\tilde{I}_0(\bphi) = I(\bphi)$.  By definition of $\tilde{f}_u$, we have for $u' \in [0,u]$ that
  \begin{equation}
    \label{eqn:fhat-property-a-eqn1}
    \tilde{f}_u(\bphi)-\tilde{f}_{u'}(\bphi)=\tilde{I}_u(\bphi)-\tilde{I}_{u'}(\bphi)=\Delta\E\bigl(e^{-q_u[\bphi](\bm\xi)}-e^{-q_{u'}[\bphi](\bm\xi)}\bigr) \geq 0,
  \end{equation}
where the inequality follows from Proposition~\ref{prop:QP-properties}(b).  Hence, for every $u \geq 0$ and $\bphi \in \bm\Phi$, 
\begin{equation}\label{eqn:fhat-property-a-eqn3}\tilde{I}_{u}(\bphi)\leq e^{u/2}\tilde{I}_0(\bphi) = e^{u/2}I(\bphi).
\end{equation}
Now, from~\eqref{eqn:fhat-property-a-eqn1}, Proposition~\ref{prop:QP-properties}(b) and~\eqref{eqn:fhat-property-a-eqn3}, we deduce that
\[
  \tilde{f}_u(\bphi)-\tilde{f}_{u'}(\bphi)\leq \bigl(e^{(u-u')/2}-1\bigr)\tilde{I}_{u'}(\bphi) \leq \frac{u-u'}{2}e^{u'/2}I(\bphi),
\]
as required.



\medskip

\noindent (b) For each $\bx \in C_n$, the function $\bphi \mapsto q_u[\bphi](\bx)$ is the infimum of a set of affine functions of $\bphi$, so it is concave.  Moreover, $y \mapsto e^{-y}$ is a decreasing convex function, so $\bphi \mapsto e^{-q_u[\bphi](\bx)}$ is convex, and it follows that $\bphi \mapsto (1/n)\bm1^\top \bphi + \Delta\mathbb{E}(e^{-q_u[\bphi](\bm\xi)}) = \tilde{f}_u(\bphi)$ is convex. Similarly to the proof of Proposition~\ref{prop:bound-subgradient}, any subgradient $\tilde{\bm g}_u(\phi)$ of $\tilde{f}_u$ at $\bphi$ satisfies
\begin{equation}\label{proof-prop3-linea-1}
\norm{\tilde{\bm{g}}_u(\bphi)}^2\leq \max\biggl\{\frac14+\frac1n,\tilde{I}_u(\bphi)^2\biggr\}\leq \max\biggl\{\frac14+\frac1n,\Delta^2e^{-2\phi_0+u}\biggr\}.
\end{equation}
But $\bphi^*\in\bm{\Phi}$, so $\Delta^2 e^{-2\phi_0+u}\geq (\Delta e^{-\phi_0})^2 \geq I(\bphi^*)^2 = 1 \geq 1/4 + 1/n$. Hence, $\tilde{f}_u$ is $e^{-\phi_0+u/2}$-Lipschitz.

\medskip

\noindent (c) 
To establish the Lipschitz property of $\nabla_{\bm\phi}\tilde{f}_u$, for any $\bx \in C_n$, any $\bphi,\tilde\bphi\in \bm{\Phi}$, and $t\in[0,1]$, we define
\[
  \eta(t):=e^{-q_u[\bphi+t(\tilde\bphi-\bphi)](\bx)}.
\]
Then
\[
  \eta'(t)=-e^{-q_u[\bphi+t(\tilde\bphi-\bphi)](\bx)}(\tilde\bphi-\bphi)^\top\bm\alpha_u^*[\bphi+t(\tilde\bphi-\bphi)](\bx).
\]
By the mean value theorem there exists $t_0\in[0,1]$ such that  
\begin{align}\label{prop-3-c-line}
  &\bigl|e^{-q_u[\tilde\bphi](\bx)}-e^{-q_u[\bphi](\bx)}\bigr| = |\eta(1)\!-\!\eta(0)|=|\eta'(t_0)| \nonumber\\
  &\leq e^{-q_u[\bphi+t_0(\tilde\bphi-\bphi)](\bx)}\|\tilde\bphi-\bphi\|\bigl\|\bm\alpha_u^*[\bphi+t_0(\tilde\bphi-\bphi)](\bx)\bigr\| \nonumber \\
  &\leq e^{-\phi_0+u/2}\|\tilde\bphi-\bphi\|,
\end{align}
where the final bound follows from Proposition~\ref{prop:QP-properties}(a) and~(c).  Now, for any $\bx \in C_n$, we have by~\eqref{prop-3-c-line} as well as Proposition~\ref{prop:QP-properties}(a),~(c) and~(d) that 
\begin{align*}
    &\bigl\|e^{-q_u[\tilde\bphi](\bx)}\bm\alpha_u^*[\tilde\bphi](\bx)-e^{-q_u[\bphi](\bx)}\bm\alpha_u^*[\bphi](\bx)\bigr\| \nonumber \\
    &= \bigl\|\bigl(e^{-q_u[\tilde\bphi](\bx)}-e^{-q_u[\bphi](\bx)}\bigr)\bm\alpha_u^*[\tilde\bphi](\bx)+e^{-q_u[\bphi](\bx)}\bigl(\bm\alpha_u^*[\tilde\bphi](\bx)-\bm\alpha_u^*[\bphi](\bx)\bigr)\bigr\| \nonumber \\
    &\leq e^{-\phi_0+u/2}\|\tilde\bphi-\bphi\|+ e^+u/2{-\phi_0}\frac{1}{u}\|\tilde\bphi-\bphi\|  =e^{-\phi_0+u/2}(1+u^{-1})\|\tilde\bphi-\bphi\|. \label{prop3-long-ineq-diff1}
\end{align*}
It follows that for any $\bphi,\tilde\bphi\in \bm{\Phi}$, we have
\begin{align*}
    \bigl\|\nabla_{\tilde{\bphi}}\tilde{f}_u(\tilde\bphi)-\nabla_{\bphi}\tilde{f}_u(\bphi)\bigr\|\leq \Delta\E\bigl\|e^{-q_u[\tilde\bphi](\bm\xi)}\bm\alpha_u^*[\tilde\bphi](\bm\xi)&-e^{-q_u[\bphi](\bm\xi)}\bm\alpha_u^*[\bphi](\bm\xi)\bigr\| \nonumber\\
    & \leq \Delta e^{-\phi_0+u/2}(1+u^{-1})\|\tilde\bphi-\bphi\|,
\end{align*}
as required.

\medskip

\noindent (d) For $u \geq 0$ and $\bphi\in\bm\Phi$, it follows from Proposition~\ref{prop:QP-properties}(a) and~(c) that
\begin{align*}
    \E\bigl(\|{\tilde{\bm G}_u(\bphi,\bm\xi)-\nabla\tilde{f}_u(\bphi)}\|^2\bigr) &= \E \bigl\|\Delta e^{-q_u[\bphi](\bm\xi)}\bm\alpha_u^*[\bphi](\bm\xi)-\Delta\E\bigl(e^{-q_u[\bphi](\bm\xi)}\bm\alpha_u^*[\bphi](\bm\xi)\bigr)\bigr\|^2\\
    &\leq \Delta^2 \E \bigl\|e^{-q_u[\bphi](\bm\xi)}\bm\alpha_u^*[\bphi](\bm\xi)\bigr\|^2 \leq (\Delta e^{-\phi_0+u/2})^2,
\end{align*}
as required.
\end{proof}

\begin{proposition}\label{proof-expected-prop6}
If $\bm{z}$ is uniformly distributed on the unit $\ell_{2}$-ball in $\mathbb{R}^n$, then $\E(\norm{\bz}_\infty)\leq \sqrt{\tfrac{2\log n}{n+1}}$.
\end{proposition}
\begin{proof}
By \citet[Proposition~3]{xu2019high}, we have that $\bz \overset{d}{=} U^{1/n} \bz'$, where $U\sim \mathcal{U}[0,1]$, where $\bz'$ is uniformly distributed on the unit sphere in $\R^n$, and where $U$ and $\bz'$ are independent. 
Thus, 
\begin{equation}
  \label{eqn:fbar-property-a-eqn3}
  \E(\norm{\bm z}_\infty) = \E(U^{1/n})\E(\norm{\bm z'}_{\infty})=
  \frac{n}{n+1} \E(\norm{\bm z'}_{\infty}).
\end{equation}
Moreover, if $\bm\zeta\sim\mathcal{N}(0,I_n)$, then $\norm{\bm\zeta}$ and $\bm\zeta/\norm{\bm\zeta}$ are independent, and $\bz'\overset{d}{=}\bm\zeta/\norm{\bm\zeta}$.  It follows that
\begin{equation}
\E(\norm{\bm z'}_{\infty})\! =\! \mathbb{E}\biggl(\frac{\|\bm\zeta\|_\infty}{\|\bm\zeta\|}\biggr) \cdot \frac{\mathbb{E}(\|\bm\zeta\|)}{\mathbb{E}(\|\bm\zeta\|)} \!=\!\frac{\mathbb{E}(\|\bm\zeta\|_\infty)}{\mathbb{E}(\|\bm\zeta\|)} = \frac{\sqrt{2\log n} \cdot \Gamma(n/2)}{2^{1/2}\Gamma\bigl((n+1)/2\bigr)} \leq \frac{1}{n}\sqrt{2(n+1)\log n},\label{eqn:fbar-property-a-eqn4}\end{equation}
where the final bound follows from bounds on the gamma function, e.g. \cite[Lemma~12]{duembgen2019bounding}.  The result follows from~\eqref{eqn:fbar-property-a-eqn3} and~\eqref{eqn:fbar-property-a-eqn4}.
\end{proof}

\begin{proof}[Proof of Proposition~\ref{prop:fbar-properties}]
(a) By Jensen's inequality,  
\begin{equation}\label{ineq-1-prop4-proof}
\bar{f}_u(\bphi)=\E f(\bphi+u\bm z) \geq f(\bphi).
\end{equation}
For the upper bound, let $\bm v\in \R^n$ have $\|\bm v\| \leq 1$ and, for some $\bphi\in\bm{\Phi}$, let $\bar\bphi := \bphi+u\bm v$. For any $\bm\alpha\in E(\bx)$, we have
\[
  |\bm\alpha^\top\bar\bphi-\bm\alpha^\top\bphi|=u|\bm\alpha^\top\bm v|\leq u \norm{\bm\alpha}_1\norm{\bm v}_{\infty}\leq u.
\]
Therefore, for any $\bx\in C_n$, we have $\cef[\bar\bphi](\bx)\geq \cef[\bphi](\bx)-u.$
Hence
\begin{equation}\label{proof-prop4-lineiphi-1}
I(\bar\bphi)=\Delta\E[\exp\{-\cef[\bar\bphi](\bm\xi)\}]\leq e^u\Delta\E[\exp\{-\cef[\bphi](\bm\xi)\}]=e^uI(\bphi).
\end{equation}
Recall that all subgradients of $I$ at $\tilde{\bphi} \in \R^n$ are of the form $-\bm\gamma(\tilde{\bphi})$, where
\[
  \bm\gamma(\tilde{\bphi}) = \Delta\E\bigl(\bm\alpha\exp\{-\cef[\tilde{\bphi}](\bm\xi)\}\bigr)
\]
for some $\bm\alpha \in A[\tilde{\bphi}](\bx)$.  Moreover, as we saw in the proof of Proposition~\ref{prop:bound-subgradient}, $\|\bm\gamma(\tilde{\bphi})\|_1=I(\tilde{\bphi})$. 
We deduce from \citet[Theorem~24.7]{rockafellar1997convex} that
\begin{align*}
  \label{eqn:fbar-property-a-eqn2}
  \bar{f}_u(\bphi)-f(\bphi) = \E\bigl(I(\bphi+u\bz)-I(\bphi)\bigr)& \leq u\sup_{\bphi \in \bm \Phi,\|\bm v\| \leq 1} I(\bphi + u\bm v)\mathbb{E}\bigl(\|\bm z\|_\infty\bigr)  \nonumber \\
  &\leq I(\bphi)ue^u\E\bigl(\norm{\bz}_\infty\bigr)\leq I(\bphi)ue^u\sqrt{\frac{2\log n}{n+1}},
  \end{align*}
  where the final inequality uses Proposition~\ref{proof-expected-prop6}.
\medskip

\noindent (b) 
By the convexity of $f$, we have
\[
  \bar{f}_{u'}(\bphi)= \E\bigl(f(\bphi+u'\bz)\bigr) \leq \frac{u'}{u}\bar{f}_{u}(\bphi) +\biggl(1-\frac{u'}{u}\biggr)f(\bphi) \leq \bar{f}_u(\bphi),
\]
where the last inequality uses property~(a).

\medskip

\noindent (c) For each $\bm v \in \R^n$ with $\|\bm v\| = 1$, the map  $\bphi \mapsto f(\bphi+u\bm v)$ is convex, so $\bphi \mapsto \E\bigl(f(\bphi+u\bz)\bigr) = \bar{f}_u(\bphi)$ is convex.  The proof of the Lipschitz property is very similar to that of~Proposition~\ref{prop:fhat-properties}(b) and is omitted for brevity.

\medskip

\noindent (d) As in the proof of~(a), for any $\bm v \in \R^n$ with $\|\bm v\| \leq 1$, $\bx \in C_n$ and $\bm \alpha \in A[\bphi+u\bm v](\bx)$, we have
\begin{equation*}\label{eqn:fbar-prop-proof-d-1}
  \bigl\|\bm\alpha e^{-\cef[\bphi+u\bm v](\bx)}\bigr\| \leq e^{-\phi_0+u}.
\end{equation*}
Since $\nabla_{\bphi} \bar{f}_u(\bphi)=n^{-1}\bm1 -\Delta \E\bigl(\bm \alpha^*[\bphi+u\bz](\bm\xi) e^{-\cef[\bphi+u\bz](\bm\xi)}\bigr)$, where $\bm \alpha^*[\bphi+u\bm v](\bx) \in A[\bphi+u\bm v](\bx)$, we have by \citet[Lemma~8]{yousefian2012stochastic} that $\nabla_{\bphi} \bar{f}_u$ is $\Delta e^{-\phi_0+u}n^{1/2}/u$-Lipschitz. 


\medskip

\noindent (e) The proof is very similar to the proof of Proposition~\ref{prop:fhat-properties}(d) and is omitted for brevity.
\end{proof}

\subsection{Proofs of Theorem~\ref{thm:smoothing-guarantee-expectation} and Theorem~\ref{thm:smoothing-guarantee-highprob}}
We will make use of the following lemma:
\begin{lemma}[Lemma~4.2 of \cite{duchi2012randomized}]\label{lem:smoothing-key-lemma}
  Let $\bigl(\ell_{u_t}(\bphi)\bigr)_t$ be a smoothing sequence such that $\bphi \mapsto \ell_{u_t}(\bphi)$ has $L_t$-Lipschitz gradient. Assume that $\ell_{u_t}(\bphi)\leq \ell_{u_{t-1}}(\bphi)$ for $\bphi\in\bm{\Phi}$. Let $(\bphi_t^{(x)})_{t=0}^T, (\bphi^{(y)}_t)_{t=0}^T, (\bphi^{(z)}_t)_{t=0}^T$ be the sequences generated by Algorithm~\ref{algo:smoothing}.  Let $\bm g_t$ denote an approximation of $\nabla \ell_{u_t}(\bphi_t^{(y)})$ with error $\bm e_t= \bm g_t - \nabla \ell_{u_t}(\bphi_t^{(y)})$.  Then for any $\bphi\in\bm{\Phi}$ and $t \in \mathbb{N}$, we have
\begin{align*}
    \frac{1}{\theta_t^2}\ell_{u_t}(\bphi^{(x)}_{t+1})\leq \sum_{\tau=0}^t\frac1{\theta_\tau}\ell_{u_\tau}(\bphi)&+\frac{1}{2}\biggl(L_{t+1}+\frac{\eta_{t+1}}{\theta_{t+1}}\biggr)\norm{\bphi-\bphi_0}^2\\
    &+\sum_{\tau=0}^t\frac{\norm{\bm e_\tau}^2}{2\theta_\tau\eta_\tau}+\sum_{\tau=0}^t\frac{1}{\theta_\tau}\langle \bm e_\tau,\bphi - \bphi^{(z)}_\tau\rangle.
\end{align*}
\end{lemma}
Recall the definition of the diameter $D$ of $\bm \Phi$ given just before Theorem~\ref{thm:smoothing-guarantee-expectation}.
\begin{cor}\label{cor:smoothing-key-lemma-cor}
Fix $u, \eta > 0$, and assume that Assumption~\ref{assump:smoothing-seq} holds with $r \geq u$.  Suppose in Algorithm~\ref{algo:smoothing} that $u_t=\theta_tu$, $L_t={B_1}/{u_t}$ and $\eta_t=\eta$.  Let $(\bphi_t^{(x)})_{t=0}^T, (\bphi^{(y)}_t)_{t=0}^T, (\bphi^{(z)}_t)_{t=0}^T$ be the sequences generated by Algorithm~\ref{algo:smoothing} and let $\bm e_t= \bm g_t - \nabla \ell_{u_t}(\bphi_t^{(y)})$.  Then for any $\bphi \in \bm{\Phi}$, we have 
\begin{align*}
    f(\bphi_T^{(x)})-f(\bphi)\leq  \frac{B_1D^2}{Tu}+\frac{\eta D^2}{T}+\frac{1}{T\eta}\sum_{t=0}^{T-1}\norm{\bm e_t}^2 + \theta_{T-1}^2\sum_{t=0}^{T-1}\frac{1}{\theta_t}\langle \bm e_t,\bphi \!-\! \bphi_t^{(z)}\rangle+\frac{4B_0I(\bphi)u}{T}.
\end{align*}
\end{cor}
\begin{proof}
By induction, we have that $\theta_t\leq {2}/{(t+2)}$ and $\sum_{\tau=0}^t {1}/{\theta_\tau}={1}/{\theta_t^2}$ for all $t = 0,1,\ldots,T$ \cite{tseng2008accelerated,duchi2012randomized}. 
Using Assumption~\ref{assump:smoothing-seq}, we have 
\begin{align*}
\frac{1}{\theta_{T-1}^2}\bigl(f(\bphi_T^{(x)})-f(\bphi)\bigr)    &\leq \frac{1}{\theta_{T-1}^2}\ell_{u_{T-1}}(\bphi^{(x)}_T)-\sum_{t=0}^{T-1}\frac{1}{\theta_t} \bigl(\ell_{u_t}(\bphi)-B_0u_t\bigr)\nonumber\\
    &= \frac{1}{\theta_{T-1}^2}\ell_{u_{T-1}}(\bphi_T^{(x)})-\sum_{t=0}^{T-1}\frac{1}{\theta_t}\ell_{u_t}(\bphi)+TB_0I(\bphi)u.\label{eqn:smoothing-cor-proof-eqn1}
\end{align*}
Hence, by Lemma~\ref{lem:smoothing-key-lemma},
\begin{align*}
    f(\bphi_T^{(x)})&-f(\bphi) \nonumber \\
    \leq &~ \frac{\theta_{T-1}^2}{2}\norm{\bphi-\bphi_0}^2\biggl(L_T + \frac{\eta_T}{\theta_T}\biggr)+\sum_{t=0}^{T-1}\frac{\theta_{T-1}^2}{2\theta_t\eta_t}\norm{\bm e_t}^2\\
    &\hspace{4.5cm}+\theta_{T-1}^2\sum_{t=0}^{T-1}\frac{1}{\theta_t}\langle \bm e_t,\bphi - \bphi^{(z)}_t\rangle+\theta_{T-1}^2TB_0I(\bphi)u \nonumber  \\
    \leq &~\frac{B_1D^2\theta_{T-1}^2}{2u\theta_T}+\frac{\eta_T\theta_{T-1}^2D^2}{2\theta_T}+\sum_{t=0}^{T-1}\frac{\theta_{T-1}^2}{2\theta_t\eta_t}\norm{\bm e_t}^2\\
    &\hspace{4.5cm}+\theta_{T-1}^2 \sum_{t=0}^{T-1} \frac{1}{\theta_t} \langle \bm e_t,\bphi - \bphi^{(z)}_t\rangle+\theta_{T-1}^2TB_0I(\bphi)u \\
    \leq&~\frac{B_1D^2}{Tu}+\frac{\eta_TD^2}{T}+\frac{1}{T}\sum_{t=0}^{T-1}\frac{1}{\eta_t}\norm{\bm e_t}^2+\theta_{T-1}^2\sum_{t=0}^{T-1}\frac{1}{\theta_t}\langle \bm e_t,\bphi - \bphi^{(z)}_t\rangle+\frac{4B_0I(\bphi)u}{T}, 
\end{align*}
where we have used the facts that $\theta_{T-1}^2/\theta_T = \theta_T/(1-\theta_T) \leq 2/T$ and $\theta_{T-1}^2/\theta_t \leq \theta_{T-1} \leq 2/T$ for $t \in \{0,1,\ldots,T-1\}$ and $\theta_{T-1}^2\leq 4/T^2$. 
\end{proof}

\begin{proof}[Proof of Theorem~\ref{thm:smoothing-guarantee-expectation}]
According to Corollary~\ref{cor:smoothing-key-lemma-cor}, it suffices to bound  $\theta_{T-1}^2\sum_{t=0}^{T-1}\frac{1}{\theta_t}\E\bigl(\langle \bm e_t,\bphi - \bphi^{(z)}_t\rangle\bigr)$ and $\sum_{t=0}^{T-1}\E\bigl(\norm{\bm e_t}^2\bigr)$.  To this end, we have by Assumption~\ref{assump:error} that
\begin{align}
    \E\bigl(\langle \bm e_t,\bphi - \bphi^{(z)}_t\rangle\bigr) &= \E\bigl(\E(\langle \bm e_t,\bphi - \bphi^{(z)}_t\rangle \bigm| \mathcal{F}_{t-1})\bigr) \leq \E\bigl(\E(\|\bm e_t\|\|\bphi - \bphi^{(z)}_t\| \mid \mathcal{F}_{t-1})\bigr)\nonumber\\
    &\leq D\cdot \E\bigl(\E(\norm{\bm e_t}\mid\mathcal{F}_{t-1})\bigr) \leq D \cdot \E\Bigl(\sqrt{\E\bigl(\|\bm e_t\|^2\bigm| \mathcal{F}_{t-1}\bigr)}\Bigr) \leq \frac{D\sigma}{\sqrt{m_t}}.
\end{align}
We deduce that
\begin{equation}
  \theta_{T-1}^2\sum_{t=0}^{T-1}\frac{1}{\theta_t}\E\bigl(\langle \bm e_t,\bphi - \bphi_t^{(z)}\rangle \bigr) \leq \frac{2D\sigma M_T^{(1/2)}}{T}.\label{eqn:smoothing-thm-expectation-proof-eqn2}
\end{equation}
Moreover, by Assumption~\ref{assump:error} again, 
\begin{equation}
  \sum_{t=0}^{T-1}\E\bigl(\|\bm e_t\|^2\bigr)\leq \sigma^2\sum_{t=0}^{T-1} \frac{1}{m_t} = \sigma^2(M_T^{(1)})^2.\label{eqn:smoothing-thm-expectation-proof-eqn3}
\end{equation}
The bound~\eqref{thm1-state-1} follows from Corollary~\ref{cor:smoothing-key-lemma-cor}, together with~\eqref{eqn:smoothing-thm-expectation-proof-eqn2} and~\eqref{eqn:smoothing-thm-expectation-proof-eqn3}, and the bound~\eqref{thm1-state-2} then follows directly from the parameter choice of $u$ and $\eta$ and the fact that $I(\bphi^*)=1$.

Finally, if $\E(\bm e_t\mid \mathcal{F}_{t-1})=\bm0$, then
\begin{align*}
    \E\bigl(\langle \bm e_t,\bphi - \bphi_t^{(z)}\rangle\bigr)=\E\bigl(\E(\langle \bm e_t,\bphi - \bphi_t^{(z)}\rangle\mid\mathcal{F}_{t-1})\bigr)=\E\bigl(\langle \E(\bm e_t|\mathcal{F}_{t-1}),\bphi- \bphi_t^{(z)}\rangle\bigr)=0
\end{align*}
where the second equality uses the fact that $\bphi - \bphi_t^{(z)}$ is $\mathcal{F}_{t-1}$-measurable.  This allows us to remove the last term of the two inequalities in the theorem.
\end{proof}

\begin{proof}[Proof of Theorem~\ref{thm:smoothing-guarantee-highprob}]
According to Corollary~\ref{cor:smoothing-key-lemma-cor}, it suffices to obtain a high-probability bound for $\theta_{T-1}^2\sum_{t=0}^{T-1}\frac1{\theta_t}\langle \bm e_t,\bphi-\bphi_t^{(z)}\rangle$ and $\sum_{t=0}^{T-1}\norm{\bm e_t}^2$.  Writing $\bm\zeta_t:=\bphi - \bphi_t^{(z)}$, we have from the proof of Theorem~\ref{thm:smoothing-guarantee-expectation} that $(1/{\theta_t})\langle \bm e_t,\bm\zeta_t\rangle$ is a martingale difference sequence under Assumption~\ref{assump:error-subGaussian}.  Note that $\bm\zeta_t$ is $\mathcal{F}_{t-1}$-measurable, and we will now show that $\langle\bm e_t, \bm\zeta_t\rangle$ is $\sqrt{2}\sigma_t D$ sub-Gaussian, conditional on $\mathcal{F}_{t-1}$.

For any $x\in\R$, we have $e^x\leq x+e^{x^2}$. Hence, for $\lambda \in \R$ such that $\lambda^2\sigma_t^2D^2\leq 1$, we have by the conditional version of Jensen's inequality that
\begin{align*}
    \E\bigl(e^{\lambda\langle \bm e_t,\bm\zeta_t\rangle}\mid\mathcal{F}_{t-1}\bigr) &\leq \E\bigl(\lambda\langle \bm e_t,\bm\zeta_t\rangle\mid\mathcal{F}_{t-1}\bigr)+\E\bigl(e^{\lambda^2\langle \bm e_t,\bm\zeta_t\rangle^2}\mid\mathcal{F}_{t-1}\bigr)\\ &\leq \E\bigl(e^{\lambda^2\norm{\bm e_t}^2D^2}|\mathcal{F}_{t-1}\bigr) \leq e^{\lambda^2\sigma_t^2D^2}.
\end{align*}
On the other hand, if $\lambda^2\sigma_t^2D^2>1$, then since $2ab \leq a^2 + b^2$ for all $a,b \in \R$, we have
\begin{align*}
    \E\bigl(e^{\lambda\langle \bm e_t,\bm\zeta_t\rangle}\mid\mathcal{F}_{t-1} \bigr)&\leq e^{\lambda^2\sigma_t^2D^2/2}\E\bigl(e^{\langle\bm e_t,\bm\zeta_t\rangle^2/(2\sigma_t^2D^2)}\mid\mathcal{F}_{t-1}\bigr)\\
    &\leq e^{\lambda^2\sigma_t^2D^2/2}\E\bigl(e^{\norm{\bm e_t}^2/(2\sigma_t^2)}\mid\mathcal{F}_{t-1}\bigr) \leq e^{\lambda^2\sigma_t^2D^2/2}e^{1/2} \leq e^{\lambda^2\sigma_t^2D^2}.
\end{align*}
We deduce that $\langle \bm e_t,\bm\zeta_t\rangle/\theta_t$ is ${(\sqrt{2}\sigma_tD)}/{\theta_t}$ sub-Gaussian, conditional on $\mathcal{F}_{t-1}$.  Applying the Azuma--Hoeffding inequality \citep[e.g.][]{azuma1967weighted} therefore yields that for every $\epsilon > 0$,
\begin{align*}
\p\biggl(\theta_{T-1}^2\sum_{t=0}^{T-1}\frac1{\theta_t}\langle \bm e_t,\bm\zeta_t\rangle \geq \epsilon\biggr)& \leq \exp\biggl(-\frac{\epsilon^2}{4D^2\theta_{T-1}^2\sum_{t=0}^{T-1} \sigma_t^2\theta_{T-1}^2/\theta_t^2}\biggr)\\
&\leq \exp\biggl(-\frac{T^2\epsilon^2}{16D^2\sigma^2(M_T^{(1)})^2}\biggr),
\end{align*}
where the last inequality uses the facts that $\theta_{T-1}\leq \theta_t$ and $\theta_{T-1} \leq {2}/{T}$.  Therefore, for every $\delta \in (0,1)$, we have with probability at least $1-\delta/2$ that
\begin{equation}\label{eqn:smoothing-thm-highprob-proof-eqn1}\theta_{T-1}^2\sum_{t=0}^{T-1}\frac1{\theta_t}\langle \bm e_t,\bm\zeta_t\rangle \leq  \frac{4\sigma DM_T^{(1)}\sqrt{\log(2/\delta)}}{T}.\end{equation}

Next we will turn to finding a tail bound for $\sum_{t=0}^{T-1}\norm{\bm e_t}^2$. By Assumption~\ref{assump:error-subGaussian} and Jensen's inequality, we have
\begin{equation}\label{eqn:smoothing-thm-highprob-proof-eqn2}
  \sum_{t=0}^{T-1}\E\bigl(\|\bm e_t\|^2\bigm|\mathcal{F}_{t-1}\bigr) \leq \sum_{t=0}^{T-1}\sigma^2_t\log\bigl(\E(e^{\|\bm e_t\|^2/\sigma_t^2}|\mathcal{F}_{t-1})\bigr) \leq \sigma^2(M_T^{(1)})^2.
\end{equation}
Now define the random variables $\Xi_t:=\|\bm e_t\|^2-\E\bigl(\|\bm e_t\|^2\bigm|\mathcal{F}_{t-1}\bigr)$.  
Then by Markov's inequality, for every $\epsilon > 0$,
\[
  \mathbb{P}\bigl(\Xi_t > \epsilon \bigm| \mathcal{F}_{t-1}\bigr)\! \leq\! \mathbb{P}\bigl(\|\bm e_t\|^2/\sigma_t^2 > \epsilon/\sigma_t^2 \bigm| \mathcal{F}_{t-1}\bigr)\! \leq\! e^{-\epsilon/\sigma_t^2}\mathbb{E}\bigl(e^{\|\bm e_t\|^2/\sigma_t^2} \bigm| \mathcal{F}_{t-1}\bigr) \leq e^{1 - \epsilon/\sigma_t^2}.
\]
Moreover, by Markov's inequality again, and then Jensen's inequality, we have for every $\epsilon > 0$ that
\begin{align*}
  \mathbb{P}\bigl(\Xi_t < -\epsilon \bigm| \mathcal{F}_{t-1}\bigr) &\leq e^{-\epsilon/\sigma_t^2} \mathbb{E}\bigl(e^{\mathbb{E}(\|\bm e_t\|^2/\sigma_t^2| \mathcal{F}_{t-1}) - \|\bm e_t\|^2/\sigma_t^2} \bigm| \mathcal{F}_{t-1}\bigr)\\ &\leq e^{-\epsilon/\sigma_t^2} e^{\mathbb{E}(\|\bm e_t\|^2/\sigma_t^2| \mathcal{F}_{t-1})} \leq e^{1-\epsilon/\sigma_t^2}.
\end{align*}
It follows by, e.g., \citet[Lemma~F.7]{duchi2012randomized} that $\Xi_t$ is sub-exponential with parameters $\lambda_t:=1/(2\sigma_t^2)={m_t}/{(2\sigma^2)}$ and $\tau_t^2 := 16e\sigma_t^4={16e\sigma^4}/{m_t^2}$, in the sense that 
\begin{equation}
\label{Eq:BaseCase}
  \E\bigl(e^{\lambda \Xi_t}\bigm|\mathcal{F}_{t-1}\bigr) \leq e^{8e\lambda^2\sigma_t^4},
\end{equation}
for $|\lambda| \leq 1/(2\sigma_t^2)$.

Now define $\Lambda_T := \min_{t=0, \ldots, T-1} \lambda_t={m_0}/{(2\sigma^2)}$ (as we assume $(m_t)$ is increasing) and $C_T := \bigl(\sum_{t=0}^{T-1}\tau_t^2\bigr)^{1/2}=4e^{1/2}\sigma^2M_T^{(2)}$.  We claim that $\sum_{t=0}^{T-1} \Xi_t$ is sub-exponential with parameters $\Lambda_T$ and $C_T$, and prove this by induction on $T$.  The base case $T=1$ holds by~\eqref{Eq:BaseCase}, so suppose it holds for a given $T \in \mathbb{N}$.  Then for $\lambda \in \mathbb{R}$ with $|\lambda| \leq \min(\Lambda_T,\lambda_T) = \Lambda_{T+1}$, we have
\begin{align*}
\mathbb{E}\biggl\{\exp\biggl(\lambda\sum_{t=0}^T \Xi_t\biggr)\biggr\} &= \mathbb{E}\biggl[\exp\biggl(\lambda\sum_{t=0}^{T-1} \Xi_t\biggr)\mathbb{E}\bigl\{\exp\bigl(\lambda\Xi_T|\mathcal{F}_{T-1}\bigr)\bigr\}\biggr]\\
&\leq e^{(\lambda^2C_T^2 + 16e\lambda^2\sigma_T^4)/2} = e^{\lambda^2 C_{T+1}^2/2},
\end{align*}
which proves the claim by induction.  We deduce by, e.g. \citet[Lemma~1.4.1]{buldygin2000metric}, that for every $\epsilon > 0$ and $T \in \mathbb{N}$,
\[
  \p\biggl(\sum_{t=0}^{T-1} \Xi_t\geq \epsilon\biggr)\leq \exp\biggl(-\min\biggl\{\frac{\epsilon^2}{2C_T^2},\frac{\Lambda_T\epsilon}{2}\biggr\}\biggr).
\]
In other words, with probability at least $1-\delta/2$,
\begin{equation}
  \sum_{t=0}^{T-1}\Xi_t\leq 4\sigma^2\max\biggl\{M_T^{(2)}\sqrt{2e\log\frac2\delta},\frac{1}{m_0}\log\frac2\delta\biggr\}.\label{eqn:smoothing-thm-highprob-proof-eqn3}
\end{equation}
Applying~\eqref{eqn:smoothing-thm-highprob-proof-eqn1}, \eqref{eqn:smoothing-thm-highprob-proof-eqn2} and \eqref{eqn:smoothing-thm-highprob-proof-eqn3} in Corollary~\ref{cor:smoothing-key-lemma-cor}, together with a union bound, yields that with probability at least $1-\delta$, 
\begin{align*}
  f(\bphi_T^{(x)})-f(\bphi)&\leq \frac{B_1D^2}{Tu}+\frac{4B_0I(\bphi)u}{T}+\frac{\eta D^2}{T}+\frac{\sigma^2(M_T^{(1)})^2}{T\eta}+\frac{4\sigma DM_T^{(1)}\sqrt{\log(2/\delta)}}{T}\\
                           &\hspace{3.4cm}+\frac{4\sigma^2\max\bigl\{M_T^{(2)}\sqrt{2e\log(2/\delta)},m_0^{-1}\log(2/\delta)\bigr\}}{T\eta}.
\end{align*}
Taking the same choices of $\bm \phi$, $u$ and $\eta$ as  in Theorem~\ref{thm:smoothing-guarantee-expectation}, we obtain the final result.
\end{proof}


\subsection{Proof of Theorem~\ref{thm:error-in-cef}}
To prove Theorem~\ref{thm:error-in-cef}, we first introduce the following lemma.  Recall that $\mathcal{C}_d$ denotes the class of proper, convex lower-semicontinuous functions $\varphi:\R^d \rightarrow (-\infty,\infty]$ that are coercive in the sense that $\varphi(\bx) \rightarrow \infty$ as $\|\bx\| \rightarrow \infty$.  Recall further from \citet[Theorem~2.2]{dumbgen2011approximation} that if $P$ is a distribution on $\R^d$ with $\int_{\mathbb{R}^d} \|\bx\| \, dP(\bx) < \infty$ and $P(H) < 1$ for all hyperplanes $H$, then the strictly convex function $\Gamma:\mathcal{C}_d \rightarrow (-\infty,\infty]$ given by
\begin{equation}
\label{lemma2-main-defn-1}
  \Gamma(\varphi) := \int_{\R^d} \varphi(\bx)\,dP(\bx)+\int_{\R^d}e^{-\varphi(\bx)}\,d\bx
\end{equation}
has a unique minimizer $\varphi^* \in \mathcal{C}_d$ satisfying $\Gamma(\varphi^*) \in \R$.
\begin{lemma}\label{lem:strong-convexity-with respect to-varphi}
  Let $P$ be a distribution on $\R^d$ with $\int_{\mathbb{R}^d} \|x\| \, dP(x) < \infty$ and $P(H) < 1$ for all hyperplanes~$H$, and let $\varphi^* := \argmin_{\varphi \in \mathcal{C}_d} \Gamma(\varphi)$.  Then \\
{\bf (1)} For any $\lambda \in [0,1]$, and $\varphi,\tilde{\varphi}\in \mathcal{C}_d$, we have
\begin{equation}\label{eqn:strong-convexity-with respect to-varphi}
    \Gamma\bigl(\lambda\varphi+(1-\lambda)\tilde\varphi\bigr) \leq \lambda \Gamma(\varphi)+(1-\lambda)\Gamma(\tilde\varphi)-\frac{\lambda(1-\lambda)}{2}\int_{\R^d}e^{-\max\{\varphi(\bx),\tilde\varphi(\bx)\}}\bigl\{\varphi(\bx)-\tilde\varphi(\bx)\bigr\}^2\,d\bx.
\end{equation}
Here, when $\max\bigl\{\varphi(\bx),\tilde\varphi(\bx)\bigr\} =\infty$, we define the integrand to be zero.\\
{\bf (2)} Furthermore, if $\varphi \in \mathcal{C}_d$ is such that $\max\{\varphi(\bx),\varphi^*(\bx)\}\leq \phi^0$ for all $\bx \in \dom{\varphi}\cap \dom{\varphi^*}$, then
\begin{equation}\label{eqn:error-Gamma-varphi}
    \Gamma(\varphi)-\Gamma(\varphi^*)\geq \frac12e^{-\phi^0}\int_{\dom{\varphi}\cap \dom{\varphi^*}}\bigl\{\varphi(\bx)-\varphi^*(\bx)\bigr\}^2\,d\bx.
\end{equation}
\end{lemma}

\begin{proof}

\textbf{(1)} Fix $\varphi,\tilde{\varphi}\in \mathcal{C}_d$ with $\max\bigl\{\Gamma(\varphi),\Gamma(\tilde{\varphi})\bigr\} < \infty$ (because otherwise the result is clear).  For any $M \in \mathbb{R}$, the function $y \mapsto e^{-y}$ is $e^{-M}$-strongly convex on $y \leq M$. 
Therefore, for any $\lambda \in [0,1]$, we have $\Gamma\bigl(\lambda\varphi+(1-\lambda)\tilde\varphi\bigr) \geq \Gamma(\varphi^*) > -\infty$, so
\begin{align}
    \Gamma\bigl(\lambda\varphi+(1-\lambda)\tilde\varphi\bigr)-\lambda\Gamma(\varphi)&-(1-\lambda)\Gamma(\tilde\varphi)\nonumber\\
    &= \int_{\R^d} \bigl[e^{-\{\lambda \varphi(\bx)+(1-\lambda)\tilde\varphi(\bx)\}}-\lambda e^{-\varphi(\bx)}-(1-\lambda)e^{-\tilde\varphi(\bx)}\bigr]\,d\bx\nonumber\\
    &\leq -\frac{\lambda(1-\lambda)}{2}\int_{\R^d}e^{-\max\{\varphi(\bx),\tilde
    \varphi(\bx)\}}\bigl\{\varphi(\bx)-\varphi^*(\bx)\bigr\}^2\,d\bx, \label{thm3-lemma-last-line1}
\end{align}
as required.

\medskip

\textbf{(2)}  
By~\eqref{thm3-lemma-last-line1}, we have for any $\lambda \in (0,1)$ that
\begin{align*}
  \frac{\lambda(1-\lambda)}{2}e^{-\phi^0}\int_{\dom{\varphi}\cap \dom{\varphi^*}} &\bigl\{\varphi(\bx)-\varphi^*(\bx)\bigr\}^2\dx \\
  &\leq \lambda\Gamma(\varphi)+(1-\lambda)\Gamma(\varphi^*)-\Gamma\bigl(\lambda\varphi+(1-\lambda)\varphi^*\bigr) \\
  &\leq \lambda\bigl\{\Gamma(\varphi)-\Gamma(\varphi^*)\bigr\},
\end{align*}
where the last inequality follows by definition of $\varphi^*$. We deduce that
\[
  \frac{1}{2}e^{-\phi^0}\int_{\dom{\varphi}\cap \dom{\varphi^*}} \bigl\{\varphi(\bx)-\varphi^*(\bx)\bigr\}^2\dx\leq \frac{\Gamma(\varphi)-\Gamma(\varphi^*)}{1-\lambda}.
\]
The result follows on taking $\lambda \searrow 0$.
\end{proof}

\begin{proof}[Proof of Theorem~\ref{thm:error-in-cef}]
In Lemma~\ref{lem:strong-convexity-with respect to-varphi}, let $P$ be the empirical distribution of $\{\bx_i\}_{i=1}^n$.  From the proof of Theorem~2 of~\cite{cule2010maximum} (see also the proof of Theorem~\ref{thm:s-concave-MLE}), 
Problem~\eqref{lemma2-main-defn-1} is equivalent to \eqref{eqn:MLE-phi} in the sense that $\Gamma(\cef[\bphi]) = f(\bphi)$ for all $\bphi \in \R^n$, and $\varphi^* = \cef[\bphi^*]$.  Now fix $\bphi\in \bm{\Phi}$ and let $\varphi=\cef[\bphi]$, so that $\varphi(\bx)\leq \phi^0$ for all $\bx \in C_n$.  From~\eqref{eqn:error-Gamma-varphi}, we obtain
\[
  \frac{1}{2}e^{-\phi^0}\int_{C_n} \bigl\{\cef[\bphi](\bx)-\cef[\bphi^*](\bx)\bigr\}^2 \, d\bx \leq \Gamma(\cef[\bphi])-\Gamma(\cef[\bphi^*])=f(\bphi)-f(\bphi^*),
\]
as required.\end{proof}


\subsection{Proofs of Theorem~\ref{thm:s-concave-MLE} and Theorem~\ref{thm:quasi-concave}}
 The proof of Theorem~\ref{thm:s-concave-MLE} is based on following proposition:
\begin{proposition}\label{prop:psi-s-scale}
Given $s\in \R$, $\varphi\in\mathcal{C}_d$ and $c>0$, there exists $\tilde{\varphi}\in\mathcal{C}_d$, such that
$\psi_s\circ\varphi(\cdot) = c\cdot\psi_s\circ\tilde\varphi(\cdot).$
\end{proposition}
\begin{proof}[of Proposition~\ref{prop:psi-s-scale}]
Given $s$, $\varphi$ and $c$, define 
\[
\tilde{\varphi}(\cdot) :=\left\{\begin{array}{ll}
     \varphi(\cdot)+\log c& ~\text{if}~ s=0 \\
     c^{-s}\varphi(\cdot)& ~\text{if}~ s\neq 0. 
                                \end{array}\right.
\]
Then $\tilde\varphi\in\mathcal{C}_d$, and for $\bx \in \mathcal{D}_s$,
\[
\psi_s\circ \tilde{\varphi}(\bx)=\left\{\begin{array}{ll}
     \bigl(c^{-s}\varphi(\bx)\bigr)^{1/s}& ~\text{if}~ s<0 \\
     \exp\bigl(-\varphi(\bx)-\log c\bigr)& ~\text{if}~ s=0 \\
     \bigl(-c^{-s}\varphi(\bx)\bigr)^{1/s}& ~\text{if}~ s> 0 
                                          \end{array}\right.=c^{-1}\cdot\psi_s\circ\varphi(\bx),
\]
as required.                                      
\end{proof}

\begin{proof}[Proof of Theorem~\ref{thm:s-concave-MLE}] The proof is split into four steps. The first three steps hold for any $s \in \R$, while in Step~4, we show the convexity of the objective when $s\in[0,1]$.

\medskip  

\noindent {\bf Step 1:} We claim that any solution $\hat{p}_n$ to~\eqref{eqn:MLE-psi-functional} is supported on $C_n$, so that~$\varphi^*(\bx)=\infty$ when $\bx\notin C_n$, where $\varphi^*$ is the solution to \eqref{eqn:MLE-psi-varphi}. Indeed, suppose for a contradiction that $p=\psi_s\circ\varphi\in\mathcal{P}_s(\R^d)$ is such that $\sum_{i=1}^n\log p(\bx_i) > -\infty$, and that $\int_{C_n}p(\bx)\dx=c<1 = \int_{\R^d}p(\bx)\dx$.  We may assume that $c > 0$, because otherwise $p(\bx) = 0$ for almost all $\bx\in C_n$, which would mean that $\sum_{i=1}^n\log p(\bx_i) = -\infty$ since $p\in\mathcal{P}_s(\R^d)$.  Define $\bar\varphi \in \mathcal{C}_d$ by
\[
\bar\varphi(\bx) := \left\{\begin{array}{ll}
     \varphi(\bx)&~\text{if}~ \bx\in C_n  \\
     \infty&~\text{otherwise,}
                           \end{array}\right.
\]
so that $\int_{\R^d}\psi_s\circ\bar\varphi(\bx)\dx=\int_{\R^d}\psi_s\circ\varphi(\bx)\dx=c<1$. Applying Proposition~\ref{prop:psi-s-scale} to $\bar\varphi\in\mathcal{C}_d$ and $c>0$, we can find $\tilde \varphi\in \mathcal{C}_d$ with $\int_{\R^d}\psi_s\circ\tilde\varphi(\bx)\dx= \int_{\R^d}c^{-1}\cdot\psi_s\circ\bar\varphi(\bx)\dx=1$, and 
\[
  \sum_{i=1}^n\log \psi_s\circ\tilde\varphi(\bx_i)=\sum_{i=1}^n\log \psi_s\circ\varphi(\bx_i)-n\log c> \sum_{i=1}^n\log \psi_s\circ\varphi(\bx_i).
\]
This establishes our desired contradiction. 

\noindent {\bf Step 2:} We claim that any solution $\varphi^*$ to~\eqref{eqn:MLE-psi-varphi} satisfies
\begin{equation}\label{eqn:MLE-psi-step2}
\varphi^* =  \argmin_{\substack{\varphi \in\mathcal{C}_d: \mathrm{Im}(\varphi) \subseteq \mathcal{D}_s\cup\{\infty\},\\\dom{\varphi}=C_n}} \biggl\{-\frac{1}{n}\sum_{i=1}^n\log\psi_s\circ\varphi(\bx_i)+\int_{C_n}\psi_s\circ\varphi(\bx)\dx\biggr\}.
\end{equation}
Indeed, for any $\varphi\in\mathcal{C}_d$ such that $\dom{\varphi}=C_n$ and $\int_{C_n}\psi_s\circ\varphi(\bx)\dx=c\neq 1$, we can again apply Proposition~\ref{prop:psi-s-scale} to $\varphi$ and $c$ to obtain $\tilde\varphi$.  Then
\begin{align*}
  -\frac{1}{n}&\sum_{i=1}^n\log\psi_s\circ\varphi(\bx_i)+\int_{C_n}\psi_s\circ\varphi(\bx)\dx = -\frac{1}{n}\sum_{i=1}^n\log\psi_s\circ\varphi(\bx_i) + c \\
  &> -\frac{1}{n}\sum_{i=1}^n\log\psi_s\circ\varphi(\bx_i) + \log c + 1 = -\frac{1}{n}\sum_{i=1}^n\log\psi_s\circ\tilde\varphi(\bx_i)+\int_{C_n}\psi_s\circ\tilde\varphi(\bx)\dx,
\end{align*}
so $\int_{C_n}\psi_s\circ\varphi^*(\bx)\dx=1$, which establishes our claim.

\medskip

\noindent {\bf Step 3:} 
Letting $\bm\phi^* = (\phi_1^*,\ldots,\phi_n^*)$ denote an optimal solution to~\eqref{eqn:MLE-psi-phi}, we claim that $\cef[\bphi^*](\bx_i)=\phi^*_i$ holds for all $i\in[n]$. Indeed, for any $\bphi = (\phi_1,\ldots,\phi_n) \in \mathcal{D}_s^n$, if there exists $i^* \in [n]$ such that $\cef[\bphi](\bx_{i^*})\neq \phi_{i^*}$, then we can define $\tilde{\bphi} = (\tilde{\phi}_1,\ldots,\tilde{\phi}_n) \in \mathcal{D}_s^n$ such that $\tilde\phi_i=\cef[\bphi](\bx_i)$ for all $i\in[n]$.  We now claim that $\cef[\tilde\bphi]=\cef[\bphi]$. On the one hand, by~\eqref{Eq:cef}, $\tilde\phi_i=\cef[\bphi](\bx_i)\leq \phi_i$ for any $i\in[n]$. From the LP expression \eqref{eqn:cef-LP}, $\cef[\tilde\bphi] \leq \cef[\bphi]$. On the other hand, since $\cef[\bphi](\cdot)$ is a convex function with $\cef[\bphi](\bx_i)=\tilde\phi_i\leq \tilde\phi_i$ for any $i$, we have $\cef[\tilde\bphi](\cdot)\geq \cef[\bphi](\cdot)$.  It follows that $\cef[\tilde\bphi]=\cef[\bphi]$, and $\tilde\bphi$ with a smaller objective than $\bphi$. This establishes our claim, and shows that \eqref{eqn:MLE-psi-phi} is equivalent to \eqref{eqn:MLE-psi-step2} in the sense that 
$\hat{p}_n$ and $\bphi^*$ satisfy $\hat{p}_n=\psi_s(\cef[\bphi^*])$.

\medskip

\noindent {\bf Step 4:} 
When $s=0$, the function $\bphi \mapsto \frac1n\bm1^\top\bphi$ is convex on $\R^n$; when $s>0$, the function $\bphi \mapsto -\frac{1}{ns}\sum_{i=1}^n\log(-\phi_i)$ is convex on $(-\infty,0]^n$.  Moreover, when $s\leq 1$, the function $\psi_s$ is decreasing and convex, and since $\bphi \mapsto \cef[\bphi](\bx)$ is concave for every $\bx \in \R^n$, the result follows.
\end{proof}



\section{Background on shape-constrained inference}\label{app:background}


Entry points to the field of nonparametric inference under shape constraints include the book by \cite{groeneboom2014nonparametric}, as well as the 2018 special issue of the journal {\emph{Statistical Science}} \citep{samworth2018special}.  Other canonical problems in shape constraints that involve non-trivial computational issues include isotonic regression \citep{brunk1972statistical,zhang2002risk,chatterjee2015risk,durot2018limit,bellec2018sharp,yang2019contraction,han2019isotonic,pananjady2020isotonic} and convex regression \citep{hildreth1954point,seijo2011nonparametric,cai2015framework,guntuboyina2015global,han2016multivariate,fang2019risk,chen2020multivariate}, or combinations and variants of these \citep{chen2016generalized}.

Beyond papers already discussed, early theoretical work on log-concave density estimation includes \cite{pal2007estimating}, \cite{dumbgen2009maximum}, \cite{walther2009inference}, \cite{cule2010theoretical}, \cite{dumbgen2011approximation}, \cite{schuhmacher2011multivariate}, \cite{samworth2012independent} and \cite{chen2013smoothed}.  Sometimes, the class $\mathcal{P}_d$ is considered as a special case of the class of $s$-concave densities \citep{koenker2010quasi,seregin2010nonparametric,han2016approximation,doss2016global,han2019global}; see also Section~\ref{Sec:sconcave}.  Much recent work has focused on rates of convergence, which are best understood in the Hellinger distance $d_{\mathrm{H}}$, given by
\[
  d_{\mathrm{H}}^2(p,q) := \int_{\R^d} (p^{1/2}-q^{1/2})^2.
\]
For the case of correct model specification, i.e.~where $\hat{p}_n$ is computed from an independent and identically distributed sample of size $n$ from $p_0 \in \mathcal{P}_d$, it is now known \citep{kim2016global,kur2019log} that
\[
  \sup_{p_0 \in \mathcal{P}_d} \mathbb{E} d_{\mathrm{H}}^2(\hat{p}_n,p_0) \leq K_d \cdot \left\{ \begin{array}{ll} n^{-4/5} & \mbox{when $d=1$} \\
n^{-2/(d+1)}\log n & \mbox{when $d \geq 2$,} \end{array} \right.
\]
where $K_d > 0$ depends only on $d$, and that this risk bound is minimax optimal (up to the logarithmic factor when $d \geq 2$).  See also \cite{carpenter2018near} for an earlier result in the case $d \geq 4$, and \cite{xu2019high} for an alternative approach to high-dimensional log-concave density estimation that seeks to evade the curse of dimensionality in the additional presence of symmetry constraints.  It is further known that when $d \leq 3$, the log-concave maximum likelihood estimator can adapt to certain subclasses of log-concave densities, including log-concave densities whose logarithms are piecewise affine \citep{kim2018adaptation,feng2018adaptation}.  See also \cite{barber2020local} for recent work on extensions to the misspecified setting (where the true distribution from which the data are drawn does not have a log-concave density).

\end{appendix}

\section{Acknowledgements}
This research was supported in part by grants from the Office of Naval Research: ONR-N000141812298, N000142112841, the National Science Foundation: NSF-IIS-1718258 and IBM to Rahul Mazumder. The research of Richard J. Samworth was supported by EPSRC grants
EP/N031938/1 and EP/P031447/1, as well as ERC Advanced Grant 101019498.

The authors acknowledge MIT SuperCloud and Lincoln Laboratory Supercomputing Center for providing HPC resources that have contributed to the research
results reported within this paper.  We also thank the anonymous reviewers for constructive comments on an earlier version that helped to improve this paper.



\end{document}